\pdfoutput=1
\documentclass[format=acmsmall, review=false, screen=true, final]{acmart}
\usepackage{booktabs} 

\setcopyright{none}
\acmJournal{TKDD}
\acmYear{2020} 
\acmVolume{15} \acmNumber{1} \acmArticle{9} \acmMonth{10} \acmPrice{15.00}
\acmDOI{10.1145/3418773}

\usepackage{bold-extra}

\usepackage{array}
\newcolumntype{P}[1]{>{\centering\arraybackslash}p{#1}}
\newcolumntype{M}[1]{>{\centering\arraybackslash}m{#1}}
\newcolumntype{R}[1]{>{\arraybackslash}m{#1}}

\usepackage{comment}
\usepackage{paralist}
\usepackage{nicefrac}
\definecolor{thelightblue}{RGB}{0,191,255}

\usepackage{nicefrac}
\newcommand{\mychoose}[2]{\left( \begin{smallmatrix} #1 \\ #2 \end{smallmatrix} \right)} 

\usepackage{mathtools}
\usepackage{blkarray, bigstrut} 

\usepackage{tikz}
\usepackage{verbatim}
\usetikzlibrary{arrows}
\usetikzlibrary{shapes,snakes}
\usetikzlibrary{decorations.pathmorphing} 
\usetikzlibrary{fit}					
\usetikzlibrary{backgrounds}	

\makeatletter
\global\let\tikz@ensure@dollar@catcode=\relax
\makeatother

\usepackage{subfigure}
\usepackage{xcolor}

\definecolor{thelightblue}{RGB}{0,191,255}
\definecolor{theblue}{RGB}{0,0,180}

\usepackage{blkarray}
\usepackage{bigstrut, booktabs}

\makeatletter
\renewcommand*\env@matrix[1][*\c@MaxMatrixCols c]{
\hskip -\arraycolsep
\let\@ifnextchar\new@ifnextchar
\array{#1}}
\makeatother

\usepackage{booktabs} 
\usepackage{enumitem}

\usepackage{subfigure}
\usepackage{rotating}

\usepackage{multirow}

\usepackage{color}
\usepackage{xcolor}
\definecolor{mydarkblue}{RGB}{0, 20, 159} 
\definecolor{mydarkblue}{rgb}{0,0.08,0.45}

\usepackage{tikz}
\usepackage{verbatim}
\usetikzlibrary{arrows}
\usetikzlibrary{shapes,snakes}
\usetikzlibrary{decorations.pathmorphing} 
\usetikzlibrary{fit}					
\usetikzlibrary{backgrounds}	

\definecolor{gray}{RGB}{150,150,150}
\definecolor{theblue}{RGB}{0, 20, 159} 

\definecolor{myyellow}{RGB}{255,255,204}
\definecolor{myred}{RGB}{255,204,204}
\definecolor{myblue}{RGB}{0,200,255}
\definecolor{mygreen}{RGB}{80,220,80}

\newcommand{\eg}{\emph{e.g.}}
\newcommand{\ie}{\emph{i.e.}}

\newcommand{\wrt}{\emph{w.r.t.}\ }

\newtheorem{thm}{Theorem}
\newtheorem{cor}{Corollary}

\newtheorem{Property}{Property}
\newtheorem{Proposition}{Proposition}
\newtheorem{lemma}{Lemma}
\newtheorem{Definition}{Definition}
\newtheorem{Problem}{Problem}

\usepackage{balance}

\usepackage{epsfig}

\usepackage{graphicx}

\usepackage{amsmath}
\usepackage{amsfonts}

\usepackage{tabularx}
\usepackage{booktabs}
\usepackage{relsize}

\usepackage{numprint} 

\usepackage{setspace}
\graphicspath{{./}{./graphics/}}
\newcolumntype{H}{>{\setbox0=\hbox\bgroup}c<{\egroup}@{}}

\newcommand{\abs}[1]{\left|#1\right|}

\usepackage{algorithm}
\usepackage{algorithmicx}
\usepackage{algpseudocode}
\algblockdefx[parallel]{ParFor}{EndPar}[1][]{$\textbf{parallel for}$ #1 $\textbf{do}$}{$\textbf{end parallel}$}
\algrenewcommand{\alglinenumber}[1]{\fontsize{6.5}{7}\selectfont#1}
\algtext*{EndPar}

\algblockdefx[parallel]{parfor}{endpar}[1][]{$\textbf{parallel for}$ #1 $\textbf{do}$}{$\textbf{end parallel}$}
\algrenewcommand{\alglinenumber}[1]{\scriptsize#1:}

\algblockdefx[parallel]{parallelfor}{parallelend}
[1][]{\textbf{parallel for} #1}
{\textbf{end parallel}}

\usepackage{nicefrac}

\algnotext{EndFor}
\algnotext{EndIf}
\algnotext{EndProcedure}

\newcommand{\setAlgFontSize}{\fontsize{9pt}{10pt}\selectfont} 

\newcommand{\multilinenospace}[1]{\State \parbox[t]{\dimexpr\linewidth-\algorithmicindent}{\begin{spacing}{1.1}\setAlgFontSize#1\strut \end{spacing}}}
\newcommand{\multilinenospaceD}[1]{\State \parbox[t]{\dimexpr0.96 \linewidth-\algorithmicindent}{\begin{spacing}{1.1}\setAlgFontSize#1\strut \end{spacing}}}

\newcommand\TTT{\rule{0pt}{3.2ex}}
\newcommand\BBB{\rule[-1.4ex]{0pt}{0pt}}

\usepackage[notheorems]{rr-math}
\newcommand{\hash}{\ensuremath{c}} 

\usepackage{braket}

\DeclareFixedFont{\tensy}{OMS}{cmsy}{m}{n}{10pt}
\definecolor{blue}{RGB}{0,0,0}
\definecolor{magenta}{RGB}{0,0,0} 

\definecolor{originalBlue}{RGB}{0,0,255}

\begin{document}
\title{Heterogeneous Graphlets} 

\author{Ryan A. Rossi}
\orcid{1234-5678-9012-3456}
\affiliation{
\institution{Adobe Research}
\streetaddress{345 Park Ave}
\city{San Jose}
\state{CA}
\country{USA}
}
\email{rrossi@adobe.com}
\author{Nesreen K. Ahmed}
\affiliation{
\institution{Intel Labs}
\streetaddress{3065 Bowers Avenue}
\city{Santa Clara}
\state{CA}
\country{USA}
}
\email{nesreen.k.ahmed@intel.com}
\author{Aldo Carranza}
\affiliation{
\institution{Stanford University}
\streetaddress{Huang Building 475 Via Ortega}
\city{Stanford}
\state{CA}
\country{USA}
}
\email{aldogael@stanford.edu}
\author{David Arbour}
\affiliation{
\institution{Adobe Research}
\streetaddress{345 Park Ave}
\city{San Jose}
\state{CA}
\country{USA}
}
\email{arbour@adobe.com}
\author{Anup Rao}
\affiliation{
\institution{Adobe Research}
\streetaddress{345 Park Ave}
\city{San Jose}
\state{CA}
\country{USA}
}
\email{anuprao@adobe.com}
\author{Sungchul Kim}
\affiliation{
\institution{Adobe Research}
\streetaddress{345 Park Ave}
\city{San Jose}
\state{CA}
\country{USA}
}
\email{sukim@adobe.com}
\author{Eunyee Koh}
\affiliation{
\institution{Adobe Research}
\streetaddress{345 Park Ave}
\city{San Jose}
\state{CA}
\country{USA}
}
\email{eunyee@adobe.com}

\renewcommand{\shortauthors}{R.~A.~Rossi et al.}

\begin{abstract}
In this paper, we introduce a generalization of graphlets to heterogeneous networks called \emph{typed graphlets}.
Informally, typed graphlets are small typed induced subgraphs.
Typed graphlets generalize graphlets to rich heterogeneous networks as they explicitly capture the higher-order typed connectivity patterns in such networks. 
To address this problem, we describe a general framework for counting the occurrences of such typed graphlets.
The proposed algorithms leverage a number of combinatorial relationships for different typed graphlets.
For each edge, we count a few typed graphlets, and with these counts along with the combinatorial relationships, we obtain the exact counts of the other typed graphlets in $o(1)$ constant time.
Notably, the worst-case time complexity of the proposed approach matches the
time complexity of the best known untyped algorithm.
In addition, the approach lends itself to an efficient lock-free and asynchronous parallel implementation.
While there are no existing methods for typed graphlets, there has been some work that focused on computing a different and much simpler notion called colored graphlet.
The experiments confirm that our proposed approach is orders of magnitude faster \emph{and} more space-efficient than 
methods for computing the simpler notion of colored graphlet.
Unlike these methods that take hours on small networks, the proposed approach takes only seconds on large networks with millions of edges.
Notably, since typed graphlet is more general than colored graphlet (and untyped graphlets), the counts of various typed graphlets can be combined to obtain the counts of the much simpler notion of colored graphlets.
The proposed methods give rise to new opportunities and applications for typed graphlets.
\end{abstract}



\maketitle

\section{Introduction} \label{sec:intro}
\noindent
Higher-order connectivity patterns such as small induced subgraphs called graphlets\footnote{The terms graphlet and induced subgraph are used interchangeably.} are known to be the fundamental building blocks of homogeneous networks~\cite{Milo2002} and are essential for modeling and understanding the fundamental components of these networks~\cite{pgd,pgd-kais,benson2016higher}.
Furthermore, graphlets are also important for many predictive and descriptive modeling application tasks~\cite{zhang-image-categorization-via-graphlets,vishwanathan2010graph,shervashidze2009efficient,Milo2002,prvzulj2004modeling,milenkovic2008uncovering,hayes2013graphlet,ahmed17streams,lichtenwalter2012vertex} such as 
image processing and computer vision~\cite{zhang-image-categorization-via-graphlets,zhang2013probabilistic}, 
network alignment~\cite{koyuturk2006pairwise,prvzulj2007biological,milenkovic2008uncovering,crawford2015great}, 
classification~\cite{vishwanathan2010graph,shervashidze2009efficient}, 
visualization and sensemaking~\cite{pgd,pgd-kais}, 
dynamic network analysis~\cite{kovanen2011temporal,hulovatyy2015exploring}, 
community detection~\cite{radicchi2004defining,palla2005uncovering,solava2012graphlet,benson2016higher}, 
role discovery~\cite{ahmed17aaai,role2vec}, 
anomaly detection~\cite{noble2003graph,akoglu2015graph}, 
and link prediction~\cite{Rossi2018a}. 

However, such (untyped) graphlets are \emph{unable} to capture the rich typed connectivity patterns in more complex networks such as those that are heterogeneous, which includes  bipartite, k-partite, k-star, and attributed graphs as special cases, among others.
In heterogeneous networks, nodes and edges can be of different types and explicitly modeling such types is crucial~\cite{banerjee2007multi,acar2011all,gu2018heterAlignment,higher-order-clustering-heter}.
Such heterogeneous networks arise ubiquitously in the natural world where nodes and edges of multiple types are observed, \eg, 
between humans~\cite{kong2013inferring-heterSoc}, 
neurons~\cite{bullmore2009complex,bassett2006small}, 
routers and autonomous systems (ASes)~\cite{rossi2013topology}, 
web pages~\cite{yin2009exploring}, 
devices \& sensors~\cite{eagle2006reality}, 
infrastructure (roads, airports, power stations)~\cite{powergrid}, 
vehicles (cars, satelites, UAVs)~\cite{hung2008mobility-heter}, 
and information in general~\cite{yu2014personalized,sun2011pathsim,rossi16collective-factor}.

\definecolor{typeOneColor}{RGB}{8,81,156} 
\definecolor{typeTwoColor}{RGB}{222,45,38} 
\definecolor{typeThreeColor}{RGB}{49,163,84} 
\definecolor{typeFourColor}{RGB}{117,107,177} 

\makeatletter
\global\let\tikz@ensure@dollar@catcode=\relax
\makeatother
\tikzstyle{every node}=[font=\large,line width=1.5pt]
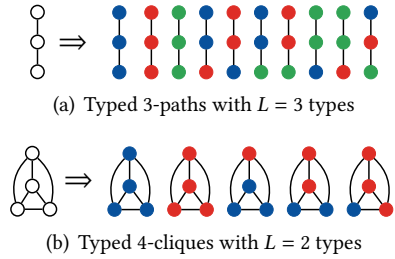
\begin{figure}[h!]

\subfigure[\vspace{-2mm}Typed 3-paths with $L=3$ types]{
\scalebox{0.70}{
\scalebox{0.28}{

\begin{tikzpicture}[-,>=latex,auto,node distance=2.0cm,thick,
main node/.style={circle,draw=black,fill=white,draw,font=\sffamily\Huge\bfseries,text=black,minimum width=0.9cm, line width=1mm},
]

\node[main node] (1) {}; 
\node[main node] (2) [below of=1] {}; 
\node[main node] (3) [below of=2] {}; 

\tikzstyle{LabelStyle}=[below=3pt]
\path[every node/.style={font=\sffamily}] 
(1) edge [line width=1.0mm, left] node [above left] {} (2) 
(2) edge [line width=1.0mm, left] node[below left] {} (3);
\end{tikzpicture}
}
\scalebox{0.28}{

\begin{tikzpicture}[-,>=latex,auto,node distance=2.0cm,thick,
main node/.style={circle,draw=black,fill=white,draw,font=\sffamily\Huge\bfseries,text=black,minimum width=0.9cm},
white node/.style={draw=white,draw,font=\sffamily\Huge\bfseries,text=black,minimum width=0.9cm}
]
\node[white node] (1) {};
\node[white node] (2) [below of=1] {};
\node[white node] (3) [below of=2] {};
\node[white node] (4) [right of=2, left=30pt, below=10pt, above=0.05pt] {\vspace{2mm}\fontsize{56}{56}\selectfont $\Rightarrow$};
\end{tikzpicture}
}
\hspace{2mm}
\scalebox{0.28}{

\begin{tikzpicture}[-,>=latex,auto,node distance=2.0cm,thick,
typeOne node/.style={circle,draw=typeOneColor,fill=typeOneColor,draw,font=\sffamily\Huge\bfseries,text=white,minimum width=0.9cm},
typeTwo node/.style={circle,draw=typeTwoColor,fill=typeTwoColor,draw,font=\sffamily\Huge\bfseries,text=white,minimum width=0.9cm},
typeThree node/.style={circle,draw=typeThreeColor,fill=typeThreeColor,draw,font=\sffamily\Huge\bfseries,text=white,minimum width=0.9cm},
]

\node[typeOne node] (1) {}; 
\node[typeOne node] (2) [below of=1] {}; 
\node[typeOne node] (3) [below of=2] {}; 

\tikzstyle{LabelStyle}=[below=3pt]
\path[every node/.style={font=\sffamily}] 
(1) edge [line width=1.0mm, left] node [above left] {} (2) 
(2) edge [line width=1.0mm, left] node[below left] {} (3);
\end{tikzpicture}
}
\label{fig:typed-3-path-homo-typeOne-3colors}
\hspace{0.3mm}
\scalebox{0.28}{
\begin{tikzpicture}[-,>=latex,auto,node distance=2.0cm,thick,
typeOne node/.style={circle,draw=typeOneColor,fill=typeOneColor,draw,font=\sffamily\Huge\bfseries,text=white,minimum width=0.9cm},
typeTwo node/.style={circle,draw=typeTwoColor,fill=typeTwoColor,draw,font=\sffamily\Huge\bfseries,text=white,minimum width=0.9cm},
typeThree node/.style={circle,draw=typeThreeColor,fill=typeThreeColor,draw,font=\sffamily\Huge\bfseries,text=white,minimum width=0.9cm},
]

\node[typeTwo node] (1) {}; 
\node[typeTwo node] (2) [below of=1] {}; 
\node[typeTwo node] (3) [below of=2] {}; 

\tikzstyle{LabelStyle}=[below=3pt]
\path[every node/.style={font=\sffamily}] 
(1) edge [line width=1.0mm, left] node [above left] {} (2) 
(2) edge [line width=1.0mm, left] node[below left] {} (3);
\end{tikzpicture}
}
\hspace{0.3mm}
\scalebox{0.28}{
\begin{tikzpicture}[-,>=latex,auto,node distance=2.0cm,thick,
typeOne node/.style={circle,draw=typeOneColor,fill=typeOneColor,draw,font=\sffamily\Huge\bfseries,text=white,minimum width=0.9cm},
typeTwo node/.style={circle,draw=typeTwoColor,fill=typeTwoColor,draw,font=\sffamily\Huge\bfseries,text=white,minimum width=0.9cm},
typeThree node/.style={circle,draw=typeThreeColor,fill=typeThreeColor,draw,font=\sffamily\Huge\bfseries,text=white,minimum width=0.9cm},
]

\node[typeThree node] (1) {}; 
\node[typeThree node] (2) [below of=1] {}; 
\node[typeThree node] (3) [below of=2] {}; 

\tikzstyle{LabelStyle}=[below=3pt]
\path[every node/.style={font=\sffamily}] 
(1) edge [line width=1.0mm, left] node [above left] {} (2) 
(2) edge [line width=1.0mm, left] node[below left] {} (3);
\end{tikzpicture}
}
\hspace{0.3mm}
\scalebox{0.28}{

\begin{tikzpicture}[-,>=latex,auto,node distance=2.0cm,thick,
typeOne node/.style={circle,draw=typeOneColor,fill=typeOneColor,draw,font=\sffamily\Huge\bfseries,text=white,minimum width=0.9cm},
typeTwo node/.style={circle,draw=typeTwoColor,fill=typeTwoColor,draw,font=\sffamily\Huge\bfseries,text=white,minimum width=0.9cm},
typeThree node/.style={circle,draw=typeThreeColor,fill=typeThreeColor,draw,font=\sffamily\Huge\bfseries,text=white,minimum width=0.9cm},
]

\node[typeOne node] (1) {}; 
\node[typeOne node] (2) [below of=1] {}; 
\node[typeTwo node] (3) [below of=2] {}; 

\tikzstyle{LabelStyle}=[below=3pt]
\path[every node/.style={font=\sffamily}] 
(1) edge [line width=1.0mm, left] node [above left] {} (2) 
(2) edge [line width=1.0mm, left] node[below left] {} (3);
\end{tikzpicture}
}
\hspace{0.3mm}
\scalebox{0.28}{

\begin{tikzpicture}[-,>=latex,auto,node distance=2.0cm,thick,
typeOne node/.style={circle,draw=typeOneColor,fill=typeOneColor,draw,font=\sffamily\Huge\bfseries,text=white,minimum width=0.9cm},
typeTwo node/.style={circle,draw=typeTwoColor,fill=typeTwoColor,draw,font=\sffamily\Huge\bfseries,text=white,minimum width=0.9cm},
typeThree node/.style={circle,draw=typeThreeColor,fill=typeThreeColor,draw,font=\sffamily\Huge\bfseries,text=white,minimum width=0.9cm},
]

\node[typeTwo node] (1) {}; 
\node[typeTwo node] (2) [below of=1] {}; 
\node[typeOne node] (3)  [below of=2] {}; 

\tikzstyle{LabelStyle}=[below=3pt]
\path[every node/.style={font=\sffamily}] 
(1) edge [line width=1.0mm, left] node [above left] {} (2) 
(2) edge [line width=1.0mm, left] node[below left] {} (3);
\end{tikzpicture}
}
\hspace{0.3mm}
\scalebox{0.28}{

\begin{tikzpicture}[-,>=latex,auto,node distance=2.0cm,thick,
typeOne node/.style={circle,draw=typeOneColor,fill=typeOneColor,draw,font=\sffamily\Huge\bfseries,text=white,minimum width=0.9cm},
typeTwo node/.style={circle,draw=typeTwoColor,fill=typeTwoColor,draw,font=\sffamily\Huge\bfseries,text=white,minimum width=0.9cm},
typeThree node/.style={circle,draw=typeThreeColor,fill=typeThreeColor,draw,font=\sffamily\Huge\bfseries,text=white,minimum width=0.9cm},
]

\node[typeOne node] (1) {}; 
\node[typeOne node] (2) [below of=1] {}; 
\node[typeThree node] (3)  [below of=2] {}; 

\tikzstyle{LabelStyle}=[below=3pt]
\path[every node/.style={font=\sffamily}] 
(1) edge [line width=1.0mm, left] node [above left] {} (2) 
(2) edge [line width=1.0mm, left] node[below left] {} (3);
\end{tikzpicture}
}
\hspace{0.3mm}
\scalebox{0.28}{

\begin{tikzpicture}[-,>=latex,auto,node distance=2.0cm,thick,
typeOne node/.style={circle,draw=typeOneColor,fill=typeOneColor,draw,font=\sffamily\Huge\bfseries,text=white,minimum width=0.9cm},
typeTwo node/.style={circle,draw=typeTwoColor,fill=typeTwoColor,draw,font=\sffamily\Huge\bfseries,text=white,minimum width=0.9cm},
typeThree node/.style={circle,draw=typeThreeColor,fill=typeThreeColor,draw,font=\sffamily\Huge\bfseries,text=white,minimum width=0.9cm},
]

\node[typeTwo node] (1) {}; 
\node[typeTwo node] (2) [below of=1] {}; 
\node[typeThree node] (3)  [below of=2] {}; 

\tikzstyle{LabelStyle}=[below=3pt]
\path[every node/.style={font=\sffamily}] 
(1) edge [line width=1.0mm, left] node [above left] {} (2) 
(2) edge [line width=1.0mm, left] node[below left] {} (3);
\end{tikzpicture}
}
\hspace{0.3mm}
\scalebox{0.28}{

\begin{tikzpicture}[-,>=latex,auto,node distance=2.0cm,thick,
typeOne node/.style={circle,draw=typeOneColor,fill=typeOneColor,draw,font=\sffamily\Huge\bfseries,text=white,minimum width=0.9cm},
typeTwo node/.style={circle,draw=typeTwoColor,fill=typeTwoColor,draw,font=\sffamily\Huge\bfseries,text=white,minimum width=0.9cm},
typeThree node/.style={circle,draw=typeThreeColor,fill=typeThreeColor,draw,font=\sffamily\Huge\bfseries,text=white,minimum width=0.9cm},
]

\node[typeThree node] (1) {}; 
\node[typeThree node] (2) [below of=1] {}; 
\node[typeOne node] (3)  [below of=2] {}; 

\tikzstyle{LabelStyle}=[below=3pt]
\path[every node/.style={font=\sffamily}] 
(1) edge [line width=1.0mm, left] node [above left] {} (2) 
(2) edge [line width=1.0mm, left] node[below left] {} (3);
\end{tikzpicture}
}
\hspace{0.3mm}
\scalebox{0.28}{

\begin{tikzpicture}[-,>=latex,auto,node distance=2.0cm,thick,
typeOne node/.style={circle,draw=typeOneColor,fill=typeOneColor,draw,font=\sffamily\Huge\bfseries,text=white,minimum width=0.9cm},
typeTwo node/.style={circle,draw=typeTwoColor,fill=typeTwoColor,draw,font=\sffamily\Huge\bfseries,text=white,minimum width=0.9cm},
typeThree node/.style={circle,draw=typeThreeColor,fill=typeThreeColor,draw,font=\sffamily\Huge\bfseries,text=white,minimum width=0.9cm},
]

\node[typeThree node] (1) {}; 
\node[typeThree node] (2) [below of=1] {}; 
\node[typeTwo node] (3)  [below of=2] {}; 

\tikzstyle{LabelStyle}=[below=3pt]
\path[every node/.style={font=\sffamily}] 
(1) edge [line width=1.0mm, left] node [above left] {} (2) 
(2) edge [line width=1.0mm, left] node[below left] {} (3);
\end{tikzpicture}
}
\hspace{0.3mm}
\scalebox{0.28}{

\begin{tikzpicture}[-,>=latex,auto,node distance=2.0cm,thick,
typeOne node/.style={circle,draw=typeOneColor,fill=typeOneColor,draw,font=\sffamily\Huge\bfseries,text=white,minimum width=0.9cm},
typeTwo node/.style={circle,draw=typeTwoColor,fill=typeTwoColor,draw,font=\sffamily\Huge\bfseries,text=white,minimum width=0.9cm},
typeThree node/.style={circle,draw=typeThreeColor,fill=typeThreeColor,draw,font=\sffamily\Huge\bfseries,text=white,minimum width=0.9cm},
]

\node[typeOne node] (1) {}; 
\node[typeTwo node] (2) [below of=1] {}; 
\node[typeThree node] (3)  [below of=2] {}; 

\tikzstyle{LabelStyle}=[below=3pt]
\path[every node/.style={font=\sffamily}] 
(1) edge [line width=1.0mm, left] node [above left] {} (2) 
(2) edge [line width=1.0mm, left] node[below left] {} (3);
\end{tikzpicture}
}
\label{fig:typed-3-path-with-3-types}
}
}

\subfigure[Typed 4-cliques with $L=2$ types]{
\scalebox{0.70}{
\scalebox{0.28}{
\begin{tikzpicture}[-,>=latex,auto,node distance=2.2cm,thick,
main node/.style={circle,draw=black,fill=white,draw,font=\sffamily\Huge\bfseries,text=black,minimum width=0.9cm, line width=1mm},
]

\node[main node] (1) {}; 
\node[main node] (2) [right of=1] {}; 
\node[main node] (3) [above right of=1, left=1.5pt] {}; 
\node[main node] (4) [above of=3] {}; 

\tikzstyle{LabelStyle}=[below=3pt]
\path[every node/.style={font=\sffamily}] 
	(1) edge [bend left,line width=1.0mm] node[above right] {} (4)
	(2) edge [bend right,line width=1.0mm] node[above right] {} (4)
(1) edge [line width=1.0mm, left] node [above left] {} (2) 
	(1)  edge [line width=1.0mm, right] node[below right] {} (3)
(2) edge [line width=1.0mm, left] node[below left] {} (3)
	(3) edge [line width=1.0mm,right] node[above right] {} (4);
\end{tikzpicture}
}
\hspace{-1.4mm}
\scalebox{0.28}{
\begin{tikzpicture}[-,>=latex,auto,node distance=2.0cm,thick,
main node/.style={circle,draw=black,fill=white,draw,font=\sffamily\Huge\bfseries,text=black,minimum width=0.9cm},
white node/.style={draw=white,draw,font=\sffamily\Huge\bfseries,text=black,minimum width=0.9cm}
]
\node[white node] (1) {};
\node[white node] (2) [below of=1,left=1.5pt] {};
\node[white node] (3) [below of=2] {};
\node[white node] (4) [right of=2, left=60pt, below=10pt, above=-0.05pt] {\vspace{-2mm}\fontsize{56}{56}\selectfont $\Rightarrow$};
\end{tikzpicture}
}
\hspace{0mm}
\scalebox{0.28}{
\begin{tikzpicture}[-,>=latex,auto,node distance=2.2cm,thick,
main node/.style={circle,draw=black,fill=black,draw,font=\sffamily\Huge\bfseries,text=white,minimum width=0.9cm},
white/.style={circle,draw=white,fill=white,draw,font=\sffamily,text=white,minimum width=0.001cm},
typeOne node/.style={circle,draw=typeOneColor,fill=typeOneColor,draw,font=\sffamily\Huge\bfseries,text=white,minimum width=0.9cm},
typeTwo node/.style={circle,draw=typeTwoColor,fill=typeTwoColor,draw,font=\sffamily\Huge\bfseries,text=white,minimum width=0.9cm},
]

\node[typeOne node] (1) {}; 
\node[typeOne node] (2) [right of=1] {}; 
\node[typeOne node] (3) [above right of=1, left=1.5pt] {}; 
\node[typeOne node] (4) [above of=3] {}; 

\tikzstyle{LabelStyle}=[below=3pt]
\path[every node/.style={font=\sffamily}] 
	(1) edge [bend left,line width=1.0mm] node[above right] {} (4)
	(2) edge [bend right,line width=1.0mm] node[above right] {} (4)
(1) edge [line width=1.0mm, left] node [above left] {} (2) 
	(1)  edge [line width=1.0mm, right] node[below right] {} (3)
(2) edge [line width=1.0mm, left] node[below left] {} (3)
	(3) edge [line width=1.0mm,right] node[above right] {} (4);
\end{tikzpicture}
}
\hspace{0.3mm}
\scalebox{0.28}{
\begin{tikzpicture}[-,>=latex,auto,node distance=2.2cm,thick,
main node/.style={circle,draw=black,fill=black,draw,font=\sffamily\Huge\bfseries,text=white,minimum width=0.9cm},
white/.style={circle,draw=white,fill=white,draw,font=\sffamily,text=white,minimum width=0.001cm},
typeOne node/.style={circle,draw=typeOneColor,fill=typeOneColor,draw,font=\sffamily\Huge\bfseries,text=white,minimum width=0.9cm},
typeTwo node/.style={circle,draw=typeTwoColor,fill=typeTwoColor,draw,font=\sffamily\Huge\bfseries,text=white,minimum width=0.9cm},
]

\node[typeTwo node] (1) {}; 
\node[typeTwo node] (2) [right of=1] {}; 
\node[typeTwo node] (3) [above right of=1, left=1.5pt] {}; 
\node[typeTwo node] (4) [above of=3] {}; 

\tikzstyle{LabelStyle}=[below=3pt]
\path[every node/.style={font=\sffamily}] 
	(1) edge [bend left,line width=1.0mm] node[above right] {} (4)
	(2) edge [bend right,line width=1.0mm] node[above right] {} (4)
(1) edge [line width=1.0mm, left] node [above left] {} (2) 
	(1)  edge [line width=1.0mm, right] node[below right] {} (3)
(2) edge [line width=1.0mm, left] node[below left] {} (3)
	(3) edge [line width=1.0mm,right] node[above right] {} (4);
\end{tikzpicture}
}
\hspace{0.3mm}
\scalebox{0.28}{
\begin{tikzpicture}[-,>=latex,auto,node distance=2.2cm,thick,
main node/.style={circle,draw=black,fill=black,draw,font=\sffamily\Huge\bfseries,text=white,minimum width=0.9cm},
white/.style={circle,draw=white,fill=white,draw,font=\sffamily,text=white,minimum width=0.001cm},
typeOne node/.style={circle,draw=typeOneColor,fill=typeOneColor,draw,font=\sffamily\Huge\bfseries,text=white,minimum width=0.9cm},
typeTwo node/.style={circle,draw=typeTwoColor,fill=typeTwoColor,draw,font=\sffamily\Huge\bfseries,text=white,minimum width=0.9cm},
]

\node[typeOne node] (1) {}; 
\node[typeOne node] (2) [right of=1] {}; 
\node[typeOne node] (3) [above right of=1, left=1.5pt] {}; 
\node[typeTwo node] (4) [above of=3] {}; 

\tikzstyle{LabelStyle}=[below=3pt]
\path[every node/.style={font=\sffamily}] 
	(1) edge [bend left,line width=1.0mm] node[above right] {} (4)
	(2) edge [bend right,line width=1.0mm] node[above right] {} (4)
(1) edge [line width=1.0mm, left] node [above left] {} (2) 
	(1)  edge [line width=1.0mm, right] node[below right] {} (3)
(2) edge [line width=1.0mm, left] node[below left] {} (3)
	(3) edge [line width=1.0mm,right] node[above right] {} (4);
\end{tikzpicture}
}
\hspace{0.3mm}
\scalebox{0.28}{
\begin{tikzpicture}[-,>=latex,auto,node distance=2.2cm,thick,
main node/.style={circle,draw=black,fill=black,draw,font=\sffamily\Huge\bfseries,text=white,minimum width=0.9cm},
white/.style={circle,draw=white,fill=white,draw,font=\sffamily,text=white,minimum width=0.001cm},
typeOne node/.style={circle,draw=typeOneColor,fill=typeOneColor,draw,font=\sffamily\Huge\bfseries,text=white,minimum width=0.9cm},
typeTwo node/.style={circle,draw=typeTwoColor,fill=typeTwoColor,draw,font=\sffamily\Huge\bfseries,text=white,minimum width=0.9cm},
]

\node[typeOne node] (1) {}; 
\node[typeOne node] (2) [right of=1] {}; 
\node[typeTwo node] (3) [above right of=1, left=1.5pt] {}; 
\node[typeTwo node] (4) [above of=3] {}; 

\tikzstyle{LabelStyle}=[below=3pt]
\path[every node/.style={font=\sffamily}] 
	(1) edge [bend left,line width=1.0mm] node[above right] {} (4)
	(2) edge [bend right,line width=1.0mm] node[above right] {} (4)
(1) edge [line width=1.0mm, left] node [above left] {} (2) 
	(1)  edge [line width=1.0mm, right] node[below right] {} (3)
(2) edge [line width=1.0mm, left] node[below left] {} (3)
	(3) edge [line width=1.0mm,right] node[above right] {} (4);
\end{tikzpicture}
}
\hspace{0.3mm}
\scalebox{0.28}{
\begin{tikzpicture}[-,>=latex,auto,node distance=2.2cm,thick,
main node/.style={circle,draw=black,fill=black,draw,font=\sffamily\Huge\bfseries,text=white,minimum width=0.9cm},
white/.style={circle,draw=white,fill=white,draw,font=\sffamily,text=white,minimum width=0.001cm},
typeOne node/.style={circle,draw=typeOneColor,fill=typeOneColor,draw,font=\sffamily\Huge\bfseries,text=white,minimum width=0.9cm},
typeTwo node/.style={circle,draw=typeTwoColor,fill=typeTwoColor,draw,font=\sffamily\Huge\bfseries,text=white,minimum width=0.9cm},
]

\node[typeOne node] (1) {}; 
\node[typeTwo node] (2) [right of=1] {}; 
\node[typeTwo node] (3) [above right of=1, left=1.5pt] {}; 
\node[typeTwo node] (4) [above of=3] {}; 

\tikzstyle{LabelStyle}=[below=3pt]
\path[every node/.style={font=\sffamily}] 
	(1) edge [bend left,line width=1.0mm] node[above right] {} (4)
	(2) edge [bend right,line width=1.0mm] node[above right] {} (4)
(1) edge [line width=1.0mm, left] node [above left] {} (2) 
	(1)  edge [line width=1.0mm, right] node[below right] {} (3)
(2) edge [line width=1.0mm, left] node[below left] {} (3)
	(3) edge [line width=1.0mm,right] node[above right] {} (4);
\end{tikzpicture}
}
\vspace{-15mm}
\label{fig:typed-4-cliques-with-2-types}
}
} 

\vspace{-2mm}
\caption{
Examples of typed (heterogeneous) graphlets
}
\label{fig:typed-motifs-3colors}
\end{figure}

In this work, we introduce the notion of a \emph{typed graphlet} that naturally generalizes the notion of graphlet to heterogeneous networks.\footnote{The terms heterogeneous and typed graphlet are used interchangeably.}
Typed graphlets generalize the notion of graphlets to rich heterogeneous networks as they capture both the induced subgraph of interest and the types associated with the nodes in the induced subgraph (Figure~\ref{fig:typed-motifs-3colors}).
These small induced typed subgraphs are the fundamental \emph{building blocks of rich heterogeneous networks}.
Typed graphlets naturally capture the higher-order typed connectivity patterns in bipartite, k-partite, signed, k-star, attributed graphs, and more generally heterogeneous networks.
As such, typed graphlets are useful for a wide variety of predictive and descriptive modeling applications in these rich complex networks.
Closest work related to our own has focused on colored graphlets~\cite{ribeiro2014discovering,gu2018heterAlignment}, which is a different problem.
See Figure~\ref{fig:typed-graphlet-vs-colored-graphlet-key-differences} for an intuitive illustration of the difference between the proposed notion of typed graphlets and recent work that focuses on colored graphlets.

Despite their fundamental and practical importance, counting typed graphlets 
remains a challenging and unsolved problem.
To address this problem, we propose a fast, parallel, and space-efficient framework for counting typed graphlets in large networks.
The time complexity is provably optimal and matches the time complexity of the best known untyped graphlet counting algorithm, \ie, PGD~\cite{pgd} and variants based on it~\cite{dave2017clog,PinarWWW17}.
Using non-trivial combinatorial relationships between lower-order ($k\!-\!1$)-node typed graphlets, 
we derive equations that allow us to compute many of the $k$-node typed graphlet counts in $o(1)$ constant time.
Thus, we avoid explicit enumeration of many typed graphlets by simply computing the exact count directly in constant time using the discovered combinatorial relationships.
For every edge, we count a few typed graphlets and obtain the exact counts of the remaining typed graphlets in $o(1)$ constant time.
Furthermore, we store only the nonzero typed graphlet counts for every edge.
To better handle large-scale heterogeneous networks with an arbitrary number of types, we propose an efficient parallel algorithm for typed graphlets
that scales almost linearly as the number of processing units increases.
As an aside, this paper focuses on counting typed graphlets with up to four nodes.
Typed graphlets of a larger size are outside the scope of this paper and left for future work.
However, the ideas and theoretical foundations formalized in this work naturally extend to typed graphlets of larger sizes (See Section~\ref{sec:framework-discussion} for further discussion).

Theoretically, we show that typed graphlets are more powerful and encode more information than untyped graphlets.
In addition, we theoretically show the worst-case time and space complexity of the proposed framework.
Notably, the time complexity of the proposed approach is shown to be equivalent to the best untyped graphlet counting algorithm.
Furthermore, we derive many of the typed graphlets directly in $o(1)$ constant time using counts of lower-order $(k\!-\!1)$-node typed graphlets.

\begin{table}[t!]
\caption{Summary of notation. Matrices are bold upright letters; vectors are bold lowercase letters.}
\vspace{-3mm}
\renewcommand{\arraystretch}{1.15} 
\scalebox{1.0}{
\centering 
\fontsize{8}{9.5}\selectfont
\setlength{\tabcolsep}{6pt} 
\label{table:notation}
\hspace*{-2.5mm}
\begin{tabularx}{0.9\linewidth}{@{}r X@{}} 
\toprule
$G$ & graph \\ 
$H,F$ & graphlet of $G$ \\

$N, M$ & number of nodes $N = |V|$ and edges $M = |E|$ in the graph \\
$K$ & size of a graphlet ($\#$ nodes) \\
$L$ & number of types (\ie, labels) \\

$\mathcal{H}$ & set of all untyped graphlets in $G$ \\
$\mathcal{H}_T$ & set of all typed graphlets in $G$ \\

$T$ & \# of typed graphlets $T = |\mathcal{H}_T|$ observed in $G$ with $L$ types \\ 
$T_{\max}$ & \# of possible typed graphlets with $L$ types, hence $T \leq T_{\max}$ \\ 
$T_H$ & \# of different typed graphlets for a particular graphlet $H \in \mathcal{H}$ \\

$\mathcal{T}_V$ & set of node types in $G$ \\
$\mathcal{T}_E$ & set of edge types in $G$ \\
$\phi$ & type function $\phi : V \rightarrow \mathcal{T}_V$ \\
$\xi$ & type function $\xi : E \rightarrow \mathcal{T}_E$ \\

$\vt$ & $K$-dimensional type vector $\vt = \big[\,\! \phi_{w_1} \;\! \cdots \; \phi_{w_K} \,\!\big]$ \\

$f_{ij}(H, \vt)$ & \# of instances of graphlet $H$ that contain nodes $i$ and $j$ with type vector $\vt$ \\

$\mathbb{F}$ & an arbitrary typed graphlet hash function (Section~\ref{sec:typed-motif-hash-function}) \\

$\Delta$ & maximum degree of a node in $G$ \\
$\Gamma_{i}^{t}$ & set of neighbors of node $i$ with type $t$ \\
$d_{i}^{t}$ & degree of node $i$ with type $t$, $d_{i}^{t} = |\Gamma_{i}^{t}|$ \\
$T_{ij}^{t}$ & set of nodes of type $t$ that form typed triangles with $i$ and $j$ \\
$S_{i}^{t}$, $S_{j}^{t}$ & set of nodes of type $t$ that form typed 3-node stars centered at $i$ (or $j$) \\

$\mathcal{M}_{ij}$ & set of typed graphlets for a given pair of nodes $(i,j)$ \\
$\mathcal{X}_{ij}$ & nonzero (typed-graphlet, count) pairs for edge $(i,j) \in E$ \\

$\Psi$ & hash table for checking whether a node is connected to $i$ or $j$ and its ``relationship'' (\eg, $\lambda_{1}$, $\lambda_{2}$, $\lambda_{3}$) in constant time \\

\bottomrule
\end{tabularx}}
\end{table}

\definecolor{typeOneColor}{RGB}{8,81,156} 
\definecolor{typeTwoColor}{RGB}{222,45,38} 
\definecolor{typeThreeColor}{RGB}{49,163,84} 
\definecolor{typeFourColor}{RGB}{117,107,177} 

\makeatletter
\global\let\tikz@ensure@dollar@catcode=\relax
\makeatother
\tikzstyle{every node}=[font=\large,line width=1.5pt]
\begin{figure}[h!]

\tikzstyle{background-page}=[rectangle,
fill=gray!30,
fill=white,
outer ysep=0.1cm,
outer xsep=0.0cm,
inner xsep=0.2cm,
inner ysep=0.5cm,
rounded corners=0mm]  
\tikzstyle{background-white}=[rectangle,
fill=white,
inner ysep=0.5cm,
rounded corners=0mm]
\scalebox{1.0}{

\tikzstyle{background-page}=[rectangle,
fill=gray!30,
outer ysep=0.1cm,
inner sep=0.5cm,
rounded corners=0mm]  
\tikzstyle{background-white}=[rectangle,
fill=white,
inner ysep=0.5cm,
rounded corners=0mm]
\hspace{-5.0mm}
\subfigure[\vspace{-2mm}Typed Graphlets (Ours)]{
\scalebox{0.96}{

\scalebox{0.28}{
\begin{tikzpicture}[-,>=latex,auto,node distance=2.0cm,thick,
main node/.style={circle,draw=black,fill=white,draw,font=\sffamily\Huge\bfseries,text=black,minimum width=0.9cm, line width=1mm},
]

\node[main node] (1) {}; 
\node[main node] (2) [below of=1] {}; 
\node[main node] (3) [below of=2] {}; 

\tikzstyle{LabelStyle}=[below=3pt]
\path[every node/.style={font=\sffamily}] 
(1) edge [line width=1.0mm, left] node [above left] {} (2) 
(2) edge [line width=1.0mm, left] node[below left] {} (3);

\begin{pgfonlayer}{background}
		\node [background-white, 
fit=(1) (2) (3),
label=below:\fontsize{32}{32}\selectfont 
\textcolor{white}{$H_1$}
] {};
\end{pgfonlayer} 

\end{tikzpicture}
}
\hspace{-2.7mm}
\scalebox{0.28}{
\begin{tikzpicture}[-,>=latex,auto,node distance=2.0cm,thick,fill=none,
main node/.style={circle,draw=black,fill=none,draw,font=\sffamily\Huge\bfseries,text=black,minimum width=0.9cm},
white node/.style={draw=white,fill=none,draw,font=\sffamily\Huge\bfseries,text=black,minimum width=0.9cm}
]
\node[white node] (1) {};
\node[white node] (2) [below of=1] {};
\node[white node] (3) [below of=2] {};
\node[white node] (4) [right of=2, left=30pt, below=10pt, above=0.05pt] {\vspace{2mm}\fontsize{56}{56}\selectfont $\Rightarrow$};

\begin{pgfonlayer}{background}
		\node [background-white, 
fit=(1) (2) (3) (4),
label=below:\fontsize{32}{32}\selectfont 
\textcolor{white}{$H_1$}
] {};
\end{pgfonlayer} 

\end{tikzpicture}
}
\hspace{-2mm}
\scalebox{0.28}{
\begin{tikzpicture}[-,>=latex,auto,node distance=2.0cm,thick,
typeOne node/.style={circle,draw=typeOneColor,fill=typeOneColor,draw,font=\sffamily\Huge\bfseries,text=white,minimum width=0.9cm},
typeTwo node/.style={circle,draw=typeTwoColor,fill=typeTwoColor,draw,font=\sffamily\Huge\bfseries,text=white,minimum width=0.9cm},
typeThree node/.style={circle,draw=typeThreeColor,fill=typeThreeColor,draw,font=\sffamily\Huge\bfseries,text=white,minimum width=0.9cm},
]

\node[typeOne node] (1) {}; 
\node[typeOne node] (2) [below of=1] {}; 
\node[typeOne node] (3) [below of=2] {}; 

\tikzstyle{LabelStyle}=[below=3pt]
\path[every node/.style={font=\sffamily}] 
(1) edge [line width=1.0mm, left] node [above left] {} (2) 
(2) edge [line width=1.0mm, left] node[below left] {} (3);

\begin{pgfonlayer}{background}
\node [background-page, 
fit=(1) (2) (3),
label=below:\fontsize{32}{32}\selectfont $H_1$
] {};
\end{pgfonlayer} 
\end{tikzpicture}
}
\label{fig:typed-3-path-homo-typeOne-3colors}
\hspace{0.3mm}
\scalebox{0.28}{
\begin{tikzpicture}[-,>=latex,auto,node distance=2.0cm,thick,
typeOne node/.style={circle,draw=typeOneColor,fill=typeOneColor,draw,font=\sffamily\Huge\bfseries,text=white,minimum width=0.9cm},
typeTwo node/.style={circle,draw=typeTwoColor,fill=typeTwoColor,draw,font=\sffamily\Huge\bfseries,text=white,minimum width=0.9cm},
typeThree node/.style={circle,draw=typeThreeColor,fill=typeThreeColor,draw,font=\sffamily\Huge\bfseries,text=white,minimum width=0.9cm},
]

\node[typeTwo node] (1) {}; 
\node[typeTwo node] (2) [below of=1] {}; 
\node[typeTwo node] (3) [below of=2] {}; 

\tikzstyle{LabelStyle}=[below=3pt]
\path[every node/.style={font=\sffamily}] 
(1) edge [line width=1.0mm, left] node [above left] {} (2) 
(2) edge [line width=1.0mm, left] node[below left] {} (3);

\begin{pgfonlayer}{background}
\node [background-page, 
fit=(1) (2) (3),
label=below:\fontsize{32}{32}\selectfont $H_2$
] {};
\end{pgfonlayer} 
\end{tikzpicture}
}
\hspace{0.3mm}
\scalebox{0.28}{
\begin{tikzpicture}[-,>=latex,auto,node distance=2.0cm,thick,
typeOne node/.style={circle,draw=typeOneColor,fill=typeOneColor,draw,font=\sffamily\Huge\bfseries,text=white,minimum width=0.9cm},
typeTwo node/.style={circle,draw=typeTwoColor,fill=typeTwoColor,draw,font=\sffamily\Huge\bfseries,text=white,minimum width=0.9cm},
typeThree node/.style={circle,draw=typeThreeColor,fill=typeThreeColor,draw,font=\sffamily\Huge\bfseries,text=white,minimum width=0.9cm},
]

\node[typeOne node] (1) {}; 
\node[typeOne node] (2) [below of=1] {}; 
\node[typeTwo node] (3) [below of=2] {}; 

\node[typeTwo node] (4) [right of=1] {}; 
\node[typeOne node] (5) [below of=4] {}; 
\node[typeOne node] (6) [below of=5] {}; 

\node[typeOne node] (7) [right of=4] {}; 
\node[typeTwo node] (8) [below of=7] {}; 
\node[typeOne node] (9) [below of=8] {}; 

\tikzstyle{LabelStyle}=[below=3pt]
\path[every node/.style={font=\sffamily}] 
(1) edge [line width=1.0mm, left] node [above left] {} (2) 
(2) edge [line width=1.0mm, left] node[below left] {} (3)

(4) edge [line width=1.0mm, left] node [above left] {} (5) 
(5) edge [line width=1.0mm, left] node[below left] {} (6)

(7) edge [line width=1.0mm, left] node [above left] {} (8) 
(8) edge [line width=1.0mm, left] node[below left] {} (9);

\begin{pgfonlayer}{background}
\node [background-page, 
fit=(1) (2) (3) (4) (5) (6) (7) (8) (9),
label=below:\fontsize{32}{32}\selectfont $H_3$
] {};
\end{pgfonlayer} 
\end{tikzpicture}
}
\hspace{0.3mm}
\scalebox{0.28}{
\begin{tikzpicture}[-,>=latex,auto,node distance=2.0cm,thick,
typeOne node/.style={circle,draw=typeOneColor,fill=typeOneColor,draw,font=\sffamily\Huge\bfseries,text=white,minimum width=0.9cm},
typeTwo node/.style={circle,draw=typeTwoColor,fill=typeTwoColor,draw,font=\sffamily\Huge\bfseries,text=white,minimum width=0.9cm},
typeThree node/.style={circle,draw=typeThreeColor,fill=typeThreeColor,draw,font=\sffamily\Huge\bfseries,text=white,minimum width=0.9cm},
]

\node[typeTwo node] (4) {}; 
\node[typeTwo node] (5) [below of=4] {}; 
\node[typeOne node] (6) [below of=5] {}; 

\node[typeOne node] (7) [right of=4] {}; 
\node[typeTwo node] (8) [below of=7] {}; 
\node[typeTwo node] (9) [below of=8] {}; 

\node[typeTwo node] (10) [right of=7] {}; 
\node[typeOne node] (11) [below of=10] {}; 
\node[typeTwo node] (12) [below of=11] {}; 

\tikzstyle{LabelStyle}=[below=3pt]
\path[every node/.style={font=\sffamily}] 

(4) edge [line width=1.0mm, left] node [above left] {} (5) 
(5) edge [line width=1.0mm, left] node[below left] {} (6)

(7) edge [line width=1.0mm, left] node [above left] {} (8) 
(8) edge [line width=1.0mm, left] node[below left] {} (9)

(10) edge [line width=1.0mm, left] node [above left] {} (11) 
(11) edge [line width=1.0mm, left] node[below left] {} (12);

\begin{pgfonlayer}{background}
\node [background-page, 
fit=(4) (5) (6) (7) (8) (9) (10) (11) (12),
label=below:\fontsize{32}{32}\selectfont $H_4$
] {};
\end{pgfonlayer} 
\end{tikzpicture}
}

\hspace{6.0mm}
\label{fig:typed-3-path-with-2-types}
}
}
\hfill
\subfigure[\vspace{-0.0mm}Colored Graphlets~\protect\cite{gu2018heterAlignment}]{
\scalebox{0.96}{

\scalebox{0.28}{
\begin{tikzpicture}[-,>=latex,auto,node distance=2.0cm,thick,
main node/.style={circle,draw=black,fill=white,draw,font=\sffamily\Huge\bfseries,text=black,minimum width=0.9cm, line width=1mm},
]

\node[main node] (1) {}; 
\node[main node] (2) [below of=1] {}; 
\node[main node] (3) [below of=2] {}; 

\tikzstyle{LabelStyle}=[below=3pt]
\path[every node/.style={font=\sffamily}] 
(1) edge [line width=1.0mm, left] node [above left] {} (2) 
(2) edge [line width=1.0mm, left] node[below left] {} (3);

\begin{pgfonlayer}{background}
		\node [background-white, 
fit=(1) (2) (3),
label=below:\fontsize{32}{32}\selectfont 
\textcolor{white}{$H_1$}
] {};
\end{pgfonlayer} 

\end{tikzpicture}
}
\hspace{-2.7mm}
\scalebox{0.28}{
\begin{tikzpicture}[-,>=latex,auto,node distance=2.0cm,thick,
main node/.style={circle,draw=black,fill=white,draw,font=\sffamily\Huge\bfseries,text=black,minimum width=0.9cm},
white node/.style={draw=white,draw,font=\sffamily\Huge\bfseries,text=black,minimum width=0.9cm}
]
\node[white node] (1) {};
\node[white node] (2) [below of=1] {};
\node[white node] (3) [below of=2] {};
\node[white node] (4) [right of=2, left=30pt, below=10pt, above=0.05pt] {\vspace{2mm}\fontsize{56}{56}\selectfont $\Rightarrow$};

\begin{pgfonlayer}{background}
		\node [background-white, 
fit=(1) (2) (3) (4),
label=below:\fontsize{32}{32}\selectfont 
\textcolor{white}{$H_1$}
] {};
\end{pgfonlayer} 

\end{tikzpicture}
}
\hspace{-2mm}
\scalebox{0.28}{
\begin{tikzpicture}[-,>=latex,auto,node distance=2.0cm,thick,
typeOne node/.style={circle,draw=typeOneColor,fill=typeOneColor,draw,font=\sffamily\Huge\bfseries,text=white,minimum width=0.9cm},
typeTwo node/.style={circle,draw=typeTwoColor,fill=typeTwoColor,draw,font=\sffamily\Huge\bfseries,text=white,minimum width=0.9cm},
typeThree node/.style={circle,draw=typeThreeColor,fill=typeThreeColor,draw,font=\sffamily\Huge\bfseries,text=white,minimum width=0.9cm},
]

\node[typeOne node] (1) {}; 
\node[typeOne node] (2) [below of=1] {}; 
\node[typeOne node] (3) [below of=2] {}; 

\tikzstyle{LabelStyle}=[below=3pt]
\path[every node/.style={font=\sffamily}] 
(1) edge [line width=1.0mm, left] node [above left] {} (2) 
(2) edge [line width=1.0mm, left] node[below left] {} (3);

\begin{pgfonlayer}{background}
\node [background-page, 
fit=(1) (2) (3),
label=below:\fontsize{32}{32}\selectfont $C_1$
] {};
\end{pgfonlayer} 
\end{tikzpicture}
}
\label{fig:typed-3-path-homo-typeOne-3colors}
\hspace{0.3mm}
\scalebox{0.28}{
\begin{tikzpicture}[-,>=latex,auto,node distance=2.0cm,thick,
typeOne node/.style={circle,draw=typeOneColor,fill=typeOneColor,draw,font=\sffamily\Huge\bfseries,text=white,minimum width=0.9cm},
typeTwo node/.style={circle,draw=typeTwoColor,fill=typeTwoColor,draw,font=\sffamily\Huge\bfseries,text=white,minimum width=0.9cm},
typeThree node/.style={circle,draw=typeThreeColor,fill=typeThreeColor,draw,font=\sffamily\Huge\bfseries,text=white,minimum width=0.9cm},
]

\node[typeTwo node] (1) {}; 
\node[typeTwo node] (2) [below of=1] {}; 
\node[typeTwo node] (3) [below of=2] {}; 

\tikzstyle{LabelStyle}=[below=3pt]
\path[every node/.style={font=\sffamily}] 
(1) edge [line width=1.0mm, left] node [above left] {} (2) 
(2) edge [line width=1.0mm, left] node[below left] {} (3);

\begin{pgfonlayer}{background}
\node [background-page, 
fit=(1) (2) (3),
label=below:\fontsize{32}{32}\selectfont $C_2$
] {};
\end{pgfonlayer} 
\end{tikzpicture}
}
\hspace{0.3mm}
\scalebox{0.28}{
\begin{tikzpicture}[-,>=latex,auto,node distance=2.0cm,thick,
typeOne node/.style={circle,draw=typeOneColor,fill=typeOneColor,draw,font=\sffamily\Huge\bfseries,text=white,minimum width=0.9cm},
typeTwo node/.style={circle,draw=typeTwoColor,fill=typeTwoColor,draw,font=\sffamily\Huge\bfseries,text=white,minimum width=0.9cm},
typeThree node/.style={circle,draw=typeThreeColor,fill=typeThreeColor,draw,font=\sffamily\Huge\bfseries,text=white,minimum width=0.9cm},
]

\node[typeOne node] (1) {}; 
\node[typeOne node] (2) [below of=1] {}; 
\node[typeTwo node] (3) [below of=2] {}; 

\node[typeTwo node] (4) [right of=1] {}; 
\node[typeOne node] (5) [below of=4] {}; 
\node[typeOne node] (6) [below of=5] {}; 

\node[typeOne node] (7) [right of=4] {}; 
\node[typeTwo node] (8) [below of=7] {}; 
\node[typeOne node] (9) [below of=8] {}; 

\node[typeTwo node] (10) [right of=7] {}; 
\node[typeTwo node] (11) [below of=10] {}; 
\node[typeOne node] (12) [below of=11] {}; 

\node[typeOne node] (13) [right of=10] {}; 
\node[typeTwo node] (14) [below of=13] {}; 
\node[typeTwo node] (15) [below of=14] {}; 

\node[typeTwo node] (16) [right of=13] {}; 
\node[typeOne node] (17) [below of=16] {}; 
\node[typeTwo node] (18) [below of=17] {}; 

\tikzstyle{LabelStyle}=[below=3pt]
\path[every node/.style={font=\sffamily}] 
(1) edge [line width=1.0mm, left] node [above left] {} (2) 
(2) edge [line width=1.0mm, left] node[below left] {} (3)

(4) edge [line width=1.0mm, left] node [above left] {} (5) 
(5) edge [line width=1.0mm, left] node[below left] {} (6)

(7) edge [line width=1.0mm, left] node [above left] {} (8) 
(8) edge [line width=1.0mm, left] node[below left] {} (9)

(10) edge [line width=1.0mm, left] node [above left] {} (11) 
(11) edge [line width=1.0mm, left] node[below left] {} (12)

(13) edge [line width=1.0mm, left] node [above left] {} (14) 
(14) edge [line width=1.0mm, left] node[below left] {} (15)

(16) edge [line width=1.0mm, left] node [above left] {} (17) 
(17) edge [line width=1.0mm, left] node[below left] {} (18);

\begin{pgfonlayer}{background}
\node [background-page, 
fit=(1) (2) (3) (4) (5) (6) (7) (8) (9) (10) (11) (12) (13) (14) (15) (16) (17) (18),
label=below:\fontsize{32}{32}\selectfont $C_3$
] {};
\end{pgfonlayer} 
\end{tikzpicture}
}

\label{fig:colored-3-path-with-2-colors}
}
}

} 

\vspace{-2mm}
\caption{
\textbf{Typed graphlets vs. colored graphlets.}
The intuitive example shows the difference between typed graphlets that are formally defined in this paper and colored graphlets from~\protect\cite{gu2018heterAlignment}.
In particular, (a) shows the typed 3-paths with $L=2$ types whereas (b) shows the ``colored 3-paths''.
In the above example, there are three colored graphlets, that is, the last 4 typed graphlets are considered a single colored graphlet.
Note given $L$ colors, there are $2^{L}-1$ colored graphlets.
However, for $K$ nodes and $L$ types, there are $L+K-1 \choose K$ typed graphlets.
}
\label{fig:typed-graphlet-vs-colored-graphlet-key-differences}
\end{figure}

Empirically, the proposed approach is shown to be orders of magnitude faster than state-of-the-art methods for the simpler colored graphlet counting problem. 
In particular, we observe between 89 and 10,981 times speedup in runtime performance compared to the best method.
Notably, on graphs of even moderate size (thousands of nodes/edges), these approaches fail to finish in a reasonable amount of time (24 hours).
In terms of space, the proposed approach uses between 42x and 776x less space than these methods.
We also demonstrate the parallel scaling of the parallel algorithm and observe nearly linear speedups as the number of processing units increases.
In addition to real-world graphs from a wide range of domains, we also show results on a number of synthetically generated graphs from a variety of graph models.
Finally, we demonstrate the utility of typed graphlets for exploratory network analysis using a variety of well-known networks.

Compared to the untyped/homogeneous graphlet counting problem 
(which has found many important applications~\cite{koyuturk2006pairwise,prvzulj2007biological,vishwanathan2010graph,shervashidze2009efficient,solava2012graphlet,benson2016higher,ahmed17aaai,role2vec,noble2003graph,akoglu2015graph,Rossi2018a}), 
typed graphlets are more powerful containing a significant amount of additional information.
We show this formally using information theory (Section~\ref{sec:complexity-analysis}) and demonstrate the importance of typed graphlets empirically using real-world graphs for exploratory analysis (Section~\ref{sec:exploratory-analysis}) and graph-based predictive modeling (Section~\ref{sec:exp-link-pred}).
Importantly, we find that only a handful of the possible typed graphlets actually occur in the real-world graphs studied in this work (Table~\ref{table:unique-typed-motif-occur}).
Furthermore, among the typed graphlets with nonzero counts (\ie, the typed graphlets that actually occur in $G$), we find that a few of those typed graphlets occur very frequently while the vast majority have very few occurrences (see Figure~\ref{fig:typed-tri-prob-dist-cora-citeseer} and Figure~\ref{fig:typed-4-node-prob-dist-cora-citeseer} for a few examples).
This observation indicates a power-law relationship between the counts of the different typed graphlets.
The rare typed graphlets (\ie, typed graphlets that rarely occur in the graph) also contain useful information as the appearance of these typed graphlets may indicate anomalies/outliers or simply unique structural behaviors that are fundamentally important but extremely difficult to identify using traditional methods.
Moreover, the typed graphlets found to be important are easily interpretable and provide key insights into the structure and underlying phenomena governing the formation of the complex network that would otherwise be hidden using traditional untyped methods, see Section~\ref{sec:exploratory-analysis} for further details.
Finally, we also demonstrate the effectiveness of typed graphlets in Section~\ref{sec:exp-link-pred} for improving a predictive modeling task.

This work introduces and formally defines a generalization of the notion of graphlet to heterogeneous networks called \emph{typed graphlets}.
We describe a general framework for counting the proposed formalization of typed graphlets.
The proposed framework has the following desired properties:
{
\begin{itemize}
\setlength{\itemsep}{4pt}

\item \textbf{Fast}: The approach is fast for large graphs by leveraging novel \emph{non-trivial combinatorial relationships} to derive many of the typed graphlets in $o(1)$ constant time. 
Theoretically, the worst-case time complexity is shown to match the best untyped graphlet algorithm (Section~\ref{sec:time-complexity}).
As shown in Table~\ref{table:runtime-perf}-\ref{table:runtime-speedup-perf}, the approach is orders of magnitude faster than recent methods proposed for the simpler colored graphlet problem.

\item \textbf{Space-Efficient}: The approach is space-efficient by hashing and storing only the typed graphlet counts that appear on a given edge.

\item \textbf{Scalable for Large Networks}:
The proposed approach is scalable for large heterogeneous networks.
In particular, the approach scales nearly linearly as the size of the graph increases.

\item \textbf{Parallel}: 
The typed graphlet approach
lends itself to an efficient lock-free \& asynchronous parallel implementation.
We observe near-linear parallel scaling results in Section~\ref{sec:exp-parallel-scaling}. 

\item \textbf{Effectiveness}: 
We demonstrate the utility of typed graphlets for graph mining/exploratory analysis (Section~\ref{sec:exploratory-analysis}) and predictive modeling 
(Section~\ref{sec:exp-link-pred}) where leveraging typed graphlets significantly improves predictive performance.
This work brings new opportunities to leverage typed graphlets for many other real-world applications.

\end{itemize}
}

\section{Related Work} \label{sec:related-work}
\noindent
Closest work related to our own is that of colored graphlets~\cite{gu2018heterAlignment,ribeiro2014discovering}.
However, the notion of colored graphlet is different from the notion of typed graphlets (and position-aware typed graphlets (Def.~\ref{def:position-aware-typed-graphlet-instance})) that are formally defined and investigated in this paper.
For an intuitive illustration of the difference between the proposed notion of typed graphlets and the  colored graphlet counting problem studied in prior work, see Figure~\ref{fig:typed-graphlet-vs-colored-graphlet-key-differences}.
Besides the difference in problem, all of the prior work has focused on counting colored graphlets for \emph{nodes} whereas this paper focuses on the problem of counting typed graphlets for \emph{edges} (or more generally, between a pair of nodes $i$ and $j$).
It is straightforward to see that the definition of colored graphlets from~\citet{gu2018heterAlignment} is only able to cover a subset of the typed graphlets given by our definition.
Thus, the notion of typed graphlet described in our work is more general than the notion of colored graphlet.
Besides the fundamental difference in problem as shown in Figure~\ref{fig:typed-graphlet-vs-colored-graphlet-key-differences}, that work also focused mainly
on the application to network alignment (using very small networks) and 
not on the approach for computing colored graphlets.
Nevertheless, the method GC used in that work and the other methods for colored graphlets 
are only able to handle extremely small graphs as shown in Table~\ref{table:runtime-perf}.

Despite the difference between colored and typed graphlets (Figure~\ref{fig:typed-graphlet-vs-colored-graphlet-key-differences}),
the approach proposed in this work also differs from the colored graphlet methods in three fundamental ways.
First, while we leverage new combinatorial relationships to derive a number of typed graphlets in $o(1)$ constant time, GC and other colored graphlet methods must enumerate all 
graphlets in order to obtain their color configuration.
Therefore, our approach is significantly faster (even though we count typed graphlets, a more complex and representationally powerful notion) than these methods as they require a lot of extra work to compute the colored graphlets that our approach can derive in constant time.\footnote{Notice from Figure~\ref{fig:typed-graphlet-vs-colored-graphlet-key-differences} that any method for counting typed graphlets can by definition be used to count colored graphlets, but not vice-versa.}
For instance, the small citeseer graph with only 3.3k nodes and 4.5k edges takes 46.27 seconds using the best method (for colored graphlets) whereas our approach for typed graphlets takes only a fraction of a second, notably, $\nicefrac{2}{100}$ seconds.
In addition, while the methods for colored graphlets (a simpler relaxation of typed graphlet, see Figure~\ref{fig:typed-graphlet-vs-colored-graphlet-key-differences}) are only able to handle small networks, our approach naturally scales to large networks with millions of nodes and edges (Section~\ref{sec:exp}).
Second, our approach is significantly more space-efficient and stores only the nonzero counts of the typed graphlets discovered at each edge (Section~\ref{sec:exp-space-efficiency}).
Third, our approach lends itself to an efficient, lock-free, and asynchronous parallelization.
As an aside, unlike the methods for colored graphlets, our approach enumerates only a few typed graphlets and derives the remaining typed graphlets in $o(1)$ constant time using new non-trivial combinatorial relationships that involve counts of lower-order typed graphlets.
These lower-order typed graphlet counts are used as building blocks to directly derive many of the higher-order typed graphlet counts directly without any enumeration or knowledge of the explicit node types.
Therefore, the worst-case time complexity of the proposed approach is equivalent to the best known untyped/homogeneous graphlet algorithm (as shown formally in Section~\ref{sec:complexity-analysis}).

\section{Heterogeneous Graphlets}
\label{sec:typed-network-motifs}
\noindent
This section introduces a generalization of graphlets called 
\emph{heterogeneous graphlets} 
(or simply \emph{typed graphlets}).
See Table~\ref{table:notation} for a summary of key notation.

\subsection{Heterogeneous Graph Model} \label{sec:heter-graph-model}
\noindent
We use the following heterogeneous graph formulation:
\begin{Definition}[Heterogeneous network] \label{def:heter-network}
A heterogeneous network is defined as $G=(V,E)$ consisting of a set of node objects $V$ and a set of edges $E$ connecting the nodes in $V$.
A heterogeneous network also has a \emph{node type mapping function} $\,\phi : V \rightarrow \mathcal{T}_V$ 
and an \emph{edge type mapping function} defined as $\,\xi : E \rightarrow \mathcal{T}_E$ 
where $\mathcal{T}_V$ and $\mathcal{T}_E$ denote the set of node object types and edge types, respectively.
The type of node $i$ is denoted as $\phi_i$ 
whereas the type of edge $e = (i,j) \in E$ is denoted as $\xi_{ij}=\xi_e$.
\end{Definition}\noindent
A few special cases of heterogeneous networks are shown in Figure~\ref{fig:heter-special-cases}.

\begin{figure}[h!]
\vspace{1mm}
\begin{minipage}{0.54\linewidth}
\centering
\hspace{-2mm}
\includegraphics[width=0.60\linewidth]{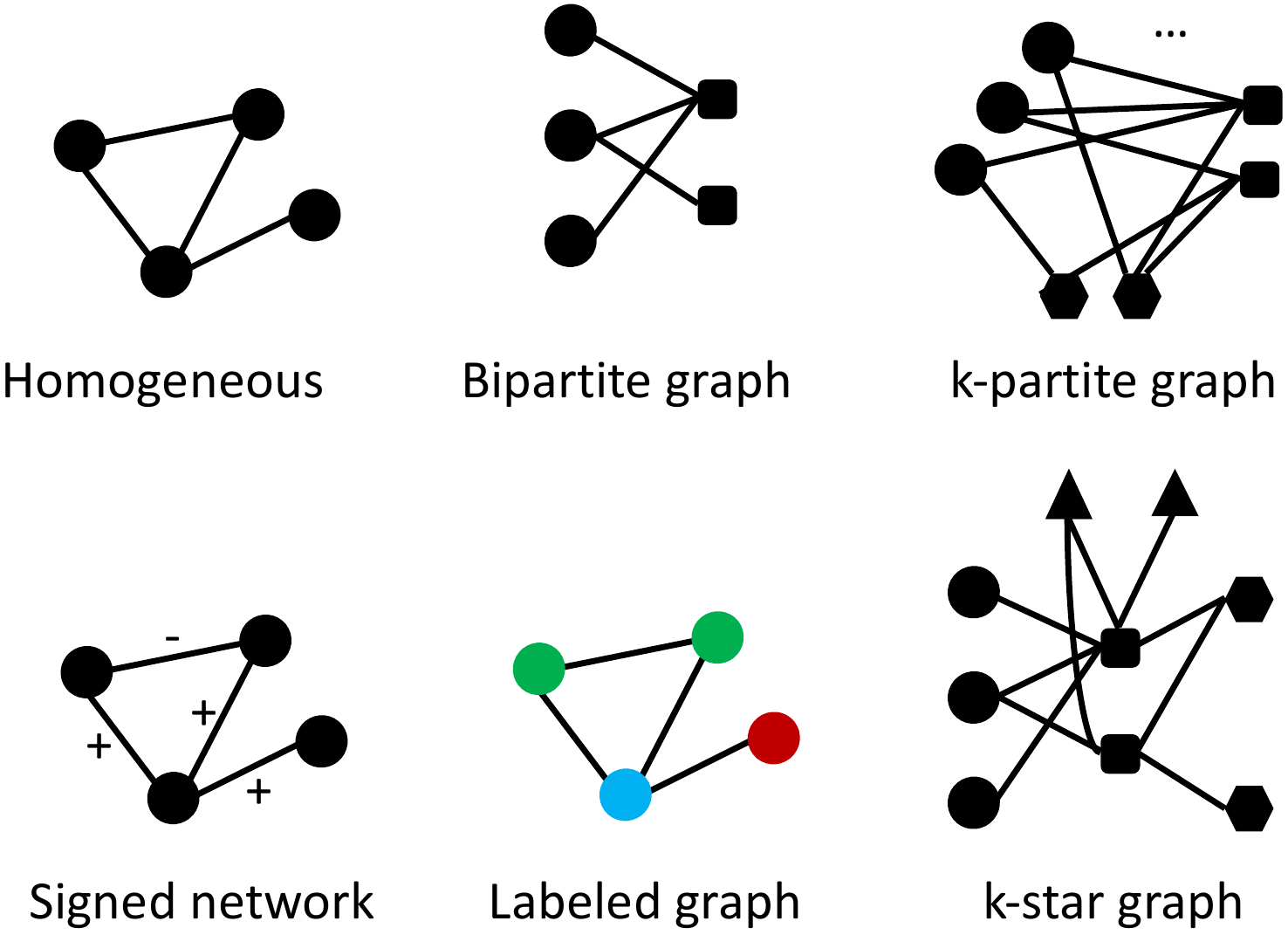} 
\end{minipage}
\hfill
\begin{minipage}{0.45\linewidth}
\centering
\scalebox{0.90}{
\small
\setlength{\tabcolsep}{9pt}
\begin{tabularx}{0.90\linewidth}{rH cc HH}
\toprule
\textbf{Graph Type} &&  $|\mathcal{T}_V|$ & $|\mathcal{T}_E|$ \\
\midrule
\textsc{Homogeneous} && $1$ & $1$ \\
\textsc{Bipartite} && $2$ & $1$ \\
\textsc{K-partite} && $k$ & $\frac{1}{2} k(k-1)$ \\ 
\textsc{Signed} && $1$ & $2$ \\
\textsc{Labeled} && $k$ & $\ell$ \\
\textsc{Star} && $k$ & $k-1$ \\
\bottomrule
\end{tabularx}
}
\hspace{3mm}
\end{minipage}
\vspace{-2mm}
\caption{
\emph{Typed graphlets} are useful for a wide variety of graphs.
These graphs are only a few examples that are naturally supported by the proposed framework.
}
\label{fig:heter-special-cases}
\end{figure}

\subsection{Graphlet Generalization} \label{sec:typed-graphlet-formulation}
\noindent
In this section, we introduce a more general notion of graphlet called \emph{typed graphlet} that naturally extends to both homogeneous and general heterogeneous networks.
We use $G$ to represent a graph and $H$ or $F$ to represent graphlets.

\subsubsection{Untyped Graphlets}
\noindent
We begin by defining untyped graphlets for graphs with a single type.

\begin{Definition}[\sc Untyped Graphlet] \label{def:graphlet}
An untyped graphlet $H$ is a connected induced subgraph of $G$.
\end{Definition} \noindent

Given a graphlet in some graph, it may be the case that we can find other topologically identical ``appearances" of this structure in that graph. 
We call these ``appearances" \textit{graphlet instances}.

\begin{Definition}[\scshape Untyped Graphlet Instance]\label{def:graphlet-instance}
An instance of an untyped graphlet $H$ in graph $G$ is an untyped graphlet $F$ in $G$ that is isomorphic to $H$.
\end{Definition}

\subsubsection{Typed Graphlets}
\noindent
In heterogeneous graphs, nodes/edges can be of many different types and so explicitly and jointly modeling such types is essential.
In this work, we introduce the notion of a \emph{typed graphlet} that explicitly captures both the connectivity pattern of interest and the types.
Notice that typed graphlets are a generalization of graphlets to heterogeneous networks.

\begin{Definition}[\scshape Typed Graphlet]\label{def:typed-graphlet}
A typed graphlet of a graph $G=(V,E,\phi,\xi)$ is a connected induced heterogeneous subgraph $H=(V',E',\phi',\xi')$ of $G$ such that
\begin{compactenum}
\item $(V',E')$ is a graphlet of $(V,E)$,
\item $\phi'=\phi|_{V'}$, that is, $\phi'$ is the restriction of $\phi$ to $V'$,
\item $\xi'=\xi|_{E'}$, that is, $\xi'$ is the restriction of $\xi$ to $E'$.
\end{compactenum}
\end{Definition}\noindent
The terms typed graphlet and heterogeneous graphlet are used interchangeably.
See Figure \ref{fig:typed-motifs-3colors} for examples of typed graphlets and untyped graphlets (in which the type structure is ignored).
We can consider the presence of topologically identical ``appearances" of a typed graphlet in a graph.
\begin{Definition}[\scshape Typed Graphlet Instance]\label{def:typed-graphlet-instance}
An instance of a typed graphlet $H=(V',E',\phi',\xi')$ of graph $G$ is a typed graphlet $F=(V'',E'',\phi'',\xi'')$ of $G$ such that
\begin{compactenum}
\item $(V'',E'')$ is isomorphic to $(V',E')$, 
\item $\mathcal{T}_{V''}=\mathcal{T}_{V'}$ and $\mathcal{T}_{E''}=\mathcal{T}_{E'}$, that is, the multisets of node and edge types are correspondingly equal.
\end{compactenum}
The set of typed graphlet instances of $H$ in $G$ is denoted as $I_G(H)$.
\end{Definition}\noindent
Comparing the above definitions of graphlet and typed graphlet, we see at first glance that typed graphlets are nontrivial extensions of their homogeneous counterparts.
The ``position'' of an edge (node) in a typed graphlet is often topologically important, \eg, an edge at the end of the 4-path vs. an edge at the center of a 4-path.
These topological differences of a typed graphlet are called (automorphism) \emph{typed orbits} since they take into account ``symmetries'' between edges (nodes) of a graphlet.
Typed graphlet orbits are a generalization of (homogeneous) graphlet orbits~\cite{prvzulj2007biological}.

\subsection{Number of Typed Graphlets}
\noindent
For a single $K$-node untyped graphlet (\eg, $K$-clique), the number of \emph{typed graphlets} with $L$ types is:
\begin{equation} \label{eq:num-typed-graphlets-for-motif}
\left( \!\binom{L}{K} \!\right) = \binom{L+K-1}{K}
\end{equation}\noindent
where $L=$ number of types and $K=$ size of the graphlet ($\#$ of nodes).
Table~\ref{table:typed-graphlets-example} shows the number of \emph{typed graphlets} that arise from a single graphlet $H \in \mathcal{H}$ of size $K\in \{2,\ldots,4\}$ nodes as the number of types varies from $L=1,2,\ldots,9$.
For instance, the total number of typed graphlet orbits with 4 nodes that arise from 7 types is $10 \cdot 210 = 2100$ since there are 10 connected 4-node (untyped) graphlet orbits.
See Figure~\ref{fig:typed-motifs-3colors} for other examples.
Unlike homogeneous graphlets, it is obviously impossible to show all the heterogeneous graphlets counted by the proposed approach since it works for general heterogeneous graphs with any arbitrary number of types $L$ and structure.

\begin{table}[h!]
\vspace{1mm}
\centering
\renewcommand{\arraystretch}{1.1} 
\caption{Number of \emph{typed graphlets} (for a single untyped graphlet) as the size $K$ (\ie, \# of nodes in the typed graphlet) and number of types $L$ varies.}
\label{table:typed-graphlets-example}
\vspace{-3mm}
\scalebox{0.85}{
\begin{tabularx}{0.80\linewidth}{Hl XXXXX XXX X}
\toprule
&& \multicolumn{9}{c}{\bf Types $L$} \\
\cmidrule(l{3pt}r{1pt}){3-11}

&& \textbf{1} & \textbf{2} & \textbf{3} & \textbf{4}  & \textbf{5} & \textbf{6}  & \textbf{7} & \textbf{8} & \textbf{9}  \\
\midrule
& \textbf{K=2} & 1  & 3  & 6  & 10  & 15  & 21  & 28 & 36  & 45 \\
\multirow{2}{*}{\bf K} &
\textbf{K=3} & 1  & 4  & 10  & 20  & 35  & 56  & 84   & 120 & 165 \\
& \textbf{K=4} & 1  & 5  & 15 & 35  & 70  & 126  & 210   & 330 & 495  \\
\bottomrule
\end{tabularx}
}
\end{table}

\subsection{Generalization to Other Graphs} \label{sec:special-cases-heter-generalization}
\noindent
The proposed notion of \emph{typed graphlets} can be used for applications on bipartite, k-partite, signed, attributed, and more generally heterogeneous networks.
A few examples of such graphs are shown in Figure~\ref{fig:heter-special-cases}.
The proposed framework naturally handles general heterogeneous networks with arbitrary structure and an arbitrary number of types.
It is straightforward to see that homogeneous, bipartite, k-partite, signed, star, and attributed networks are all special cases of heterogeneous graphs (Figure~\ref{fig:heter-special-cases}).
Therefore, the framework for deriving typed graphlets can easily support such networks.
For attributed graphs with more than one attribute/feature, the attributes of a node or edge can be mapped to types using any arbitrary approach such as role2vec~\cite{role2vec} or WL~\cite{wl}.

\section{Framework} 
\label{sec:framework}
\noindent
This section describes the general framework for counting typed graphlets.
The typed graphlet framework can be used for counting typed graphlets locally for every edge in $G$ as well as global typed graphlet counting (Problem~\ref{prob:global-typed-graphlet-counting}) that focuses on computing the total frequency of all typed graphlets.
This paper mainly focuses on the harder local typed graphlet counting problem:

\begin{Problem}[Local Typed Graphlet Counting]
\label{prob:local-typed-graphlet-counting}
Given a graph $G$ and an edge $(i, j) \in E$, 
the local typed graphlet counting problem is to find the set of all typed graphlets 
that contain nodes $i$ and $j$ and their corresponding frequencies.
This work focuses on computing all $\{2,3,4\}$-node typed graphlet counts for every edge in $G$.
\end{Problem}
Algorithm~\ref{alg:typed-motifs-exact} shows the general approach for counting all typed graphlets with up to four nodes.
Note that we do not make any restriction or assumption on the number of node or edge types. 
The algorithm naturally handles heterogeneous graphs with arbitrary number of types and structure.
See Table~\ref{table:notation} for a summary of notation.

{
\renewcommand{\setAlgFontSize}{\small}
\renewcommand{\multilinenospace}[1]{\State \parbox[t]{\dimexpr\linewidth-\algorithmicindent}{\begin{spacing}{1.0}\setAlgFontSize#1\strut \end{spacing}}}
\renewcommand{\multilinenospaceD}[1]{\State \parbox[t]{\dimexpr0.96 \linewidth-\algorithmicindent}{\begin{spacing}{1.0}\setAlgFontSize#1\strut \end{spacing}}}

\algblockdefx[parallel]{parfor}{endpar}[1][]{$\textbf{parallel for}$ #1 $\textbf{do}$}{$\textbf{end parallel}$}
\algrenewcommand{\alglinenumber}[1]{\fontsize{7.5}{8}\selectfont#1\;\;}
\begin{figure}[h!]
\vspace{-3mm}
\begin{center}
\begin{algorithm}[H]
\caption{\,Typed Graphlets}
\label{alg:typed-motifs-exact}
\begin{spacing}{1.2}
\small
\begin{algorithmic}[1]
\Require a graph $G$
\vspace{-0.5mm}
\Ensure nonzero typed graphlet counts $\mathcal{X}_{ij}$
for each edge $(i,j)\in E$

\parfor[{\bf each} $(i,j) \in E$] \label{algline:graphlet-for}

\State $T_{ij}^t = \Gamma_i^t \cap \Gamma_j^t,\quad \text{for } t=1,\ldots,L$ \Comment{typed triangles}

\State $S_{i}^t = \Gamma_i^t \setminus T_{ij}^t,\quad \text{for } t=1,\ldots,L$ \Comment{typed 3-paths centered at i}
\State $S_{j}^t = \Gamma_j^t \setminus T_{ij}^t,\quad \text{for } t=1,\ldots,L$ \Comment{typed 3-paths centered at j}

\State $|S_{ij}^t| = |S_{i}^t|+|S_{j}^t|,\quad \text{for } t=1,\ldots,L$ \Comment{typed 3-path count}

\State Store nonzero counts of the 3-node typed graphlets derived above 
\label{alg:store-nonzero-counts-3-node-typed-graphlets}

	\State Given $S_i$ and $S_j$, use Algorithm~\ref{alg:typed-path-based-motifs-exact} to derive a few typed \emph{path-based} graphlets \label{algline:main-alg-call-typed-path-based-motifs-exact}
	
\State Given $T_{ij}$, use Algorithm~\ref{alg:typed-triangle-based-motifs-exact} to derive a few typed \emph{triangle-based} graphlets
\label{algline:main-alg-call-typed-triangle-based-motifs-exact}

	\For{$t, t^{\prime} \in \{1,\ldots, L\}$ such that $t \leq t^{\prime}$} \label{algline:main-alg-for-type-pair}
		\multilinenospaceD{Derive remaining typed graphlet orbits in constant time via Eq.~\ref{eq:typed-4-path-center-orbit}-\ref{eq:typed-4-chordal-cycle-center-orbit} and update counts $\vx$ and set of typed graphlets $\mathcal{M}_{ij}$ (with nonzero count)}
		\label{algline:main-alg-derive-remaining-typed-graphlet-orbits-constant-time}
	\EndFor
	
\vspace{-0.3mm}
\For{$\hash \in \mathcal{M}_{ij}$} 
		$\mathcal{X}_{ij} = \mathcal{X}_{ij} \,\cup \{(\hash,\vx_{\hash})\}$
		\Comment{store nonzero typed graphlet counts}
			\label{algline:add-sparse-typed-motif-id-and-count-pairs}
\EndFor

\vspace{-0.7mm}
\endpar
\smallskip
\end{algorithmic}
\end{spacing}
\vspace{-1.mm}
\end{algorithm}
\end{center}
\vspace{-4.mm}
\end{figure}
}

\subsection{Counting 3-Node \emph{Typed} Graphlets} \label{sec:algorithm-3-node-typed-motifs}
\noindent
We begin by introducing the notion of a typed neighborhood \emph{and} typed degree of a node.
These are then used as a basis for deriving all typed 3-node graphlet counts in worst-case $\mathcal{O}(\Delta)$ time (Theorem~\ref{lem:time-complexity-3-node-graphlets}).
\begin{Definition}[Typed Neighborhood]\label{def:typed-neighborhood}
Given an arbitrary node $i$ in $G$, the \emph{typed neighborhood} $\Gamma_{i}^{t}$ is the set of nodes with type $t$ that are reachable by following edges originating from $i$ within $1$-hop distance.
More formally, 
\begin{equation}\label{eq:typed-neighborhood}
\Gamma_{i}^{t} = \{ j \in V \, | \, (i,j) \in E \wedge \phi_j=t \}
\end{equation}\noindent
Intuitively, a node $j \in \Gamma_{i}^{t}$ iff there exists an edge $(i,j) \in E$ between $i$ and $j$ and the type of node $j$ denoted as $\phi_j$ is $t$.
\end{Definition}\noindent
\vspace{-4mm}
\begin{Definition}[Typed Degree] \label{def:typed-degree}
The \emph{typed-degree} $d_{i}^{t}$ of node $i$ with type $t$ is defined as $d_{i}^{t} = |\Gamma_{i}^{t}|$ where $d_{i}^{t}$ is the number of nodes connected to node $i$ with type $t$.
\end{Definition}
Using these notions as a basis, we can define $S_{i}^{t}$, $S_{j}^{t}$, and $T_{ij}^{t}$ for $t=1,\ldots,L$ (Figure~\ref{fig:typed-lower-order-sets}).
Obtaining these sets is equivalent to computing all $3$-node typed graphlet counts.
These sets are all defined with respect to a given edge $(i,j) \in E$ between node $i$ and $j$ with types $\phi_i$ and $\phi_j$.
Since typed graphlets are counted for each edge $(i,j) \in E$, the types $\phi_i$ and $\phi_j$ are fixed ahead of time. 
Thus, there is only one remaining type to select for $3$-node typed graphlets.
\begin{cor}
\label{cor:typed-triangle-node}
Given an edge $(i,j) \in E$ between node $i$ and $j$ with types $\phi_i$ and $\phi_j$, 
let $T_{ij}^{t}$ denote the set of nodes of type $t$ that complete a typed triangle with node $i$ and $j$ defined as:
\begin{align} \label{eq:typed-triangle}
T_{ij}^{t} = \Gamma_{i}^{t} \cap \Gamma_j^{t}
\end{align}\noindent
where $|T_{ij}^{t}|$ denotes the number of nodes that form triangles with node $i$ and $j$ of type $t$.
\end{cor}
Let $\Gamma_{i}^{t}$ denote the set of neighbors of $i$ with type $t$.
If $k \in \Gamma_{i}^{t}$ and $k \in \Gamma_{j}^{t}$, then since $(i,j) \in E$, node $k$ must form a typed triangle with $i$ and $j$ (\ie, $k \in T_{ij}^{t}$).
Hence, $(i,j) \in E$ closes a triangle with node $k$ of type $t$.
This is straightforward to see since $k \in \Gamma_{i}^{t}$ implies $(i,k) \in E$, $k \in \Gamma_{j}^{t}$ implies $(j,k) \in E$, and $(i,j) \in E$.
Since typed triangles are counted for each edge $(i,j) \in E$, the types $\phi_i$ and $\phi_j$ are fixed ahead of time. 
Therefore, there is only one remaining type to select.
Let $t$ denote the remaining node type, then $T_{ij}^{t} = \Gamma_{i}^{t} \cap \Gamma_j^{t}$.
Furthermore,
since every node $k \in T_{ij}^{t}$ is of type $t$ and thus completes a typed triangle with node $i$ and $j$ consisting of types $\phi_i$, $\phi_j$, and $\phi_k=t$.
\begin{cor}
\label{cor:typed-3-star-node-centered-at-i}
Given an edge $(i,j) \in E$ between node $i$ and $j$ with types $\phi_i$ and $\phi_j$.
Let $S_{i}^{t}$ denote the set of nodes of type $t$ that form 3-node stars (or equivalently 3-node paths) centered at node $i$ (and not including $j$).
More formally, 
\begin{align} \label{eq:3-node-typed-star-centered-at-node-i}
S_{i}^{t} &= \big\lbrace k \in (\Gamma_{i}^{t} \setminus \{j\}) \; \big| \; k \notin \Gamma_{j}^{t} \big\rbrace \\
&= \Gamma_{i}^{t} \setminus \big(\Gamma_{j}^{t} \cup \{j\}\big) = \Gamma_{i}^{t} \setminus T_{ij}^{t}
\end{align}\noindent
where $|S_{i}^{t}|$ denotes the number of nodes of type $t$ that form 3-stars centered at node $i$ (not including $j$).
\end{cor}
Let $\Gamma_{i}^{t}$ denote the set of neighbors of $i$ with type $t$.
Let $k \in \Gamma_{i}^{t}$ be a node that forms a typed 3-star centered at $i$ with type $t$, then $k \not\in \Gamma_{j}^{t}$.
Otherwise if $k \in \Gamma_{j}^{t}$, then $k \in T_{ij}^{t}$, which implies $k \not\in S_{i}^{t}$.
Similarly, it is straightforward to define the set $S_j^{t}$ of typed 3-star/path nodes of type $t$ centered at $j$ in a similar fashion:
\begin{align} \label{eq:3-node-typed-star-centered-at-node-j}
S_{j}^{t} &= \big\lbrace k \in (\Gamma_{j}^{t} \setminus \{i\}) \; \big| \; k \notin \Gamma_{i}^{t} \big\rbrace \\
&= \Gamma_{j}^{t} \setminus \big(\Gamma_{i}^{t} \cup \{i\}\big) = \Gamma_{j}^{t} \setminus T_{ij}^{t}
\end{align}\noindent
where $|S_{j}^{t}|$ denotes the number of nodes of type $t$ that form 3-stars centered at node $j$ (not including $i$).
This follows from Corollary~\ref{cor:typed-3-star-node-centered-at-i}.
\begin{Property} \label{prop:relationship-between-typed-sets-and-untyped-sets}
\begin{align} 
T_{ij} = \bigcup_{t=1}^{L} T_{ij}^{t},\quad\;\;\;
S_{i} = \bigcup_{t=1}^{L} S_{i}^{t},\quad\;\;\;
S_{j} = \bigcup_{t=1}^{L} S_{j}^{t}
\end{align}
\end{Property}
This property follows directly from Corollary~\ref{cor:typed-triangle-node}-\ref{cor:typed-3-star-node-centered-at-i} and is shown in Figure~\ref{fig:typed-lower-order-sets}.
These lower-order 3-node typed graphlet counts are used to derive many higher-order typed graphlet counts in $o(1)$ constant time (Section~\ref{sec:combinatorial-relationships}).
\begin{cor}[Typed 3-Stars] 
\label{cor:typed-3-star-for-edge}
Given an edge $(i,j) \in E$ between node $i$ and $j$ with types $\phi_i$ and $\phi_j$, 
the number of typed 3-node stars that contain $(i,j) \in E$ with types $\phi_i$, $\phi_j$, $t$ is: 
\begin{align} \label{eq:3-node-typed-star-for-edge}
|S_{ij}^{t}| = |S_{i}^{t}|+|S_{j}^{t}|
\end{align}\noindent
where $|S_{ij}^{t}|$ denotes the number of typed 3-stars that contain nodes $i$ and $j$ with types $\phi_i$, $\phi_j$, $t$.
\end{cor}
Moreover, the number of typed triangles centered at $(i,j) \in E$ with types $\phi_i$, $\phi_j$, $t$ is simply $|T_{ij}^{t}|$ (Corollary~\ref{cor:typed-triangle-node}) whereas the number of typed 3-node stars that contain $(i,j) \in E$ with types $\phi_i$, $\phi_j$, $t$ is $|S_{ij}^{t}| = |S_{i}^{t}|+|S_{j}^{t}|$ (Corollary~\ref{cor:typed-3-star-for-edge}).
We do not need to actually store the sets $S_{i}^{t}$, $S_{j}^{t}$, and $T_{ij}^{t}$ for every type $t = 1, \ldots, L$.
We only need to store the \emph{size/cardinality} of the sets (as shown in Algorithm~\ref{alg:typed-motifs-exact}) since these are the counts of all possible 3-node typed graphlets.
For convenience, we denote the size of those sets as $|S_{i}^{t}|$, $|S_{j}^{t}|$, and $|T_{ij}^{t}|$ for all $t=1,\ldots,L$, respectively.
At this point, all typed 3-node graphlets with nonzero counts have been computed for edge $(i,j) \in E$ in $\mathcal{O}(|\Gamma_i|+|\Gamma_i|)=\mathcal{O}(\Delta)$ time where $\Delta$ is max degree (See Section~\ref{sec:time-complexity} for proof).
Note $|\Gamma_i| = \sum_{t} |\Gamma_i^{t}|$.

{
\algblockdefx[parallel]{parfor}{endpar}[1][]{$\textbf{parallel for}$ #1 $\textbf{do}$}{$\textbf{end parallel}$}
\algrenewcommand{\alglinenumber}[1]{\fontsize{6.5}{7}\selectfont#1\;\;}
\algrenewcommand{\alglinenumber}[1]{\fontsize{7.5}{8}\selectfont#1\;\;}
\begin{figure}[h!]
\vspace{-4mm}
\begin{center}
\begin{algorithm}[H]
\caption{\,
Typed \emph{Path}-based Graphlets
}
\label{alg:typed-path-based-motifs-exact}
\begin{spacing}{1.10}
\small
\begin{algorithmic}[1]
\Require 
{
a graph $G=(V,E,\Phi,\xi)$, 
an edge $(i,j)$, 
sets of nodes $S_i$ and $S_j$ that form 3-paths centered at $i$ and $j$, respectively, 
a typed graphlet count vector $\vx$ for $(i,j)$, 
and 
set $\mathcal{M}_{ij}$ of unique typed graphlets for $i$ and $j$.
}
\smallskip

	\For{{\bf each} $w_k \in S_i$}
	\label{algline:Si-three-paths-centered-at-i}
	    
	    \For{$w_r \in \Gamma_{w_k} \!\setminus \{i,j\}$}

\If{$w_r \not\in (\Gamma_i \cup \Gamma_j)$} 
\Comment{\emph{typed 4-path-edge} orbit} 

\State $\langle \vx, \mathcal{M}_{ij} \rangle = \textsc{Update}(\vx,\mathcal{M}_{ij},\mathbb{F}(g_3, \Phi_i, \Phi_j, \Phi_{w_k}, \Phi_{w_r}))$

\ElsIf{\!$w_r \!\in \!S_i \wedge w_r \!\leq\! w_k$\!} 
\Comment{\emph{typed tailed-tri (tail orbit)}}  
\label{algline:tailed-tri-edge-tail-edge-orbit-Si}

		   \State $\langle \vx, \mathcal{M}_{ij} \rangle = \textsc{Update}(\vx,\mathcal{M}_{ij},\mathbb{F}(g_7, \Phi_i, \Phi_j, \Phi_{w_k}, \Phi_{w_r}))$

\EndIf
	    \EndFor 
	    
	\EndFor \label{algline:Si-three-paths-centered-at-i-endfor}
	
\vspace{-0.5mm}
	\For{{\bf each} $w_k \in S_j$}
	\label{algline:Sj-three-paths-centered-at-j}
	    
	    \For{$w_r \in \Gamma_{w_k} \!\setminus \{i,j\}$}

\If{$w_r \not\in (\Gamma_i \cup \Gamma_j)$} \Comment{\emph{typed 4-path-edge} orbit}
\State $\langle \vx, \mathcal{M}_{ij} \rangle = \textsc{Update}(\vx,\mathcal{M}_{ij},\mathbb{F}(g_3, \Phi_i, \Phi_j, \Phi_{w_k}, \Phi_{w_r}))$

\ElsIf{$w_r \in S_j \wedge w_r \!\leq\! w_k$\!}
\Comment{\emph{typed tailed-tri (tail orbit)}} 
\label{algline:tailed-tri-edge-tail-edge-orbit-Sj}
\State $\langle \vx, \mathcal{M}_{ij} \rangle = \textsc{Update}(\vx,\mathcal{M}_{ij},\mathbb{F}(g_7, \Phi_i, \Phi_j, \Phi_{w_k}, \Phi_{w_r}))$

\ElsIf{$w_r \in S_i$} \label{algline:typed-4-cycle} \Comment{\emph{typed 4-cycle}}

\State $\langle \vx, \mathcal{M}_{ij} \rangle = \textsc{Update}(\vx,\mathcal{M}_{ij},\mathbb{F}(g_6, \Phi_i, \Phi_j, \Phi_{w_k}, \Phi_{w_r}))$
\EndIf	
	    \EndFor
	    
\EndFor \label{algline:Sj-three-paths-centered-at-j-endfor}

\vspace{-0.8mm}
\State {\bf return} set of typed graphlets $\mathcal{M}_{ij}$ between $i$ and $j$ and counts $\vx$	
\medskip
\end{algorithmic}
\end{spacing}
\vspace{-1.mm}
\end{algorithm}
\end{center}
\vspace{-4.mm}
\end{figure}
}

{
\algblockdefx[parallel]{parfor}{endpar}[1][]{$\textbf{parallel for}$ #1 $\textbf{do}$}{$\textbf{end parallel}$}
\algrenewcommand{\alglinenumber}[1]{\fontsize{7.5}{8}\selectfont#1}
\begin{figure}[h!]
\vspace{-6mm}
\begin{center}
\begin{algorithm}[H]
\caption{\,
Typed \emph{Triangle}-based Graphlets
}
\label{alg:typed-triangle-based-motifs-exact}
\begin{spacing}{1.10}
\small
\begin{algorithmic}[1]
\Require 
a graph $G=(V,E,\Phi,\xi)$, 
an edge $(i,j)$, 
set of nodes $T_{ij}$ that form triangles with $i$ and $j$,
sets of nodes $S_i$ and $S_j$ that form 3-paths centered at $i$ and $j$, respectively, 
a typed graphlet count vector $\vx$ for $(i,j)$, 
and 
set $\mathcal{M}_{ij}$ of unique typed graphlets for $i$ and $j$.
\smallskip

	\For{{\bf each} $w_k \in T_{ij}$}
	\label{algline:T-triangles} 

	    \For{$w_r \in \Gamma_{w_k} \!\setminus \{i,j\}$}

\If{$w_r \!\in\! T_{ij} \wedge w_r \!\leq\! w_k$} 
\Comment{\emph{typed 4-clique}}
\label{algline:typed-4-clique}
\State $\langle \vx, \mathcal{M}_{ij} \rangle \!=\! \textsc{Update}(\vx,\mathcal{M}_{ij},\mathbb{F}(g_{12}, \Phi_i, \Phi_j, \Phi_{w_k}, \Phi_{w_r}\!))$

\ElsIf{\!$w_r \! \in (S_i \cup S_j$)} 
\Comment{\emph{typed chord-cycle-edge} orbit}
\label{algline:typed-4-chordal-cycle-edge-orbit}
\State $\langle \vx, \mathcal{M}_{ij} \rangle \!=\! \textsc{Update}(\vx,\mathcal{M}_{ij},\mathbb{F}(g_{10}, \Phi_i, \Phi_j, \Phi_{w_k}, \Phi_{w_r}))$

\ElsIf{$w_r \not\in (\Gamma_i \cup \Gamma_j)$} 
\Comment{\emph{typed tailed-tri-center} orbit}
\label{algline:typed-4-tailed-tri-center-orbit}
\State $\langle \vx, \mathcal{M}_{ij} \rangle \!=\! \textsc{Update}(\vx,\mathcal{M}_{ij},\mathbb{F}(g_{8}, \Phi_i, \Phi_j, \Phi_{w_k}, \Phi_{w_r}))$
\EndIf
\EndFor
	\EndFor
\vspace{-0.8mm}
\State {\bf return} set of typed graphlets $\mathcal{M}_{ij}$ between $i$ and $j$ and counts $\vx$
\medskip
\end{algorithmic}
\end{spacing}
\vspace{-1.mm}
\end{algorithm}
\end{center}
\vspace{-3.mm}
\end{figure}
}

\subsection{Counting 4-Node \emph{Typed} Graphlets} \label{sec:algorithm-4-node-typed-motifs}
\noindent
To derive $k$-node typed graphlets, the framework leverages the lower-order ($k\!-\!1$)-node \emph{typed graphlets}.
Therefore, $4$-node typed graphlets are derived by leveraging the \emph{typed} sets $T_{ij}^t = \Gamma_i^t \cup \Gamma_j^t$, 
$S_j^t = \Gamma_j^t \setminus T_{ij}^t$,
and $S_i^t = \Gamma_i^t \setminus T_{ij}^t$ 
(for $t \in \{1,\ldots,L\}$)
computed from the lower-order $3$-node typed graphlets along with the set $I^t$ of non-adjacent nodes of type $t$ $\wrt$ $(i,j) \in E$ defined formally as follows: 
\begin{align}\label{eq:indep-node-set}
I^t &= V^t \setminus (\Gamma_i^t \cup \Gamma_j^t) \\
&= V^t \setminus (T_{ij}^t \cup S_i^t \cup S_j^t \cup \{i,j\}). \nonumber
\end{align}\noindent
where $V^t \subseteq V$ is the set of nodes in $V$ of type $t$.
\begin{Property} \label{prop:set-I}
\begin{equation}
|V^t| = |I^t| + |\Gamma_i^t| + |\Gamma_j^t|
\end{equation}\noindent
\end{Property}\noindent
The proof is straightforward by Eq.~\ref{eq:indep-node-set} and applying the principle of inclusion-exclusion.

\definecolor{gray}{RGB}{100,100,100}
\definecolor{darkgray}{RGB}{150,150,150}

\definecolor{theblue}{RGB}{0, 20, 159} 
\definecolor{thelightblue}{RGB}{0,191,255}
\definecolor{thelightred}{RGB}{255,191,0}
\definecolor{thecrimson}{RGB}{	153, 0, 0}

\makeatletter
\global\let\tikz@ensure@dollar@catcode=\relax
\makeatother
\tikzstyle{every node}=[font=\large,line width=1.5pt]
\begin{figure}[t!]
\centering

\tikzstyle{background-type1}=[circle,
fill=gray!25, 
inner sep=0.3cm,
rounded corners=4mm]

\tikzstyle{background-type2}=[circle,
fill=gray!25,
inner sep=0.3cm,
rounded corners=4mm]

\tikzstyle{background-type3}=[rectangle,
fill=green!15,
inner sep=0.3cm,
rounded corners=4mm]

\tikzstyle{background-type4}=[rectangle,
fill=gray!25,
inner sep=0.3cm,
rounded corners=4mm]                                                

\tikzstyle{background}=[rectangle,
fill=purple!10,
inner sep=0.3cm,
rounded corners=4mm]

\tikzstyle{background-white}=[rectangle,
fill=white,
inner sep=0cm,
rounded corners=0mm]                                                

\begin{minipage}[b]{0.38\linewidth}
\centering

\subfigure[typed $T_{ij}$ sets]{
\scalebox{0.36}{
\centering
\begin{tikzpicture}[-,>=latex,auto,node distance=2.2cm,thick,
main node/.style={circle,draw=black,fill=black,draw,font=\sffamily\Huge\bfseries,text=white,minimum width=1.03cm
},
type0 node/.style={circle,draw=black,fill=white,draw,font=\sffamily\Huge\bfseries,text=black,minimum width=0.9cm,minimum height=0.9cm},
type1 node/.style={rectangle,draw=black,fill=white,draw,text=black,font=\sffamily\Huge\bfseries,minimum height=0.9cm},
type2 node/.style={diamond,draw=black,fill=white,draw,font=\sffamily\Huge\bfseries,text=black,minimum width=0.9cm,minimum height=0.9cm},
type3 node/.style={circle,draw=gray!80,fill=white,draw,font=\sffamily\Huge\bfseries,text=black},
type4 node/.style={circle,draw=green!80,fill=white,draw,font=\sffamily\Huge\bfseries,text=black},
white node/.style={circle,draw=white,fill=white,draw,font=\sffamily\Huge\bfseries,text=white,minimum width=0.00009cm},
green node/.style={circle,draw=gray!25,fill=gray!25,draw,font=\sffamily\Huge\bfseries,text=black,minimum width=0.1cm},
blue node/.style={circle,draw=gray!25,fill=gray!25,draw,font=\sffamily\Huge\bfseries,text=black,minimum width=0.1cm},
text=black,minimum width=0.9cm,font=\sffamily\Huge\bfseries]

\node[main node] (1) {$\mathbf{i}$};
\node[type1 node] (3) [above right of=1] {$\mathbf{k}$};
\node[white node] (33) [below of=3] {};
\node[main node] (2) [below right of=3] {$\mathbf{j}$};
\node[type1 node] (4) [above of=3] {$\mathbf{r}$};
\node[white node] (5) [above of=4, below=15pt] {};
\node[type0 node] (6) [above of=5, above=20pt] {$\mathbf{q}$};
\node[type0 node] (7) [above of=6] {$\mathbf{p}$};
\node[green node] (10) [above of=3,below=18pt] {$\vdots$};
\node[blue node] (11) [above of=6,below=16pt] {$\vdots$};

\tikzstyle{LabelStyle}=[below=3pt]
\path[every node/.style={font=\sffamily}] 
(1) edge [line width=1.0mm, left] node [above left] {} (2) 
	(1)  edge [line width=1.0mm, right] node[below right] {} (3)
(2) edge [line width=1.0mm, left] node[below left] {} (3)
(1) edge [line width=1.0mm, left] node[below left] {} (4)
	(2) edge [line width=1.0mm,right] node[above right] {} (4)
	    (1) edge [bend left, line width=1.0mm, left] node[below left] {} (6)
	(2) edge [bend right, line width=1.0mm,right] node[above right] {} (6)
	    (1) edge [bend left, line width=1.0mm, left] node[below left] {} (7)
	(2) edge [bend right,line width=1.0mm,right] node[above right] {} (7);
	
\begin{pgfonlayer}{background}
\node [background-type1,
fit=(3) (4), font=\sffamily\Huge\bfseries,label=left:\Huge 
\scalebox{1.1}{$T_{ij}^{1}\;\;\;$}
] {};
\node [background-type2, fit=(6) (7),font=\sffamily\Huge\bfseries,label=left:\Huge 
\scalebox{1.1}{$\Huge T_{ij}^{t}\;\;\;$}
] {};
\node [background-white, fit=(5),font=\sffamily\Huge\bfseries,label=mid:\Huge {
\fontsize{32}{32}\selectfont
\rotatebox{90}{$\mathrm{...}$}
}] {};
\end{pgfonlayer}
\end{tikzpicture}
}
\label{fig:typed-triangle-nodes}
}
\vspace{3mm}
\end{minipage}
\begin{minipage}[b]{0.50\linewidth}
\centering
\subfigure[typed $S_j$ sets]{
\scalebox{0.36}{
\centering
\begin{tikzpicture}[-,>=latex,auto,node distance=2.2cm,thick,
main node/.style={circle,draw=black,fill=black,draw,font=\sffamily\Huge\bfseries,text=white,minimum width=1.03cm},
type0 node/.style={circle,draw=black,fill=white,draw,font=\sffamily\Huge\bfseries,text=black,minimum width=0.9cm,minimum height=0.9cm},
type1 node/.style={rectangle,draw=black,fill=white,draw,text=black,font=\sffamily\Huge\bfseries,minimum height=0.9cm},
type2 node/.style={diamond,draw=black,fill=white,draw,font=\sffamily\Huge\bfseries,text=black,minimum width=0.9cm,minimum height=0.9cm},
type3 node/.style={circle,draw=gray!80,fill=white,draw,font=\sffamily\Huge\bfseries,text=black},
type4 node/.style={circle,draw=green!80,fill=white,draw,font=\sffamily\Huge\bfseries,text=black},
white node/.style={circle,draw=white,fill=white,draw,font=\sffamily\Huge\bfseries,text=white,minimum width=0.1cm},
green node/.style={circle,draw=gray!25,fill=gray!25,draw,font=\sffamily\Huge\bfseries,text=black,minimum width=0.1cm},
blue node/.style={circle,draw=gray!25,fill=gray!25,draw,font=\sffamily\Huge\bfseries,text=black,minimum width=0.1cm},
text=black,minimum width=0.9cm,font=\sffamily\Huge\bfseries]

\node[main node] (1) {$\mathbf{i}$};
\node[type1 node] (3) [above right of=1, above=40pt] {$\mathbf{k}$};
\node[main node] (2) [right of=1,right=15pt] {$\mathbf{j}$};
\node[type1 node] (4) [right of=3] {$\mathbf{r}$};
\node[white node] (5) [right of=4, left=5pt] {};
\node[type0 node] (6) [right of=5] {$\mathbf{p}$};
\node[type0 node] (7) [right of=6] {$\mathbf{q}$};
\node[green node] (10) [right of=3,left=20pt] {$...$};
\node[blue node] (11) [right of=6,left=20pt] {$...$};

\tikzstyle{LabelStyle}=[below=3pt]
\path[every node/.style={font=\sffamily}] 
(1) edge [line width=1.0mm, left] node [above left] {} (2) 
(2) edge [line width=1.0mm, left] node[below left] {} (3)
	(2) edge [line width=1.0mm,right] node[above right] {} (4)
	(2) edge [line width=1.0mm,right] node[above right] {} (6)
	(2) edge [line width=1.0mm,right] node[above right] {} (7);
	
\begin{pgfonlayer}{background}
\node [background-type1,
fit=(3) (4), font=\sffamily\Huge\bfseries,label=above:\Huge $S_{j}^{1}\;\;\;\;$] {};
\node [background-type2, fit=(6) (7),font=\sffamily\Huge\bfseries,label=above:\Huge $\Huge S_{j}^{t}\;\;\;\;$] {};
\node [background-white, fit=(5),font=\sffamily\Huge\bfseries,label=mid:\Huge {\fontsize{32}{32}\selectfont $\mathrm{...}$}] {};
\end{pgfonlayer}
\end{tikzpicture}
}
\label{fig:typed-3-path-nodes-centered-at-j}
}
\hspace{4mm}
\subfigure[typed $S_i$ sets]{
\scalebox{0.36}{
\centering
\begin{tikzpicture}[-,>=latex,auto,node distance=2.2cm,thick,
main node/.style={circle,draw=black,fill=black,draw,font=\sffamily\Huge\bfseries,text=white,minimum width=1.03cm},
type0 node/.style={circle,draw=black,fill=white,draw,font=\sffamily\Huge\bfseries,text=black,minimum width=0.9cm,minimum height=0.9cm},
type1 node/.style={rectangle,draw=black,fill=white,draw,text=black,font=\sffamily\Huge\bfseries,minimum height=0.9cm},
type2 node/.style={diamond,draw=black,fill=white,draw,font=\sffamily\Huge\bfseries,text=black,minimum width=0.9cm,minimum height=1.2cm},
type3 node/.style={circle,draw=gray!80,fill=white,draw,font=\sffamily\Huge\bfseries,text=black},
type4 node/.style={circle,draw=green!80,fill=white,draw,font=\sffamily\Huge\bfseries,text=black},
white node/.style={circle,draw=white,fill=white,draw,font=\sffamily\Huge\bfseries,text=white,minimum width=0.1cm},
green node/.style={circle,draw=gray!25,fill=gray!25,draw,font=\sffamily\Huge\bfseries,text=black,minimum width=0.1cm},
blue node/.style={circle,draw=gray!25,fill=gray!25,draw,font=\sffamily\Huge\bfseries,text=black,minimum width=0.1cm},
text=black,minimum width=0.9cm,font=\sffamily\Huge\bfseries]

\node[main node] (1) {$\mathbf{i}$};
\node[type0 node] (3) [above left of=1, above right=40pt] {$\mathbf{q}$};
\node[main node] (2) [right of=1,right=15pt] {$\mathbf{j}$};
\node[type0 node] (4) [left of=3] {$\mathbf{p}$};
\node[white node] (5) [left of=4, right=5pt] {};
\node[type1 node] (6) [left of=5] {$\mathbf{r}$};
\node[type1 node] (7) [left of=6] {$\mathbf{k}$};
\node[green node] (10) [left of=3,right=20pt] {$...$};
\node[blue node] (11) [left of=6,right=20pt] {$...$};

\tikzstyle{LabelStyle}=[below=3pt]
\path[every node/.style={font=\sffamily}] 
(1) edge [line width=1.0mm, left] node [above left] {} (2) 
(1) edge [line width=1.0mm, left] node[below left] {} (3)
	(1) edge [line width=1.0mm,right] node[above right] {} (4)
	(1) edge [line width=1.0mm,right] node[above right] {} (6)
	(1) edge [line width=1.0mm,right] node[above right] {} (7);
	
\begin{pgfonlayer}{background}
\node [background-type1,
fit=(3) (4), font=\sffamily\Huge\bfseries,label=above:\Huge $S_{i}^{t}\;\;\;\;$] {};
\node [background-type2, fit=(6) (7),font=\sffamily\Huge\bfseries,label=above:\Huge $\Huge S_{i}^{1}\;\;\;\;$] {};
\node [background-white, fit=(5),font=\sffamily\Huge\bfseries,label=mid:\Huge {\fontsize{32}{32}\selectfont $\mathrm{...}\;\;$}] {};
\end{pgfonlayer}
\end{tikzpicture}
}
\label{fig:typed-3-path-nodes-centered-at-i}
}
\hspace{4mm}
\end{minipage}
\hfill

\vspace{-3mm}
\caption{Typed lower-order sets used to derive many higher-order graphlets in constant time.
Note node $i$ and $j$ can be of arbitrary types.
}
\label{fig:typed-lower-order-sets}
\end{figure}
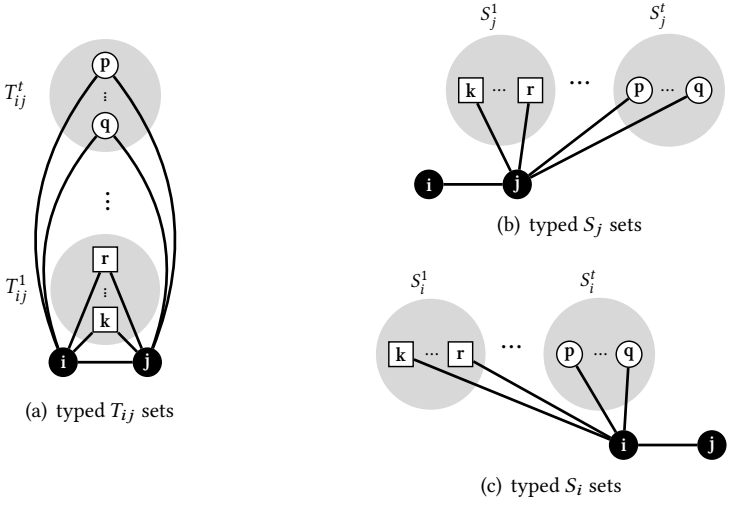

\subsubsection{A General Principle for Typed Graphlet Counting}
\label{sec:general-principle}
\noindent
We now introduce a general typed graphlet formulation.
Let $N^{e,\vt}_{P,Q}$ denote the number of distinct typed 4-node graphlets of $H$ with the type vector $\vt$ that contain edge $(i,j) \in E$ 
and have properties $P \in \{S_i^t, S_j^t, T_{ij}^t, I^t\}$ and $Q \in \{S_i^{t^{\prime}}, S_j^{t^{\prime}}, T_{ij}^{t^{\prime}}, I^{t^{\prime}}\}$ for any $t, t^{\prime} \in \{1,\ldots,L\}$ defined as:
\begin{align} \label{eq:general-typed-graphlet-formulation-without-edge-constraint} 
N^{e,\vt}_{P,Q} = \Big|\Big\{ \{i,j,w_k,w_r\} \,\big|\,
& w_k \!\in P \wedge w_r \!\in Q \wedge \\ \nonumber
& w_r \not= w_k \wedge \vt = \big[ \phi_i \;\; \phi_j \;\; \phi_{w_k} \; \phi_{w_r}\big]
\Big\}\Big|
\end{align}\noindent
Now let $e'$ denote the event $(w_k,w_r) \in E$, and let $[e']$ be the Iverson bracket that is $1$ when $(w_k,w_r) \in E$ and $0$ otherwise. 
Then, $N^{e}_{P,Q,[e']}$ denotes the number of all possible typed $4$-node graphlets conditional on $(w_k,w_r) \in E$. 
\begin{align} \label{eq:general-typed-graphlet-formulation} 
N^{e,\vt}_{P,Q,[e']} = \Big|\Big\{ \{i,j,w_k,w_r\} \,\big|\,
& w_k \!\in P \wedge w_r \!\in Q \wedge \\ \nonumber
& [e'] \wedge w_r \not= w_k \wedge \\ \nonumber
& \vt = \big[ \phi_i \;\; \phi_j \;\; \phi_{w_k} \; \phi_{w_r}\big]
\Big\}\Big|
\end{align}\noindent

\begin{thm}[General Principle for Typed Graphlet Counting]
\label{thm:general-prin-for-typed-graphlets}
Given a graph $G$, for any edge $e = (i,j)$ in $G$, for some type vector 
$\vt =\big[ \phi_i \;\; \phi_j \;\; \phi_{w_k} \; \phi_{w_r}\big]$, where $w_k \in P$, $w_r \in Q$, $\phi_{w_k} = t$, and $\phi_{w_r} = t^{\prime}$, then the number of all $4$-node typed graphlets $\{i,j,w_k,w_r\}$ satisfies the general principle,
\begin{align} \label{eq:general-typed-graphlet-formulation}
N^{e,\vt}_{P,Q,0} = N^{e,\vt}_{P,Q} - N^{e,\vt}_{P,Q,1} 
\end{align}\noindent
\end{thm}

\begin{proof}
Assume there is a typed induced subgraph $J \subset V$ such that $J = \{i,j,w_k,w_r\}$ is incident to the edge of interest $e =(i,j)$ and $J$ is associated with a type vector $\vt = \big[ \phi_i \;\; \phi_j \;\; \phi_{w_k} \; \phi_{w_r}\big]$ where $\phi_{w_k} = t$ and $\phi_{w_r} = t'$.
Suppose $(w_k,w_r) \in E$.
Then by definition, we have $J$ being counted once in the term $N^{e, \vt}_{P,Q,1}$ and once in the term $N^{e, \vt}_{P,Q}$ of Eq.~\ref{eq:general-typed-graphlet-formulation}. 
By the principle of inclusion-exclusion~\cite{stanley1986enumerative}, the total contribution of the typed subgraph $J$ with type vector $\vt$ to $N^{e, \vt}_{P,Q,0}$ is zero.
\end{proof}

Notice that $N_{P,Q}^{e,\vt}$ does not indicate whether $t=t^{\prime}$ or $t\not=t^{\prime}$.
For clarity, we often use $N_{P,Q}^{e,t=t^{\prime}}$ and $N_{P,Q}^{e,t\not=t^{\prime}}$ to denote this explicitly.
Theorem~\ref{thm:general-prin-for-typed-graphlets} shows it is sufficient to compute only two of the three quantities $\{N_{P,Q}^{e,t=t^{\prime}}, N_{P,Q,0}^{e,t=t^{\prime}},  N_{P,Q,1}^{e,t=t^{\prime}}\}$.
For instance, it is enough to compute $N_{P,Q}^{e,t=t^{\prime}}$ and $N_{P,Q,1}^{e,t=t^{\prime}}$, and then $N_{P,Q,1}^{e,t=t^{\prime}}$ can be derived in $o(1)$ constant time 
\begin{equation}\label{eq:t=tprime}
N_{P,Q,0}^{e,t=t^{\prime}} = N_{P,Q}^{e,t=t^{\prime}} - N_{P,Q,1}^{e,t=t^{\prime}}
\end{equation}
It is straightforward to see that for $t\not=t^{\prime}$ it also holds: 
\begin{equation}
N_{P,Q,0}^{e,t\not=t^{\prime}} = N_{P,Q}^{e,t\not=t^{\prime}} - N_{P,Q,1}^{e,t\not=t^{\prime}}
\end{equation}
This implies that selecting the two least computationally expensive quantities offers an obvious computational advantage.

We now show two fundamental properties that simplify the theory and discussion in Section~\ref{sec:combinatorial-relationships}.
Property~\ref{prop:P-equals-Q}
applies to $N_{P,Q}^{e,t=t^{\prime}}$ only, 
whereas Property~\ref{prop:P-not-equal-Q-mutually-exclusive} applies to either $N_{P,Q}^{e,t=t^{\prime}}$ or $N_{P,Q}^{e,t\not=t^{\prime}}$.
\begin{Property} \label{prop:P-equals-Q}
Let $e=(i,j)\in E$ and $N_{P,Q}^{e,t=t^{\prime}}$ denote the number of 4-node typed graphlet orbits $\{i,j,w_k,w_r\}$ such that $w_k$ and $w_r$ satisfy property $P \in \{S_i^t, S_j^t, T_{ij}^t, I^t\}$ and $Q \in \{S_i^{t^{\prime}}, S_j^{t^{\prime}}, T_{ij}^{t^{\prime}}, I^{t^{\prime}}\}$ for any $t=t^{\prime} \in \{1,\ldots,L\}$, respectively.
We say $N_{P,Q}^{e,t=t^{\prime}}$ is an unrestricted count since $N_{P,Q}^{e,t=t^{\prime}} = N_{P,Q,0}^{e,t=t^{\prime}} + N_{P,Q,1}^{e,t=t^{\prime}}$.
If $P=Q$, then 
$t=t^{\prime}$. Therefore, 
\begin{equation} \label{eq:P-equals-Q}
N_{P,Q}^{e,t=t^{\prime}} = {|P| \choose 2} = \frac{|P|(|P|-1)}{2} = \frac{|Q|(|Q|-1)}{2}
\end{equation}
\end{Property}
Clearly, Property~\ref{prop:P-equals-Q} holds iff $t=t^{\prime}$ \emph{and} $P=Q$.
Suppose $t\not=t^{\prime}$, then $P \cap Q = \emptyset$ by definition, hence, $P\not=Q$.
In other words, $P=Q$ implies $t=t^{\prime}$.
Assuming $t=t^{\prime}$, this property is useful for deriving the count of typed 4-stars and typed chordal-cycles (center orbit) in $o(1)$ time as shown later in Section~\ref{sec:combinatorial-relationships}.
For instance, suppose $t=t^{\prime}$, then $P=S_i^{t}$ and $Q=S_j^{t}$, and therefore the number of typed 4-stars $f_{ij}(g_5,\vt)$ with type vector $\vt=[\,\phi_i \; \phi_j \; t \; t\,]$ that occur between node $i$ and $j$ is 
\[
f_{ij}(g_5,\vt) = \frac{1}{2} \Big[|S_i^{t}|(|S_i^{t}|-1) + |S_j^{t}|(|S_j^{t}|-1)\Big] - f_{ij}(g_7,\vt)
\]
\noindent
where $f_{ij}(g_7,\vt)$ is the typed tailed-triangle (tail-edge orbit) count.

\begin{Property} \label{prop:P-not-equal-Q-mutually-exclusive}
If $P\not=Q$, then $P\cap Q = \emptyset$.
Hence, $P$ and $Q$ are mutually exclusive.
This implies
\begin{equation} \label{eq:mutual-exclusion}
N_{P,Q}^{e,t=t^{\prime}} \!= |P| \cdot |Q|\;\quad \text{ and }\quad\; N_{P,Q}^{e,t\not=t^{\prime}} \!= |P| \cdot |Q|
\end{equation}
\end{Property}
Hence, since $P\not= Q$, then the above clearly holds for both $t=t^{\prime}$ and $t \not= t^{\prime}$.
Notice that in the untyped case, if $P=S_i$ and $Q=S_i$, then $P\cap Q = P$.
However, if we consider types and set $P=S_i^t$ and $Q=S_i^{t^{\prime}}\!\!$, then $P \cap Q = \emptyset$ iff $t\not= t^{\prime}$.

\begin{Property}\label{prop:t-equals-tprime-implies-P-equals-Q}
$\forall t,t^{\prime}$,\; s.t.\;
$t\not=t^{\prime} \Rightarrow P\not=Q$
\end{Property}
The above is straightforward.
The converse is not true, that is, if $P\not=Q$, then $t\not=t^{\prime}$ does not necessarily hold.
Assume $t=t^{\prime}$, and let $P=T_{ij}^{t}$ and $Q=S_{i}^{t^{\prime}}$, then clearly $P\not=Q$ despite $t=t^{\prime}$.
Suppose $P=T_{ij}^{t}$ and $Q=T_{ij}^{t^{\prime}}$. 
If $P=Q$, then $T_{ij}^{t}=T_{ij}^{t^{\prime}}$ and therefore $t=t^{\prime}$ must hold.
Otherwise, if $t\not=t^{\prime}$ then $P\not=Q$.

\begin{table*}[t!]
\vspace{-2mm}
\caption{
Typed graphlet orbit equations. 
All typed graphlet orbits with $4$-nodes are formulated with respect to the typed node sets $\{S_i^t,S_j^t,T_{ij}^t,I^t\}$ and $\{S_i^{t^{\prime}},S_j^{t^{\prime}},T_{ij}^{t^{\prime}},I^{t^{\prime}}\}$ for $t, t^{\prime}=1,\ldots,L$ derived from the typed $3$-node graphlets.
Recall $T_{ij}^t = \Gamma_i^t \cap \Gamma_j^t$, $S_j^t = \Gamma_j^t \setminus T_{ij}^t$, $S_i^t= \Gamma_i^t \setminus T_{ij}^t$, and 
$I^t = V^t \setminus (\Gamma_i^t \cup \Gamma_j^t) = V^t \setminus (T_{ij}^t \cup S_i^t \cup S_j^t \cup \{i,j\})$ 
where $V^t$ is the set of nodes in $V$ of type $t$.
In all cases, $w_r \not= w_k$. 
}
\vspace{-3.5mm}
\centering 
\fontsize{7}{7.5}\selectfont
\setlength{\tabcolsep}{2pt} 
\renewcommand{\arraystretch}{1.0} 
\label{table:typed-graphlet-equations}
\hspace*{-2.5mm}
\begin{tabularx}{1.00\linewidth}{@{}rl  H H l@{} XH@{}} 
\toprule
\textsc{Typed Graphlet} 
& 
\textsc{Orbit} & \!\!\!$|E(H)|$ &$\rho(H)$ &
$f_{ij}(H, \vt) =$ &
$\Big|\Big\{ \{i,j,w_k,w_r\} \,\big|\,
w_k \!\in P \wedge 
w_r \!\in Q \wedge 
\mathbb{I}\{(w_k, w_r) \!\in E\} \wedge
w_r \not= w_k \wedge
\vt = \big[ \phi_i \;\; \phi_j \;\; \phi_{w_k} \; \phi_{w_r}\big]
\Big\}\Big|
$
\\
\midrule

\textbf{4-path} & \text{edge} & 3 & 0.50 &
$
f_{ij}(g_3, \vt) =$ & $\Big|\Big\{ \{i,j,w_k,w_r\} \,\big|\,
\big(w_k \!\in S_i^t \wedge 
w_r \!\in I^{t^{\prime}}\big)
\vee
\big(w_k \!\in S_j^t \wedge 
w_r \!\in I^{t^{\prime}}\big)
\wedge 
(w_k, w_r) \!\in E 
\wedge
\vt = \big[ \phi_i \;\; \phi_j \; \phi_{w_k} \; \phi_{w_r}\!\big]
\Big\}\Big|
$
\\

& \text{center} & 3 & 0.50 &
$
f_{ij}(g_4, \vt) =$ &
$\Big|\Big\{ \{i,j,w_k,w_r\} \,\big|\,
w_k \!\in S_j^t \wedge 
w_r \!\in S_i^{t^{\prime}} \wedge  (w_k,\!w_r) \!\not\in\! E
\wedge
\vt = \big[ \phi_i \;\; \phi_j \;\; \phi_{w_k} \; \phi_{w_r}\big]
\Big\}\Big|
$
\\
\midrule

\textbf{4-star} &  & 3 & 0.50 &
$
f_{ij}(g_5, \vt) =$ &
$ \Big|\Big\{ \{i,j,w_k,w_r\} \,\big|\,
(w_k \!\in \!S_i^t \wedge 
w_r \!\in \!S_i^{t^{\prime}}\!) 
\vee 
(w_k \!\in \!S_j^t \wedge 
w_r \!\in \!S_j^{t^{\prime}}\!)
\wedge 
w_r \!\not= \!w_k \wedge 
(w_k, \!w_r) \!\not\in E 
\wedge
\vt = \big[ \phi_i \;\; \phi_j \; \phi_{w_k} \; \phi_{w_r}\!\big]
\Big\}\Big|
$
\\
\midrule

\textbf{4-cycle} & & 4 & 0.67 &
$
f_{ij}(g_6, \vt) =$ &
$ \Big|\Big\{ \{i,j,w_k,w_r\} \,\big|\,
w_k \!\in S_j^t \wedge 
w_r \!\in S_i^{t^{\prime}} \wedge 
(w_k, w_r) \!\in E 
\wedge
\vt = \big[ \phi_i \;\; \phi_j \;\; \phi_{w_k} \; \phi_{w_r}\big]
\Big\}\Big|
$
\\
\midrule

\textbf{tailed-triangle} 
& \text{tail-edge} 
& 4 & 0.67 &
$f_{ij}(g_7, \vt) =$ &
$\Big|\Big\{ \{i,j,w_k,w_r\} \,\big|\,
(w_k \!\in \!S_i^t \wedge 
w_r \!\in \!S_i^{t^{\prime}}\!) 
\vee 
(w_k \!\in \!S_j^t \wedge 
w_r \!\in \!S_j^{t^{\prime}}\!)
\wedge 
w_r \!\not= \!w_k \wedge 
(w_k, \!w_r) \!\in E 
\wedge
\vt = 
\big[ \phi_i \;\; \phi_j \; \phi_{w_k} \; \phi_{w_r}\!\big]
\Big\}\Big|
$
\\

& \text{center} 
& 4 & 0.67 &
$
f_{ij}(g_8, \vt) =$ & $\Big|\Big\{ \{i,j,w_k,w_r\} \,\big|\,
w_k \!\in T_{ij}^t \wedge 
w_r \!\in I^{t^{\prime}} \wedge
(w_k, w_r) \!\in E 
\wedge
\vt = \big[ \phi_i \;\; \phi_j \;\; \phi_{w_k} \; \phi_{w_r}\big]
\Big\}\Big|
$
\\

& \text{tri-edge} 
& 4 & 0.67 &
$
f_{ij}(g_9, \vt) =$ & $\Big|\Big\{ \{i,j,w_k,w_r\} \,\big|\,
w_k \!\in T_{ij}^t \wedge 
w_r \!\!\in\! \big(S_i^{t^{\prime}} \!\cup S_j^{t^{\prime}}\!\big) \wedge  
(w_k, w_r) \!\not\in E 
\wedge
\vt = \big[ \phi_i \;\; \phi_j \;\; \phi_{w_k} \; \phi_{w_r}\big]
\Big\}\Big|
$
\\
\midrule

\textbf{chordal-cycle} 
& \text{edge} & 5 & 0.83 &
$
f_{ij}(g_{10}, \vt) =$ & 
$\Big|\Big\{ \{i,j,w_k,w_r\} \,\big|\,
w_k \!\!\in\! T_{ij}^t \wedge\, 
w_r \!\!\in\! \big(S_i^{t^{\prime}} \!\cup S_j^{t^{\prime}}\!\big) \wedge 
(w_k, \!w_r) \!\in\! E 
\wedge
\vt \!=\! \big[ \phi_i \; \phi_j \; \phi_{w_k}  \phi_{w_r}\!\big]
\Big\}\Big|
$
\\

& \text{center} & 5 & 0.83 &
$
f_{ij}(g_{11}, \vt) =$ & $\Big|\Big\{ \{i,j,w_k,w_r\} \,\big|\,
w_k \!\in T_{ij}^t \wedge 
w_r \!\in T_{ij}^{t^{\prime}} \wedge 
w_r \not= w_k \wedge 
(w_k, w_r) \!\not\in E 
\wedge
\vt = \big[ \phi_i \;\; \phi_j \;\; \phi_{w_k} \; \phi_{w_r}\big]
\Big\}\Big|
$
\\
\midrule

\textbf{4-clique} 
& & 6 & 1.00 &
$
f_{ij}(g_{12}, \vt) =$ &
$\Big|\Big\{ \{i,j,w_k,w_r\} \,\big|\,
w_k \!\in T_{ij}^t \wedge 
w_r \!\in T_{ij}^{t^{\prime}} \wedge 
w_r \not= w_k \wedge 
(w_k, w_r) \!\in E 
\wedge
\vt = \big[ \phi_i \;\; \phi_j \;\; \phi_{w_k} \; \phi_{w_r}\big]
\Big\}\Big|
$
\\

\bottomrule
\end{tabularx}
\end{table*}

The equations for deriving every typed graphlet orbit of size 4 are provided in Table~\ref{table:typed-graphlet-equations}.
Notice that all typed graphlets with $k$-nodes are formulated with respect to the typed node sets $\{S_i^t,S_j^t,T_{ij}^t,I^t\}$ 
derived from the typed graphlets with ($k\!-\!1$)-nodes.
Hence, higher-order typed graphlets of order $k$ are derived from lower-order ($k\!-\!1$)-node typed graphlets.
We classify typed graphlets as path-based or triangle-based.
Typed path-based graphlets are the typed 4-node graphlets derived from the sets $S_i = \bigcup_t S_i^t$ and $S_j = \bigcup_t S_j^t$ of nodes that form 3-node typed paths centered at node $i$ and $j$, respectively (Algorithm~\ref{alg:typed-path-based-motifs-exact}).
Conversely, typed triangle-based graphlets are the typed 4-node graphlets derived from the set $T_{ij} = \bigcup_t T_{ij}^t$ of nodes that form typed triangles (typed 3-cliques) with node $i$ and $j$ (Algorithm~\ref{alg:typed-triangle-based-motifs-exact}).
Naturally, typed path-based graphlets are the least dense (graphlets with fewest edges) whereas the typed triangle-based graphlets are the most dense.

The typed graphlet equations in Table~\ref{table:typed-graphlet-equations} are mainly used to characterize the typed graphlets, and of course can be used to count them.
However, using those equations to count all typed graphlets is still expensive since some non-negligible work is required to count every typed graphlet.
Instead, we count only a few typed graphlets and use newly discovered combinatorial relationships (see Section~\ref{sec:combinatorial-relationships}) to derive the others directly in $o(1)$ constant time.
Notably, we make no assumptions about the number of types $L$, their distribution among the nodes and edges, or any other additional information.
On the contrary, the framework is extremely general for arbitrary heterogeneous graphs (see Figure~\ref{fig:heter-special-cases} for a number of popular special cases covered by the framework).
In addition, we also avoid a lot of computations by symmetry breaking techniques, and other conditions to avoid unnecessary work.

\begin{table}[h!]
\vspace{2mm}
\centering
\renewcommand{\arraystretch}{1.2} 
\caption{
\textbf{Summary of Typed Graphlets and Position-aware Typed Graphlets.}
Enumerative and combinatorial properties of typed graphlets and position-aware typed graphlets.
With repetition allowed (in making the selection), the number of $K$-node \emph{typed graphlets} and \emph{position-aware typed graphlets} (for a single untyped graphlet) from $L$ distinguishable labels/types is given below along with properties of each.
}
\label{table:typed-graphlets-combin-properties}
\vspace{-3mm}
\scalebox{1.0}{
\begin{tabularx}{0.96\linewidth}{l cc}
\toprule

\vspace{-1.6mm}
& & \bfseries{\scshape{Position-Aware}} \\
\vspace{-1.6mm}
& \bfseries{\scshape{Typed Graphlets}}
& 
\bfseries{\scshape{Typed Graphlets}}
\\
& \text{(Definition~\ref{def:typed-graphlet-instance})}
& \text{(Definition~\ref{def:position-aware-typed-graphlet-instance})}
\\
\midrule

\textbf{With Repetition  \;\;\;\;} & 
\begin{minipage}{4cm}
\begin{equation} \nonumber
\left( \!\binom{L}{K} \!\right) = \binom{L+K-1}{K}
\end{equation}\noindent
\end{minipage}
& 
\begin{minipage}{4cm}
\begin{equation} \nonumber
L^K
\end{equation}\noindent
\end{minipage}
\vspace{1mm}
\\

\BBB\TTT
\vspace{-1.6mm}
& \textsf{Unordered Selections}  
& \textsf{Ordered Selections} 
\\
& \textsf{(Combinations)}  
& \textsf{(Permutations)} 
\\

\bottomrule
\end{tabularx}
}
\end{table}

\subsection{Combinatorial Relationships for \emph{Typed Graphlets}} \label{sec:combinatorial-relationships}
\noindent
Now, we show the existence of combinatorial relationships between the different \emph{typed graphlets} and demonstrate
how they can be leveraged to derive the counts of typed graphlets efficiently.
These combinatorial relationships allow us to derive many \emph{typed graphlets} in $o(1)$ constant time and
play a significant role in the speed/efficiency of the proposed approach (see Section~\ref{sec:exp-comparison}).
Using new combinatorial relationships between lower-order \emph{typed} graphlets, 
we derive all remaining typed graphlet orbits in $o(1)$ constant time via Eq.~\ref{eq:typed-4-path-center-orbit}-\ref{eq:typed-4-chordal-cycle-center-orbit} (See Line~\ref{algline:main-alg-for-type-pair}-\ref{algline:main-alg-derive-remaining-typed-graphlet-orbits-constant-time} in Algorithm~\ref{alg:typed-motifs-exact}).
Since we derive all typed graphlet counts for a given edge $(i,j) \in E$ between node $i$ and $j$, we already have two types $\phi_i$ and $\phi_j$.
Thus, these types are fixed ahead of time.
In the case of 4-node typed graphlets, there are two remaining types that need to be selected.
Notice that for typed graphlet orbits, we must solve $\frac{L(L-1)}{2}+L$ equations in the worst-case.
The counts of all remaining typed graphlets are derived in $o(1)$ constant time using the counts of the lower-order ($k\!-\!1$)-node typed graphlets and a few other counts from the $k$-node typed graphlets.
After deriving the exact count of each remaining graphlet with types $\phi_i$, $\phi_j$, $t$, and $t^{\prime}$ for every $t, t^{\prime} \in \{1, \ldots, L\}$ such that $t\leq t^{\prime}$ (Line~\ref{algline:main-alg-for-type-pair}-\ref{algline:main-alg-derive-remaining-typed-graphlet-orbits-constant-time}), if such count is nonzero, we compute a graphlet hash $\hash = \mathbb{F}(g, \phi_i, \phi_j, t, t^{\prime})$ for graphlet orbit $g$, set $\mathcal{M}_{ij} \leftarrow \mathcal{M}_{ij} \cup \{\hash\}$, and then set the count of that typed graphlet in $\vx_{\hash}$ to the count derived in constant $o(1)$ time.

We now demonstrate the relationship between different typed graphlets and prove the correctness of the equations used to derive a number of typed graphlet counts in $o(1)$ constant time. See Figure~\ref{fig:proof-correctness-examples} for intuition.

\subsubsection{Relationship between typed 4-cycles and 4-paths (center orbit)}
\label{sec:rel-4cycles-and-4path-center-orbits}

\begin{cor}\label{cor:typed-4-cycles}
For any edge $(i,j) \in E$ in $G$ with types $\phi_i$ and $\phi_j$, the number of typed 4-cycles containing edge $(i,j)$ with type vector $\vt = \big[ \phi_i\; \phi_j\; t \; t^{\prime} \big]$ is
$N_{S_i^t\!,S_j^{t}\!,1}^{e,t=t^{\prime}}$ for $t=t^{\prime}$ and $N_{S_i^t\!,S_j^{t^{\prime}}\!,1}^{e,t\not=t^{\prime}} + N_{S_i^{t^{\prime}}\!,S_j^{t},1}^{e,t\not=t^{\prime}}$ otherwise.
\end{cor}
\begin{cor}\label{cor:typed-4-paths}
For any edge $(i,j) \in E$ in $G$ with types $\phi_i$ and $\phi_j$, the number of typed 4-path center orbits containing edge $(i,j)$ with type vector $\vt = \big[ \phi_i\; \phi_j\; t \; t^{\prime} \big]$ is 
$N_{S_i^t\!,S_j^{t}\!,0}^{e,t=t^{\prime}}$ for $t=t^{\prime}$ and $N_{S_i^t\!,S_j^{t^{\prime}}\!,0}^{e,t\not=t^{\prime}} + N_{S_i^{t^{\prime}}\!,S_j^{t}\!,0}^{e,t\not=t^{\prime}}$ otherwise.
\end{cor}

To count the typed 4-path center orbits for a given edge $(i,j) \in E$ with types $\phi_i$ and $\phi_j$, 
we simply select the remaining two types denoted as $t$ and $t^{\prime}$ to obtain the 4-dimensional type vector $\vt = \big[\, \phi_i \;\, \phi_j \;\; t \;\; t^{\prime} \,\big]$ and derive the count directly using
Lemma~\ref{lem:typed-4path-center-orbit-count}.

\begin{lemma}
\label{lem:typed-4path-center-orbit-count}
For any edge $(i,j) \in E$ in $G$ with types $\phi_i$ and $\phi_j$ 
and any type vector $\vt = \big[ \phi_i\; \phi_j\; t \; t^{\prime} \big]$,
the relationship between 
the typed 4-cycle count $f_{ij}(g_{6}, \vt)$ 
and 
the typed 4-path center orbit count $f_{ij}(g_{4},\vt)$ 
with type vector $\vt$ is
\begin{equation} \label{eq:typed-4-path-center-orbit}
f_{ij}(g_4,\vt)  = 
\begin{cases}
(|S_{i}^{t}| \cdot |S_{j}^{t}|) - f_{ij}(g_{6}, \vt) 		 & \text{if } t=t^{\prime}\\[5pt]
(|S_{i}^{t}| \cdot |S_{j}^{t^{\prime}}|) + 	 & \text{otherwise} \\
(|S_{i}^{t^{\prime}}| \cdot |S_{j}^{t}|) -  f_{ij}(g_{6}, \vt) & \\[2pt]
\end{cases}
\end{equation}\noindent
where $f_{ij}(g_{6},\vt)$ is the typed 4-cycle count for edge $(i,j) \in E$ with type vector $\vt$.
\end{lemma}

\begin{proof}

Assume $t=t^{\prime}$. 
From Theorem~\ref{thm:general-prin-for-typed-graphlets}, we have 
\begin{equation}\label{eq:proof-path-cycle-relationship-t-equals-tp}
N_{S_i^t\!,S_j^{t}\!,0}^{e,t=t^{\prime}} = N_{S_i^t\!,S_j^{t}}^{e,t=t^{\prime}} - N_{S_i^t\!,S_j^{t}\!,1}^{e,t=t^{\prime}}
\end{equation}
Since $N_{S_i^t\!,S_j^{t}}^{e,t=t^{\prime}}$ is the number of typed 4-node induced subgraphs containing $e=(i,j)$ such that $w_k \in S_i^t$ and $w_r \in S_j^t$, then 
$N_{S_i^t\!,S_j^{t}}^{e,t=t^{\prime}} = |S_i^t| \cdot |S_j^t|$ by Property~\ref{prop:P-not-equal-Q-mutually-exclusive}.
From Corollary~\ref{cor:typed-4-cycles}, the number of typed 4-cycles that contain $e=(i,j)$ is $f_{ij}(g_{6}, \vt) = N_{S_i^t\!,S_j^{t}\!,1}^{e,t=t^{\prime}}$.
From Corollary~\ref{cor:typed-4-paths}, the count of 
typed 4-paths centered at edge $e=(i,j)$ is $f_{ij}(g_{4},\vt) = N_{S_i^t\!,S_j^{t}\!,0}^{e,t=t^{\prime}}$. 
Therefore, by direct substitution in Eq.~\ref{eq:proof-path-cycle-relationship-t-equals-tp}, we obtain Eq.~\ref{eq:typed-4-path-center-orbit}.

Assume $t\not=t^{\prime}$. 
From Theorem~\ref{thm:general-prin-for-typed-graphlets}, we have 
\begin{equation}
N_{S_i^t\!,S_j^{t^{\prime}}\!\!,0}^{e,t\not=t^{\prime}} = \, N_{S_i^t\!,S_j^{t^{\prime}}}^{e,t\not=t^{\prime}} - N_{S_i^t\!,S_j^{t^{\prime}}\!\!,1}^{e,t\not=t^{\prime}}
\quad\;\;and\;\;\quad
N_{S_i^{t^{\prime}}\!,S_j^{t},0}^{e,t\not=t^{\prime}} = \, N_{S_i^{t^{\prime}}\!,S_j^t}^{e,t\not=t^{\prime}} - N_{S_i^{t^{\prime}}\!,S_j^t,1}^{e,t\not=t^{\prime}}
\end{equation}
It is straightforward to rewrite this as
\begin{equation}
N_{S_i^t\!,S_j^{t^{\prime}}\!\!,0}^{e,t\not=t^{\prime}} + N_{S_i^{t^{\prime}}\!,S_j^{t},0}^{e,t\not=t^{\prime}} = \, 
\Big( N_{S_i^t\!,S_j^{t^{\prime}}}^{e,t\not=t^{\prime}} + N_{S_i^{t^{\prime}}\!,S_j^t}^{e,t\not=t^{\prime}}  \Big) - 
\Big(  N_{S_i^t\!,S_j^{t^{\prime}}\!\!,1}^{e,t\not=t^{\prime}} + N_{S_i^{t^{\prime}}\!,S_j^t,1}^{e,t\not=t^{\prime}} \Big)
\end{equation}
By Property~\ref{prop:P-not-equal-Q-mutually-exclusive}, there are $N_{S_i^t\!,S_j^{t^{\prime}}}^{e,t\not=t^{\prime}} + N_{S_i^{t^{\prime}}\!,S_j^t}^{e,t\not=t^{\prime}} = |S_{i}^{t}| \cdot |S_{j}^{t^{\prime}}| + |S_{i}^{t^{\prime}}| \cdot |S_{j}^{t}|$ typed 4-node induced subgraphs that contain edge $e=(i,j)$ such that $w_k \in S_{i}^{t}$ and $w_r \in S_j^{t^{\prime}}$ \emph{or} $w_k \in S_{i}^{t^{\prime}}$ and $w_r \in S_j^{t}$ where $N_{S_i^t\!,S_j^{t^{\prime}}}^{e,t\not=t^{\prime}} = |S_{i}^{t}| \cdot |S_{j}^{t^{\prime}}|$ and $N_{S_i^{t^{\prime}}\!,S_j^t}^{e,t\not=t^{\prime}} = |S_{i}^{t^{\prime}}| \cdot |S_{j}^{t}|$.
From Corollary~\ref{cor:typed-4-cycles}, the typed 4-cycle count for edge $e$ is $f_{ij}(g_{6}, \vt) = N_{S_i^t\!,S_j^{t^{\prime}}\!,1}^{e,t\not=t^{\prime}} + N_{S_i^{t^{\prime}}\!,S_j^{t},1}^{e,t\not=t^{\prime}}$.
From Corollary~\ref{cor:typed-4-paths}, the count of typed 4-paths centered at edge $e=(i,j)$ is $f_{ij}(g_{4},\vt) = N_{S_i^t\!,S_j^{t^{\prime}}\!,0}^{e,t\not=t^{\prime}} + N_{S_i^{t^{\prime}}\!,S_j^{t}\!,0}^{e,t\not=t^{\prime}}$. 
Therefore, by direct substitution in Eq.~\ref{eq:proof-path-cycle-relationship-t-equals-tp}, we obtain Eq.~\ref{eq:typed-4-path-center-orbit}.

\end{proof}

The only difference between a typed 4-path centered at $(i,j)$ (Corollary~\ref{cor:typed-4-paths}) and a typed 4-cycle (Corollary~\ref{cor:typed-4-cycles}) is whether $(w_k,w_r) \in E$ holds or not.
Clearly, if $(w_k,w_r) \in E$, then we have a 
typed 4-cycle, otherwise $(w_k,w_r) \not\in E$ and it is a typed 4-path centered at $(i,j)$ as shown in Figure~\ref{fig:lemma1}.

\subsubsection{Relationship between typed 4-stars and tailed-triangles (tail-edge orbit)}
\label{sec:rel-4stars-and-tailed-tri-tail-edge-orbits}

\begin{cor}\label{cor:typed-4-stars}
For any edge $(i,j) \in E$ in $G$ with types $\phi_i$ and $\phi_j$, the number of typed 4-stars containing edge $(i,j)$ with type vector $\vt = \big[ \phi_i\; \phi_j\; t \; t^{\prime} \big]$ is
$N_{S_i^t,S_i^{t},0}^{e,t=t^{\prime}} + N_{S_j^t,S_j^{t},0}^{e,t=t^{\prime}}$ for $t=t^{\prime}$ and 
$N_{S_i^t,S_i^{t^{\prime}}\!\!,0}^{e,t\not=t^{\prime}} + N_{S_j^t,S_j^{t^{\prime}}\!\!,0}^{e,t\not=t^{\prime}}$ otherwise.
\end{cor}

\begin{cor}\label{cor:typed-4-tailed-tri-tail-edge-orbits}
For any edge $(i,j) \in E$ in $G$ with types $\phi_i$ and $\phi_j$, 
the number of typed tailed-triangles (tail-edge orbit) containing edge $(i,j)$ with type vector $\vt = \big[ \phi_i\; \phi_j\; t \; t^{\prime} \big]$ is
$N_{S_i^t,S_i^{t},1}^{e,t=t^{\prime}} + N_{S_j^t,S_j^{t},1}^{e,t=t^{\prime}}$ for $t=t^{\prime}$ and 
$N_{S_i^t,S_i^{t^{\prime}}\!\!,1}^{e,t\not=t^{\prime}} + N_{S_j^t,S_j^{t^{\prime}}\!\!,1}^{e,t\not=t^{\prime}}$ otherwise.
\end{cor}

To count the typed 4-stars for a given edge $(i,j) \in E$ with types $\phi_i$ and $\phi_j$, we simply select the remaining two types denoted as $t$ and $t^{\prime}$ to obtain the 4-dimensional type vector $\vt = \big[\, \phi_i \;\, \phi_j \;\; t \;\; t^{\prime} \,\big]$.
We derive the typed 4-star counts with the type vector $\vt$ for edge $(i,j) \in E$ in constant time using Lemma~\ref{lem:typed-4star-count}.

\begin{lemma}
\label{lem:typed-4star-count}
For any edge $(i,j) \in E$ in $G$ with types $\phi_i$ and $\phi_j$ 
and any type vector $\vt = \big[ \phi_i\; \phi_j\; t \; t^{\prime} \big]$,
the relationship between 
the typed 4-star count $f_{ij}(g_{5}, \vt)$ and 
the typed tailed-triangle tail-edge orbit count $f_{ij}(g_{7},\vt)$ 
with type vector $\vt$ is
\begin{equation} \label{eq:typed-4-star}
\!\!\!\!\!\!
f_{ij}(g_5,\vt) = 
\begin{cases}
\mychoose{|S_{i}^{t}|}{2} + \mychoose{|S_{j}^{t}|}{2} - f_{ij}(g_{7}, \vt) 		& \text{if } t=t^{\prime} \\[7pt]
(|S_{i}^{t}| \cdot |S_{i}^{t^{\prime}}|) \;+ 	                                    & \text{otherwise} \\
(|S_{j}^{t}| \cdot |S_{j}^{t^{\prime}}|) - f_{ij}(g_{7}, \vt) & \\[2pt]
\end{cases}
\end{equation}\noindent
where $f_{ij}(g_{7},\vt)$ is the tailed-triangle tail-edge orbit count for edge $(i,j) \in E$ with type vector $\vt$.
\end{lemma}

\begin{proof}
Let $S_{i}^t$ and $S_{i}^{t^{\prime}}$ be the nodes that form typed 3-node stars with $(i,j) \in E$ of type $t$ and $t^{\prime}$ where node $i$ is the star-center node, respectively.
Similarly, $S_{i}^t$ and $S_{i}^{t^{\prime}}$ are the nodes that form typed 3-node stars with $(i,j) \in E$ of type $t$ and $t^{\prime}$ where node $j$ is the star-center node, respectively.

Assume $t=t^{\prime}$. From Theorem~\ref{thm:general-prin-for-typed-graphlets}, we have 
\begin{equation}\label{eq:proof-4-stars-and-tailed-triangles-tail-edge-orbit-relationship-t-equals-tp}
	N_{S_{i}^{t},S_{i}^{t},0}^{e,t=t^{\prime}} + N_{S_{j}^{t},S_{j}^{t},0}^{e,t=t^{\prime}} = \Big(N_{S_{i}^{t},S_{i}^{t}}^{e,t=t^{\prime}} + N_{S_{j}^{t},S_{j}^{t}}^{e,t=t^{\prime}}\Big) - 
	\Big(N_{S_{i}^{t},S_{i}^{t},1}^{e,t=t^{\prime}} + N_{S_{j}^{t},S_{j}^{t},1}^{e,t=t^{\prime}} \Big)
\end{equation}
Therefore, by Property~\ref{prop:P-equals-Q}, there are $N_{S_i^t\!,S_i^{t}}^{e,t=t^{\prime}} + N_{S_j^{t}\!,S_j^{t}}^{e,t=t^{\prime}} = \mychoose{|S_{i}^{t}|}{2} + \mychoose{|S_{j}^{t}|}{2}$ typed 4-node induced subgraphs that contain edge $e=(i,j)$ such that $w_k, w_r \in S_{i}^{t}$ \emph{or} $w_k, w_r \in S_{j}^{t}$ where $N_{S_i^t\!,S_i^{t}}^{e,t=t^{\prime}} = \mychoose{|S_{i}^{t}|}{2}$ and $N_{S_j^t\!,S_j^t}^{e,t=t^{\prime}} = \mychoose{|S_{j}^{t}|}{2}$.
From Corollary~\ref{cor:typed-4-stars}, the number of typed 4-stars that contain edge $e=(i,j)$ is $f_{ij}(g_{5}, \vt) = N_{S_{i}^{t},S_{i}^{t},0}^{e,t=t^{\prime}} + N_{S_{j}^{t},S_{j}^{t},0}^{e,t=t^{\prime}}$.
Similarly, from Corollary~\ref{cor:typed-4-tailed-tri-tail-edge-orbits}, the typed tailed-triangle tail-edge orbit count for edge $e$ is $f_{ij}(g_{7},\vt) = N_{S_{i}^{t},S_{i}^{t},1}^{e,t=t^{\prime}} + N_{S_{j}^{t},S_{j}^{t},1}^{e,t=t^{\prime}}$.
Therefore, by direct substitution in Eq.~\ref{eq:proof-4-stars-and-tailed-triangles-tail-edge-orbit-relationship-t-equals-tp}, we obtain Eq.~\ref{eq:typed-4-star}.

Assume $t\not=t^{\prime}$.
From Theorem~\ref{thm:general-prin-for-typed-graphlets}, we have 
\begin{equation} \label{eq:proof-4-stars-and-tailed-triangles-tail-edge-orbit-relationship-t-not-equals-tp}
N_{S_{i}^{t},S_{i}^{t^{\prime}},0}^{e,t\not=t^{\prime}} + N_{S_{j}^{t},S_{j}^{t^{\prime}},0}^{e,t\not=t^{\prime}} = \Big(N_{S_{i}^{t},S_{i}^{t^{\prime}}}^{e,t\not=t^{\prime}} + N_{S_{j}^{t},S_{j}^{t^{\prime}}}^{e,t\not=t^{\prime}}\Big) - 
	\Big(N_{S_{i}^{t},S_{i}^{t^{\prime}},1}^{e,t\not=t^{\prime}} + N_{S_{j}^{t},S_{j}^{t^{\prime}},1}^{e,t\not=t^{\prime}} \Big)
\end{equation}
Therefore, from Property~\ref{prop:P-not-equal-Q-mutually-exclusive}, there are $N_{S_i^t\!,S_i^{t^{\prime}}}^{e,t\not=t^{\prime}} + N_{S_j^{t}\!,S_j^{t^{\prime}}}^{e,t\not=t^{\prime}} = \big(|S_{i}^{t}| \cdot |S_{i}^{t^{\prime}}|\big) + \big(|S_{j}^{t}| \cdot |S_{j}^{t^{\prime}}|\big)$ typed 4-node induced subgraphs that contain edge $e=(i,j)$ such that $w_k \in S_{i}^{t}$, $w_r \in S_i^{t^{\prime}}$ \emph{or} $w_k \in S_{j}^{t}$, $w_r \in S_j^{t^{\prime}}$ where $N_{S_i^t\!,S_i^{t^{\prime}}}^{e,t\not=t^{\prime}} = |S_{i}^{t}| \cdot |S_{i}^{t^{\prime}}|$ and $N_{S_j^{t}\!,S_j^{t^{\prime}}}^{e,t\not=t^{\prime}} = |S_{j}^{t}| \cdot |S_{j}^{t^{\prime}}|$.
From Corollary~\ref{cor:typed-4-stars}, the number of typed 4-stars that contain edge $e=(i,j)$ is $f_{ij}(g_{5}, \vt) = N_{S_{i}^{t},S_{i}^{t^{\prime}},0}^{e,t\not=t^{\prime}} + N_{S_{j}^{t},S_{j}^{t^{\prime}},0}^{e,t\not=t^{\prime}}$.
Similarly, from Corollary~\ref{cor:typed-4-tailed-tri-tail-edge-orbits}, the typed tailed-triangle tail-edge orbit count for edge $e$ is $f_{ij}(g_{7},\vt) = N_{S_{i}^{t},S_{i}^{t^{\prime}},1}^{e,t\not=t^{\prime}} + N_{S_{j}^{t},S_{j}^{t^{\prime}},1}^{e,t\not=t^{\prime}}$.
Therefore, by direct substitution in Eq.~\ref{eq:proof-4-stars-and-tailed-triangles-tail-edge-orbit-relationship-t-not-equals-tp}, we obtain Eq.~\ref{eq:typed-4-star}.
\end{proof}
\noindent
The only path-based typed graphlet containing a triangle is the tailed-triangle tail-edge orbit.
Observe that this is the only orbit needed to derive the typed 4-star counts in constant time.

\subsubsection{Relationship between typed tailed-triangles (tri-edge orbit) and chordal-cycles (edge orbit)}
\label{sec:rel-tailed-tri-and-chordal-cycles}

\begin{cor}\label{cor:typed-4-tailed-tri-tri-edge-orbits}
For any edge $(i,j) \in E$ in $G$ with types $\phi_i$ and $\phi_j$, the number of typed tailed-triangle (paw) tri-edge orbits containing edge $(i,j)$ with type vector $\vt = \big[ \phi_i\; \phi_j\; t \; t^{\prime} \big]$ is
$N_{T_{ij}^{t}, S_i^t\!\vee S_j^{t},0}^{e,t=t^{\prime}}$ for $t=t^{\prime}$ and $N_{T_{ij}^{t}, S_i^{t^{\prime}}\!\vee S_j^{t^{\prime}},0}^{e,t\not=t^{\prime}} + N_{T_{ij}^{t^{\prime}}, S_i^{t}\!\vee S_j^{t},0}^{e,t\not=t^{\prime}}$ otherwise.
\end{cor}
\begin{cor}\label{cor:typed-4-chordal-cycle-edge-orbit}
For any edge $(i,j) \in E$ in $G$ with types $\phi_i$ and $\phi_j$, 
the typed chordal-cycle edge orbit count 
with type vector $\vt = \big[ \phi_i\; \phi_j\; t \; t^{\prime} \big]$ is 
$N_{T_{ij}^{t},S_i^{t}\! \vee S_j^{t},1}^{e,t=t^{\prime}}$ for $t=t^{\prime}$ and
$N_{T_{ij}^{t},S_i^{t^{\prime}}\! \vee S_j^{t^{\prime}}\!,1}^{e,t\not=t^{\prime}} + N_{T_{ij}^{t^{\prime}}\!,S_i^{t}\! \vee S_j^{t},1}^{e,t\not=t^{\prime}}$ otherwise.
\end{cor}

\begin{lemma}
\label{lem:typed-4-tailed-triangle-triangle-edge-orbit}
For any edge $(i,j) \in E$ in $G$ with types $\phi_i$ and $\phi_j$ 
and any type vector $\vt = \big[ \phi_i\; \phi_j\; t \; t^{\prime} \big]$,
the relationship between 
the typed tailed-triangle tri-edge orbit count $f_{ij}(g_{9}, \vt)$ 
and 
the typed chordal-cycle edge orbit count $f_{ij}(g_{10},\vt)$ 
with type vector $\vt$ is
\begin{equation} \label{eq:typed-4-tailed-triangle-triangle-edge-orbit}
\!\!\!\!\!f_{ij}(g_9,\!\vt) \!=\! 
\begin{cases}
\!\big(|T_{ij}^{t}| \!\cdot\! (|S_{i}^{t}| + |S_{j}^{t}|)\big) \!- \!f_{ij}(g_{10}, \!\vt) 						& \!\text{if } t=t^{\prime} \\[5pt]
\!\big(|T_{ij}^{t}| \!\cdot\! (|S_{i}^{t^{\prime}}| + |S_{j}^{t^{\prime}}|)\big) \; + &  \!\text{otherwise} \\
\!\big(|T_{ij}^{t^{\prime}}| \!\cdot\! (|S_{i}^{t}| + |S_{j}^{t}|)\big) \!- \!f_{ij}(g_{10}, \!\vt)	 & \\[2pt]
\end{cases}
\end{equation}\noindent
where $f_{ij}(g_{10},\vt)$ is the chordal-cycle edge orbit count for edge $(i,j) \in E$ with type vector $\vt$.
\end{lemma}
\begin{proof}
Assume $t=t^{\prime}$.
From Theorem~\ref{thm:general-prin-for-typed-graphlets}, we have 
\begin{equation}\label{eq:proof-tailed-tri-and-chordal-cycle-relationship-t-equals-tp}
N_{T_{ij}^{t}, S_i^t\!\vee S_j^{t}\!,0}^{e,t=t^{\prime}} = N_{T_{ij}^{t}, S_i^t\!\vee S_j^{t}}^{e,t=t^{\prime}} - N_{T_{ij}^{t}, S_i^t\!\vee S_j^{t}\!,1}^{e,t=t^{\prime}}
\end{equation}
Let $N_{T_{ij}^{t}, S_i^t\!\vee S_j^{t}}^{e,t=t^{\prime}} = N_{T_{ij}^{t}, S_{i}^{t}}^{e,t=t^{\prime}} + N_{T_{ij}^{t}, S_{j}^{t}}^{e,t=t^{\prime}}$.
Since $N_{T_{ij}^{t}, S_i^t\!\vee S_j^{t}}^{e,t=t^{\prime}}$ is the number of typed 4-node induced subgraphs containing $e=(i,j)$ such that $w_k \in T_{ij}^{t}$ and $w_r \in S_i^t \cup S_j^t$, then $N_{T_{ij}^{t}, S_i^t\!\vee S_j^{t}}^{e,t=t^{\prime}} = |T_{ij}^{t}| \cdot (|S_i^t| + |S_j^t|)$ by Property~\ref{prop:P-not-equal-Q-mutually-exclusive}.
From Corollary~\ref{cor:typed-4-tailed-tri-tri-edge-orbits}, the number of typed tailed-triangles tri-edge orbits that contain $e=(i,j)$ is $f_{ij}(g_{9}, \vt) = N_{T_{ij}^{t},\!S_i^t\!\vee \!S_j^{t}\!,0}^{e,t=t^{\prime}}$.
From Corollary~\ref{cor:typed-4-chordal-cycle-edge-orbit}, the number of typed chordal-cycle edge orbits centered at edge $e$ is $f_{ij}(g_{10},\vt) = N_{T_{ij}^{t},\!S_i^t\!\vee \!S_j^{t}\!,1}^{e,t=t^{\prime}}$. 
Therefore, by direct substitution in Eq.~\ref{eq:proof-tailed-tri-and-chordal-cycle-relationship-t-equals-tp}, we obtain Eq.~\ref{eq:typed-4-tailed-triangle-triangle-edge-orbit}.

Assume $t\not=t^{\prime}$.
From Theorem~\ref{thm:general-prin-for-typed-graphlets}, we have 
$P=T_{ij}^{t}$ and $Q=S_{i}^{t^{\prime}} \cup S_{j}^{t^{\prime}}$ 
\emph{or} $P=T_{ij}^{t^{\prime}}$ and $Q=S_{i}^{t} \cup S_{j}^{t}$.
Therefore, 
\begin{equation} \label{eq:proof-tailed-tri-and-chordal-cycle-relationship-t-not-equals-tp}
N_{T_{ij}^{t},S_i^{t^{\prime}}\! \vee S_j^{t^{\prime}}\!\!,0}^{e,t\not=t^{\prime}} = \, N_{T_{ij}^{t},S_i^{t^{\prime}}\! \vee S_j^{t^{\prime}}}^{e,t\not=t^{\prime}} - N_{T_{ij}^{t},S_i^{t^{\prime}}\! \vee S_j^{t^{\prime}}\!\!,1}^{e,t\not=t^{\prime}}
\quad\;\;and\;\;\quad
N_{T_{ij}^{t^{\prime}}\!,S_i^{t}\! \vee S_j^{t},0}^{e,t\not=t^{\prime}} = \, N_{T_{ij}^{t^{\prime}}\!,S_i^{t}\! \vee S_j^{t}}^{e,t\not=t^{\prime}} - N_{T_{ij}^{t^{\prime}}\!,S_i^{t}\! \vee S_j^{t},1}^{e,t\not=t^{\prime}}
\end{equation}
By rewriting the above, 
\begin{equation}\label{eq:proof-tailed-tri-and-chordal-cycle-relationship-t-not-equals-tp-combined}
N_{T_{ij}^{t},S_i^{t^{\prime}}\! \vee S_j^{t^{\prime}}\!\!,0}^{e,t\not=t^{\prime}} + N_{T_{ij}^{t^{\prime}}\!,S_i^{t}\! \vee S_j^{t}\!,0}^{e,t\not=t^{\prime}} = \, 
\Big( N_{T_{ij}^{t},S_i^{t^{\prime}}\! \vee S_j^{t^{\prime}}}^{e,t\not=t^{\prime}} + N_{T_{ij}^{t^{\prime}}\!,S_i^{t}\! \vee S_j^{t}}^{e,t\not=t^{\prime}} \Big) - 
\Big( N_{T_{ij}^{t},S_i^{t^{\prime}}\! \vee S_j^{t^{\prime}}\!\!,1}^{e,t\not=t^{\prime}} + N_{T_{ij}^{t^{\prime}}\!,S_i^{t}\! \vee S_j^{t}\!,1}^{e,t\not=t^{\prime}} \Big)
\end{equation}
By Property~\ref{prop:P-not-equal-Q-mutually-exclusive}, there are $N_{T_{ij}^{t},S_i^{t^{\prime}}\! \vee S_j^{t^{\prime}}}^{e,t\not=t^{\prime}} + N_{T_{ij}^{t^{\prime}}\!,S_i^{t}\! \vee S_j^{t}}^{e,t\not=t^{\prime}} = |T_{ij}^{t}|(|S_{i}^{t^{\prime}}| + |S_{j}^{t^{\prime}}|) + |T_{ij}^{t^{\prime}}|(|S_{i}^{t}| + |S_{j}^{t}|)$ typed 4-node induced subgraphs that contain edge $e=(i,j)$ such that $w_k \in T_{ij}^{t}$ and $w_r \in S_i^{t^{\prime}}\! \cup S_j^{t^{\prime}}$ \emph{or} $w_k \in T_{ij}^{t^{\prime}}$ and $w_r \in S_i^{t}\! \cup S_j^{t}$.
From Corollary~\ref{cor:typed-4-chordal-cycle-edge-orbit}, the typed chordal-cycle edge orbit count for edge $e=(i,j)$ is $f_{ij}(g_{10}, \vt) = N_{T_{ij}^{t},S_i^{t^{\prime}}\! \vee S_j^{t^{\prime}}\!\!,1}^{e,t\not=t^{\prime}} + N_{T_{ij}^{t^{\prime}}\!,S_i^{t}\! \vee S_j^{t}\!,1}^{e,t\not=t^{\prime}}$.
Similarly, from Corollary~\ref{cor:typed-4-tailed-tri-tri-edge-orbits}, the typed tailed-triangle tri-edge orbit count for edge $e$ is $f_{ij}(g_{9},\vt) = N_{T_{ij}^{t},S_i^{t^{\prime}}\! \vee S_j^{t^{\prime}}\!\!,0}^{e,t\not=t^{\prime}} + N_{T_{ij}^{t^{\prime}}\!,S_i^{t}\! \vee S_j^{t}\!,0}^{e,t\not=t^{\prime}}$.
Therefore, by direct substitution in Eq.~\ref{eq:proof-tailed-tri-and-chordal-cycle-relationship-t-not-equals-tp-combined}, we obtain Eq.~\ref{eq:typed-4-tailed-triangle-triangle-edge-orbit}.

\end{proof}

\makeatletter
\global\let\tikz@ensure@dollar@catcode=\relax
\tikzstyle{solid}=                   [dash pattern=]
\tikzstyle{dotted}=                  [dash pattern=on \pgflinewidth off 2pt]
\tikzstyle{densely dotted}=          [dash pattern=on \pgflinewidth off 1pt]
\tikzstyle{loosely dotted}=          [dash pattern=on \pgflinewidth off 4pt]
\tikzstyle{dashed}=                  [dash pattern=on 3pt off 3pt]
\tikzstyle{densely dashed}=          [dash pattern=on 3pt off 2pt]
\tikzstyle{loosely dashed}=          [dash pattern=on 3pt off 6pt]
\tikzstyle{dashdotted}=              [dash pattern=on 3pt off 2pt on \the\pgflinewidth off 2pt]
\tikzstyle{densely dashdotted}=      [dash pattern=on 3pt off 1pt on \the\pgflinewidth off 1pt]
\tikzstyle{loosely dashdotted}=      [dash pattern=on 3pt off 4pt on \the\pgflinewidth off 4pt]
\makeatother

\tikzstyle{every node}=[font=\large,line width=1.5pt]
\begin{figure}[t!]
\centering
\begin{center}

\scalebox{1.00}{
\subfigure[Lemma~\ref{lem:typed-4path-center-orbit-count}]
{
\scalebox{0.32}{
\centering
\begin{tikzpicture}[-,>=latex,auto,node distance=2.2cm,thick,main node/.style={circle,draw=black,fill=black,draw,font=\sffamily\Huge\bfseries,text=white,
minimum width=1.03cm
}]

\node[main node] (1) [left=5pt] {$\mathbf{i}$};
\node[main node] (2) [right of=1,right=5pt] {$\mathbf{j}$};
\node[main node] (3) [above of=2, above=6pt] {$\mathbf{r}$};
\node[main node] (4) [above of=1, above=6pt] {$\mathbf{k}$};

\path[every node/.style={font=\sffamily\huge\bfseries}] 
(1) edge [line width=1.1mm, left] node [above left] {} (2) 
	(2)  edge [line width=1.1mm, right] node[below right] {} (3)
(3) edge [loosely dotted, line width=1.3mm, left] node[above] {$\mathbf{\![e^{\prime}]}$} (4)
	(1) edge [line width=1.1mm,right] node[] {} (4);
\end{tikzpicture}
}
\label{fig:lemma1}
}}
\hspace{6mm}
\scalebox{1.0}{
\subfigure[Lemma~\ref{lem:typed-4star-count}]
{
\scalebox{0.32}{
\centering
\begin{tikzpicture}[-,>=latex,auto,node distance=2.2cm,thick,main node/.style={circle,draw=black,fill=black,draw,font=\sffamily\Huge\bfseries,text=white,
minimum width=1.03cm
}]

\node[main node] (1) {$\mathbf{i}$};
\node[main node] (2) [right of=1, right=10pt] {$\mathbf{j}$};
\node[main node] (3) [above of=1, left=10pt] {$\mathbf{r}$};
\node[main node] (4) [left of=3, left=20pt, above=2pt] {$\mathbf{k}$};

\tikzstyle{LabelStyle}=[below=3pt]
\path[every node/.style={font=\sffamily\huge\bfseries}] 
(1) edge [line width=1.1mm, left] node [above left] {} (2) 
	(1)  edge [line width=1.1mm, right] node[below right] {} (3)
	(1) edge [line width=1.1mm,right] node[above right] {} (4)
	(3) edge [loosely dotted, line width=1.3mm,right] node[above] {$\huge\mathbf{\![e^{\prime}]}$} (4);
\end{tikzpicture}
}
\label{fig:lemma2}
}
}
\hspace{2mm}
\scalebox{1.0}{
\subfigure[Lemma~\ref{lem:typed-4-tailed-triangle-triangle-edge-orbit}]
{
\scalebox{0.32}{
\centering
\begin{tikzpicture}[-,>=latex,auto,node distance=2.2cm,thick,main node/.style={circle,draw=black,fill=black,draw,font=\sffamily\Huge\bfseries,text=white,
minimum width=1.03cm
}]

\node[main node] (1) {$\mathbf{i}$};
\node[main node] (3) [above right of=1] {$\mathbf{k}$};
\node[main node] (2) [below right of=3] {$\mathbf{j}$};
\node[main node] (4) [above left of=1, above=10pt] {$\mathbf{r}$};

\tikzstyle{LabelStyle}=[below=3pt]
\path[every node/.style={font=\sffamily\huge\bfseries}] 
(1) edge [line width=1.1mm, left] node [above left] {} (2) 
	(1)  edge [line width=1.1mm, right] node[below right] {} (3)
(2) edge [line width=1.1mm, left] node[below left] {} (3)
	(1) edge [line width=1.1mm,right] node[above right] {} (4)
	(3) edge [loosely dotted, line width=1.3mm,right] node[above] {$\mathbf{[e^{\prime}]}$} (4);
\end{tikzpicture}
}
\label{fig:lemma3}
}
}
\hspace{6mm}
\scalebox{1.00}{
\subfigure[Lemma~\ref{lem:typed-chordal-cycle-center-orbit-count}]
{
\scalebox{0.32}{
\centering
\begin{tikzpicture}[-,>=latex,auto,node distance=2.2cm,thick,
main node/.style={circle,draw=black,fill=black,draw,font=\sffamily\Huge\bfseries,text=white,minimum width=1.03cm},
white/.style={circle,draw=white,fill=white,draw,font=\sffamily,text=white,minimum width=1.001cm},
]
\node[main node] (1) {$\mathbf{i}$};
\node[main node] (3) [above right of=1] {$\mathbf{k}$};
\node[main node] (2) [below right of=3] {$\mathbf{j}$};
\node[main node] (4) [above of=3, above=12pt]{$\mathbf{r}$};

\tikzstyle{LabelStyle}=[below=3pt]
\path[every node/.style={font=\sffamily\huge\bfseries}] 
	(1) edge [bend left,line width=1.1mm] node[above right] {} (4)
	(2) edge [bend right,line width=1.1mm] node[above right] {} (4)
(1) edge [line width=1.1mm, left] node [above left] {} (2) 
	(1)  edge [line width=1.1mm, right] node[below right] {} (3)
(2) edge [line width=1.1mm, left] node[below left] {} (3)
	(3) edge [loosely dotted, line width=1.3mm,right] node[] {$\mathbf{[e^{\prime}]}$} (4);
\end{tikzpicture}
}
\label{fig:lemma4}
}
}
\end{center}

\vspace{-5mm}
\caption{
Combinatorial relationships of typed graphlets.
(a) Relationship between typed 4-paths and 4-cycles.
(b) Relationship between typed 4-stars and tailed-triangles (tail-edge orbit).
(c) Relationship between typed tailed-triangles (tri-edge orbit) and chordal-cycles (edge orbit).
(d) Relationship between typed 4-cliques and chordal-cycles (center orbit).
}
\label{fig:proof-correctness-examples}
\end{figure}
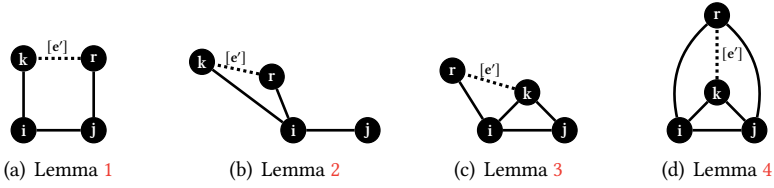

\subsubsection{Relationship between typed 4-cliques and chordal-cycles (center orbit)}

\begin{cor}\label{cor:typed-4-cliques}
For any edge $(i,j) \in E$ in $G$ with types $\phi_i$ and $\phi_j$, 
the number of typed 4-cliques containing edge $(i,j)$ with type vector $\vt = \big[ \phi_i\; \phi_j\; t \; t^{\prime} \big]$ is 
$N_{T_{ij}^{t},T_{ij}^{t},1}^{e,t=t^{\prime}}$ for $t=t^{\prime}$ and $N_{T_{ij}^{t},T_{ij}^{t^{\prime}},1}^{e,t\not=t^{\prime}}$ for $t\not= t^{\prime}$.
\end{cor}

\begin{cor}\label{cor:typed-chordal-cycle-center-orbits}
For any edge $(i,j) \in E$ in $G$ 
with types $\phi_i$ and $\phi_j$, 
the number of typed chordal-cycles (center orbit) containing edge $(i,j)$ 
with type vector $\vt = \big[ \phi_i\; \phi_j\; t \; t^{\prime} \big]$ is 
$N_{T_{ij}^{t},T_{ij}^{t},0}^{e,t=t^{\prime}}$ for $t=t^{\prime}$ and $N_{T_{ij}^{t},T_{ij}^{t^{\prime}}\!,0}^{e,t\not=t^{\prime}}$ for $t\not= t^{\prime}$.
\end{cor}
Notice the only difference between Corollary~\ref{cor:typed-4-cliques} and~\ref{cor:typed-chordal-cycle-center-orbits} 
is that $(w_k,w_r) \in E$ must hold for counting typed 4-cliques for a given edge whereas 
for counting chordal-cycle center orbits $(w_k,w_r) \not\in E$ must hold.
Further, if $t\not=t^{\prime}$, then $w_r \not= w_k$ can be removed by Property~\ref{prop:P-not-equal-Q-mutually-exclusive} since by definition $T_{ij}^{t} \cap T_{ij}^{t^{\prime}} = \emptyset$.

\begin{lemma}
\label{lem:typed-chordal-cycle-center-orbit-count}
For any edge $(i,j) \in E$ in $G$ with types $\phi_i$ and $\phi_j$ 
and any type vector $\vt = \big[ \phi_i\; \phi_j\; t \; t^{\prime} \big]$,
the relationship between 
the typed 4-clique count $f_{ij}(g_{12}, \vt)$ 
and 
the typed chordal-cycle center orbit count $f_{ij}(g_{11},\vt)$ 
with type vector $\vt$ is
\begin{equation} \label{eq:typed-4-chordal-cycle-center-orbit}
f_{ij}(g_{11},\vt) = 
\begin{cases}
\mychoose{|T_{ij}^{t}|}{2} - f_{ij}(g_{12}, \vt) 				 			 	& \text{if } t=t^{\prime} \\[5pt]
\big(|T_{ij}^{t}| \cdot |T_{ij}^{t^{\prime}}|\big) - f_{ij}(g_{12}, \vt)	 	& \text{otherwise} \\[2pt]
\end{cases}
\end{equation}\noindent
where $f_{ij}(g_{12},\vt)$ is the typed 4-clique count for edge $(i,j) \in E$ with type vector $\vt$.
\end{lemma}
\begin{proof}
Let $T_{ij}^t$ and $T_{ij}^{t^{\prime}}$ be the nodes that form typed triangles with $(i,j) \in E$ of type $t$ and $t^{\prime}$, respectively.
There are again two cases.

Assume $t=t^{\prime}$.
From Theorem~\ref{thm:general-prin-for-typed-graphlets}, we have 
\begin{equation}\label{eq:proof-chordal-cycle-center-and-4-clique-relationship-t-equals-tp}
N_{T_{ij}^{t},T_{ij}^{t},0}^{e,t=t^{\prime}} = N_{T_{ij}^{t},T_{ij}^{t}}^{e,t=t^{\prime}} - N_{T_{ij}^{t},T_{ij}^{t},1}^{e,t=t^{\prime}}
\end{equation}
Since $t = t^{\prime}$, then $P=T_{ij}^t$ and $Q=T_{ij}^t$, hence $P=Q$.
Therefore, by Property~\ref{prop:P-equals-Q}, there are $N_{T_{ij}^{t},T_{ij}^{t}}^{e,t=t^{\prime}} = \mychoose{|T_{ij}^{t}|}{2}$ typed 4-node induced subgraphs that contain edge $e=(i,j)$ such that $w_k \in T_{ij}^{t}$ and $w_r \in T_{ij}^{t}$.
From Corollary~\ref{cor:typed-4-cliques}, the number of typed 4-cliques that contain edge $e=(i,j)$ is $f_{ij}(g_{12}, \vt) = N_{T_{ij}^{t},T_{ij}^{t},1}^{e,t=t^{\prime}}$.
Similarly, from Corollary~\ref{cor:typed-chordal-cycle-center-orbits}, the typed chordal-cycle center orbit count for edge $e$ is $f_{ij}(g_{11},\vt) = N_{T_{ij}^{t},T_{ij}^{t},0}^{e,t=t^{\prime}}$.
Therefore, by direct substitution in Eq.~\ref{eq:proof-chordal-cycle-center-and-4-clique-relationship-t-equals-tp}, we obtain Eq.~\ref{eq:typed-4-chordal-cycle-center-orbit}.

Assume $t\not=t^{\prime}$.
From Theorem~\ref{thm:general-prin-for-typed-graphlets}, we have 
$P=T_{ij}^{t}$ and $Q=T_{ij}^{t^{\prime}}$, therefore
\begin{equation} \label{eq:proof-chordal-cycle-center-and-4-clique-relationship-t-not-equals-tp}
N_{T_{ij}^{t},T_{ij}^{t^{\prime}},0}^{e,t\not=t^{\prime}} = \, N_{T_{ij}^{t},T_{ij}^{t^{\prime}}}^{e,t\not=t^{\prime}} - N_{T_{ij}^{t},T_{ij}^{t^{\prime}},1}^{e,t\not=t^{\prime}}
\end{equation}
By Property~\ref{prop:P-not-equal-Q-mutually-exclusive}, there are $N_{T_{ij}^{t},T_{ij}^{t^{\prime}}}^{e,t\not=t^{\prime}} = |T_{ij}^{t}|\cdot |T_{ij}^{t^{\prime}}|$ typed 4-node induced subgraphs that contain edge $e=(i,j)$ such that $w_k \in T_{ij}^{t}$ and $w_r \in T_{ij}^{t^{\prime}}$.
From Corollary~\ref{cor:typed-4-cliques}, the number of typed 4-cliques that contain edge $e=(i,j)$ is $f_{ij}(g_{12}, \vt) = N_{T_{ij}^{t},T_{ij}^{t^{\prime}},1}^{e,t\not=t^{\prime}}$.
Similarly, from Corollary~\ref{cor:typed-chordal-cycle-center-orbits}, the typed chordal-cycle center orbit count for edge $e$ is $f_{ij}(g_{11},\vt) = N_{T_{ij}^{t},T_{ij}^{t^{\prime}},0}^{e,t\not=t^{\prime}}$.
Therefore, by direct substitution in Eq.~\ref{eq:proof-chordal-cycle-center-and-4-clique-relationship-t-not-equals-tp} we obtain Eq.~\ref{eq:typed-4-chordal-cycle-center-orbit}.

\end{proof}

{
\algblockdefx[parallel]{parfor}{endpar}[1][]{$\textbf{parallel for}$ #1 $\textbf{do}$}{$\textbf{end parallel}$}
\algrenewcommand{\alglinenumber}[1]{\fontsize{7.5}{8.5}\selectfont#1\;}
\begin{figure}[h!]
\begin{center}
\begin{algorithm}[H]
\caption{\,\small
Update Typed Graphlets.
Add typed graphlet with hash $\hash$ to $\mathcal{M}_{ij}$ if $\hash \not\in \mathcal{M}_{ij}$ and increment $\vx_{\hash}$ (frequency of that typed graphlet for a given edge).
}
\label{alg:updated-typed-motifs}
\begin{spacing}{1.15}
\small
\begin{algorithmic}[1]
\Procedure {Update}{$\vx$, $\mathcal{M}_{ij}$, $\hash = \mathbb{F}\,(g, \Phi_i, \Phi_j, \Phi_k, \Phi_r)$}
\If{$\hash \not\in \mathcal{M}_{ij}$} \label{algline:check-if-typed-motif-already-added}
$\mathcal{M}_{ij} \leftarrow \mathcal{M}_{ij} \cup \{\hash\}$ and set $\vx_{\hash} = 0$
\EndIf
\vspace{-1.1mm}
\State $\vx_{\hash} = \vx_{\hash} + 1$ \label{algline:update-typed-motif-frequency}
\State {\bf return} updated set of typed graphlets $\mathcal{M}_{ij}$ and counts $\vx$
\EndProcedure
\medskip
\end{algorithmic}
\end{spacing}
\vspace{-1.mm}
\end{algorithm}
\end{center}
\end{figure}
}

\subsection{From \emph{Typed} Orbits to Graphlets}
\noindent
Counts of the \emph{typed graphlets} for each edge $(i,j) \in E$ can be derived from the \emph{typed graphlet orbits} using the following equations:
\begin{align}
& f_{ij}(h_{3}, \vt) = f_{ij}(g_{3}, \vt) + f_{ij}(g_{4}, \vt) \\
& f_{ij}(h_{4}, \vt) = f_{ij}(g_{5}, \vt) \\
& f_{ij}(h_{5}, \vt) = f_{ij}(g_{6}, \vt) \\
& f_{ij}(h_{6}, \vt) = f_{ij}(g_{7}, \vt) + f_{ij}(g_{8}, \vt) + f_{ij}(g_{9}, \vt) \\
& f_{ij}(h_{7}, \vt) = f_{ij}(g_{10}, \vt) + f_{ij}(g_{11}, \vt) \\
& f_{ij}(h_{8}, \vt) = f_{ij}(g_{12}, \vt) 
\end{align}\noindent
where $h_{}$ is the graphlet without considering the orbit (Table~\ref{table:typed-graphlet-equations}).

\subsection{\emph{Typed} Graphlet Hash Functions} \label{sec:typed-motif-hash-function}
\noindent
Given a general heterogeneous graph with $L$ unique types such that $L<10$, then a simple and efficient typed graphlet hash function $\mathbb{F}$ is defined as follows:
\begin{equation}\label{eq:simple-typed-motif-hash-function}
\mathbb{F}(g,\vt) = g10^4 + t_1 10^3 + t_2 10^2 + t_3 10^1 + t_4
\end{equation}\noindent
where $g$ encodes the $k$-node graphlet orbit (\eg, 4-path center)
and $t_1$, $t_2$, $t_3$, $t_4$ encode the type of the nodes in $H \in \mathcal{H}$ with type vector $\vt = \big[ t_1 \; t_2 \; t_3 \; t_4 \big]$.
Since the maximum hash value resulting from Eq.~\ref{eq:simple-typed-motif-hash-function} is small (and fixed for any arbitrarily large graph $G$), we can leverage a perfect hash table to allow for fast $o(1)$ constant time lookups to determine if a typed graphlet was previously found or not as well as updating the typed graphlet count in $o(1)$ constant time.
For $k$-node graphlets where $k<4$, we simply set the last $4-k$ types to $0$.
Note the simple typed graphlet hash function defined above can be extended trivially to handle graphs with $L \geq 10$ types:
\begin{equation} \label{eq:simple-typed-motif-hash-function-10-or-more}
\mathbb{F}(g,\vt) = g10^8 + t_1 10^6 + t_2 10^4 + t_3 10^2 + t_4
\end{equation}
In general, any non-cryptographic hash function $\mathbb{F}$ can be used (see~\citet{chi2017hashing} for some other possibilities).
Thus, the approach is independent of $\mathbb{F}$ and can always leverage the best known hash function.
The only requirement of the hash function is that it is invertible $\mathbb{F}^{-1}$.

Thus far we have not made any assumption on the ordering of types in $\vt$.
As such, the hash function $\mathbb{F}$ discussed above can be used directly in the framework for counting typed graphlets such that the type structure and position are preserved (See Section~\ref{sec:position-aware-typed-graphlets} for further discussion on position-aware typed graphlets).
However, since we are interested in counting all typed graphlets $\wrt$ Definition~\ref{def:typed-graphlet-instance}, 
then we map all such orderings of the types in $\vt$ to the same hash value using a precomputed hash table.
This allows us to obtain the unique hash value in $o(1)$ constant time for any ordering of the types in $\vt$.
In our implementation, we compute $s = t_1 10^3 + t_2 10^2 + t_3 10^1 + t_4$ and then use $s$ as an index into the precomputed hash table to obtain the unique hash value $c$ in $o(1)$ constant time.

\begin{figure}[h!]
\centering
\includegraphics[width=0.4\linewidth]{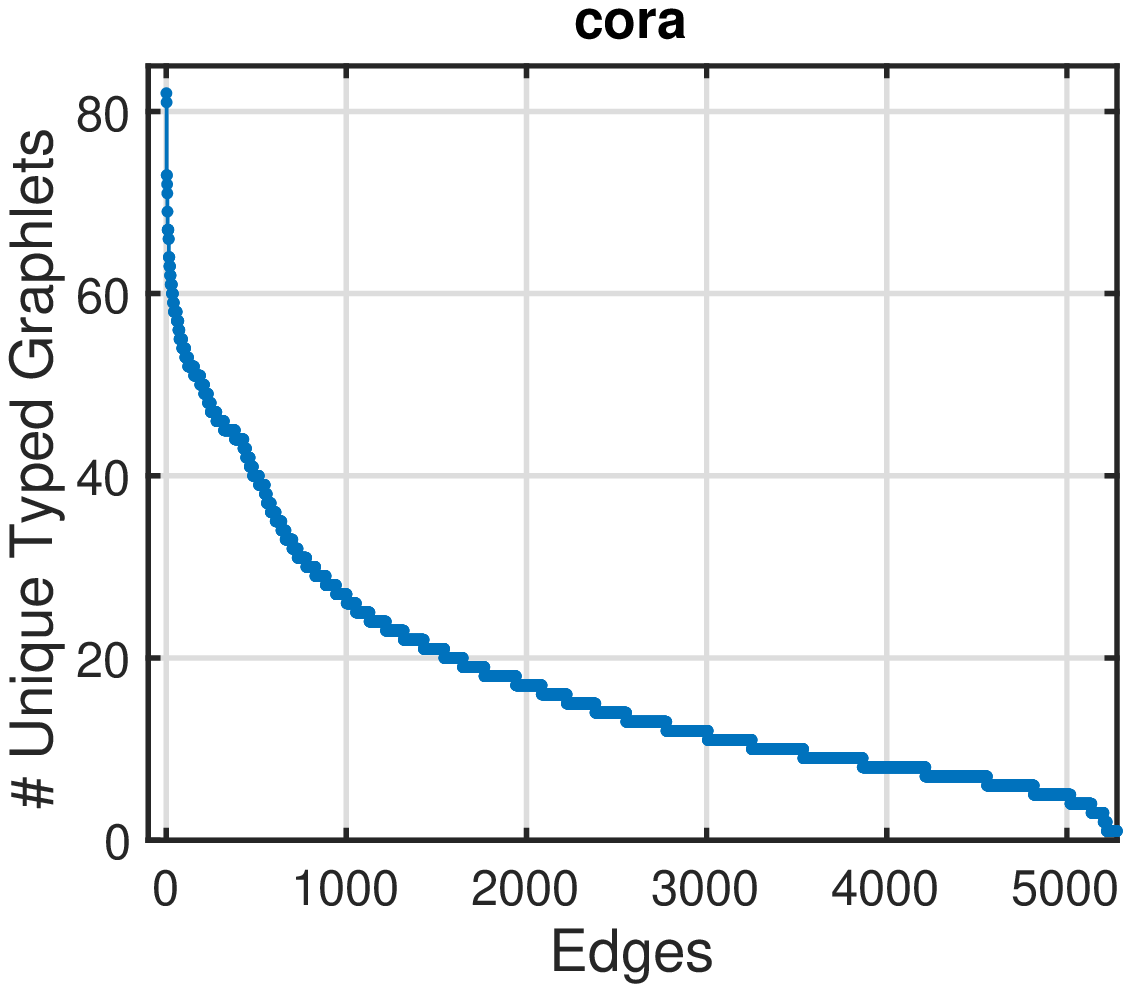}
\caption{
Distribution of unique typed graphlets that occur on the edges.
This experiment considers all typed graphlets of $\{3,4\}$-nodes.
Among the 1428 possible unique typed graphlets that could arise in $G$, there are only 876 unique typed graphlets that actually occur (at least once at an edge in $G$).
Even more striking, the maximum unique typed graphlets that occur on any edge in $G$ (cora) is only 82.
Overall, the mean number of unique typed graphlets over all edges in $G$ is 17,  \ie, only about 1.1\% of the possible typed graphlets.
These results indicate the significance of only a few typed graphlets as the vast majority of the typed graphlet counts for any arbitrary edge is zero.
Thus, the space required by the approach is nearly-optimal.
}
\label{fig:unique-typed-motif-edges-cora}
\end{figure}

\subsection{Sparse Typed Graphlet Format} \label{sec:sparse-typed-motif-format}
This section describes a space-efficient representation for typed graphlets based on a key observation.
\begin{Property}\label{prop:small-frac-typed-motifs}
Let $T$ denote the number of \emph{unique typed graphlets} that appear in an arbitrary graph $G$ with $L$ types.
Assuming the graph $G$ has a skewed degree distribution,
then most edges in $G$ appear in only a small fraction of the $T$ actual typed graphlets that can occur.
\end{Property}\noindent
This property is shown empirically in Figure~\ref{fig:unique-typed-motif-edges-cora} and 
implies that using a $M \times T$ matrix to store the typed graphlet counts is far from optimal in terms of the space required due to most of the $T$ typed graphlet counts being zero for any given edge. 
Based on this observation, we ensure the proposed approach uses near-optimal space by storing only the typed graphlets with nonzero counts for each edge $(i,j) \in E$ in the graph.
Typed graphlet counts are stored in a sparse format since it would be impractical in large graphs to store all typed graphlets as there can easily be hundreds of thousands depending on the number of types in the input graph. 

For each edge, we store only the nonzero typed graphlet counts along with the unique ids associated with them.
The unique ids allow us to map the nonzero counts to the actual typed graphlets.
We also developed a space-efficient format for storing the resulting typed graphlet counts to disk.
Instead of using the typed graphlet hash as the unique id, we remap the typed graphlets to smaller consecutive ids (starting from $1$) to reduce the space requirement even further.
Finally, we store a typed graphlet lookup table that maps a given graphlet id to its description and is useful for looking up the meaning of the typed graphlets discovered.

\subsection{Parallelization} \label{sec:parallel}
\noindent
We now describe a parallelization strategy for the proposed typed graphlet counting approach.
While our implementation uses shared memory, the parallelization is described generally such that it can be used with a distributed-memory architecture as well.
As such, our discussion is on the general scheme.
The parallel constructs we use are a worker task-queue and a global broadcast channel. 
Here, we assume that each worker has a copy of the graph and distribute edges to workers to find the typed graphlet counts that node $i$ and $j$ participate.
At this point, we view the main while loop as a task
generator and farm the current edge out to a worker to find the typed graphlet counts that co-occur between node $i$ and node $j$.
The approach is lock free since each worker uses the same graphlet hash function to obtain a unique hash value for every typed graphlet.
Thus, each worker can simply maintain the typed graphlets identified and their counts for every edge assigned to it.
In our own shared memory implementation, we avoid some of the communications by using global arrays and avoiding locked updates to them by using a unique edge id.
Counting typed graphlets on the edges as opposed to the nodes also has computational advantages with respect to parallelization and in particular load balancing.
Let $x_i$ and $x_{ij}$ denote the node and edge count of an arbitrary graphlet $H$.
Since $|E| \gg |V|$ and $\sum_{i \in V} x_{i} = \sum_{(i,j) \in E} x_{ij}$, then $\frac{1}{|V|}\sum_{i \in V} x_{i} < \frac{1}{|E|}\sum_{(i,j) \in E} x_{ij}$.
Hence, more work per vertex is required than per edge.
Therefore, counting typed graphlets on the edges is guaranteed to have better load balancing than node-centric algorithms.

\subsection{Discussion} 
\label{sec:framework-discussion}
This work formalized the notion of typed graphlet and provided a time- and space-efficient framework for counting all $\{2,3,4\}$-node typed graphlets.
Counting typed graphlets of a larger size is outside the scope of this paper and left for future work.
However, the ideas introduced in this paper can be used to extend and derive equations for typed graphlets of 5-nodes and larger.
In particular, Theorem~\ref{thm:general-prin-for-typed-graphlets} states the general principle of counting typed graphlets which is based on inclusion-exclusion, and therefore is straightforward to apply to typed graphlets of larger sizes.
This would follow directly from recent work~\cite{dave2017clog,PinarWWW17} that extended the ideas of \citet{pgd,pgd-kais} to 5-node untyped graphlets.
For instance, just as we did in this work, the node sets used to derive the different 5-node untyped graphlet counts in~\cite{dave2017clog,PinarWWW17} are further partitioned into subsets where each subset represents nodes of the same type,
\eg, just as a node $w \in T_{ij}^t$ is a node of type $t$ that forms a triangle with $i$ and $j$.
Afterwards, we can derive typed equations just as we did in this work to handle the different cases, \ie, when all types are the same vs. when they are different, and so on.
Nevertheless, counting typed graphlets of 5 nodes and larger is outside the scope of this work and left for future work.

The approach is also straightforward to adapt for directed typed graphlets.
In particular, we simply replace $\Gamma_i^{t}$ with $\Gamma_i^{t,+}$ and $\Gamma_i^{t,-}$ for typed out-neighbors and typed in-neighbors, respectively.
Thus, we also have $T_{ij}^{t,+}$, $T_{ij}^{t,-}$, $S_j^{t,+}$, $S_j^{t,-}$, $S_i^{t,+}$, and $S_i^{t,-}$.
Now it is just a matter of enumerating all combinations of these sets with the out/in-neighbor sets as well.
That is, we essentially have two additional versions of Algorithm~\ref{alg:typed-motifs-exact} and Algorithm~\ref{alg:typed-path-based-motifs-exact}-\ref{alg:typed-triangle-based-motifs-exact} for each in and out set (\wrt to the main for loop).
The other trivial modification is to ensure each directed typed graphlet is assigned a unique id (this is the same modification required for typed orbits).
The time and space complexity remains the same since all we did is split the set of neighbors (and the other sets) into two smaller sets by partitioning the nodes in $\Gamma_i^{t}$ into $\Gamma_i^{t,+}$ and $\Gamma_i^{t,-}$.
Similarly, for $T_{ij}^{t}$, $S_j^{t}$, and $S_i^{t}$.

\section{Position-Aware Typed Graphlets} \label{sec:position-aware-typed-graphlets}
\subsection{Formulation} \label{sec:position-aware-typed-graphlets-formulation}
We can consider the presence of topologically identical ``appearances" of a typed graphlet in a graph such that the type structure is preserved as well.
More formally, we define a \emph{position-aware typed graphlet} that ensures node (edge) types coincide via the isomorphism:
\begin{Definition}[\scshape Position-aware Typed Graphlet Instance]\label{def:position-aware-typed-graphlet-instance}
An instance of a position-aware typed graphlet $H=(V',E',\phi',\xi')$ of graph $G$ is a typed graphlet $F=(V'',E'',\phi'',\xi'')$ of $G$ such that
\begin{compactenum}
\item $(V'',E'')$ is isomorphic to $(V',E')$, 
\item $\mathcal{T}_{V''}=\mathcal{T}_{V'}$ and $\mathcal{T}_{E''}=\mathcal{T}_{E'}$, that is, the multisets of node and edge types are correspondingly equal.
\item $\phi''=\phi'\circ \, p$ and $\xi''=\xi'\circ\, q$ where $q = p\times p$, that is, the node and edge types coincide via the graph isomorphism $p$.
\end{compactenum}
\end{Definition}
This formulation can be used to count position-aware typed graphlets that \emph{preserve type structure}.
Such typed graphlets are called \emph{position-aware typed graphlets} to distinguish them from typed graphlets formalized in Definition~\ref{def:typed-graphlet-instance}.
See Figure~\ref{fig:position-aware-typed-graphlets} for an example of position-aware typed graphlets and Table~\ref{table:typed-graphlets-combin-properties} summarizes the combinatorial and enumerative properties.
For a single $K$-node induced subgraph, the number of \emph{position-aware typed graphlets} with $L$ types is $L^K$.

\definecolor{typeOneColor}{RGB}{8,81,156} 
\definecolor{typeTwoColor}{RGB}{222,45,38} 
\definecolor{typeThreeColor}{RGB}{49,163,84} 
\definecolor{typeFourColor}{RGB}{117,107,177} 

\makeatletter
\global\let\tikz@ensure@dollar@catcode=\relax
\makeatother
\tikzstyle{every node}=[font=\large,line width=1.5pt]
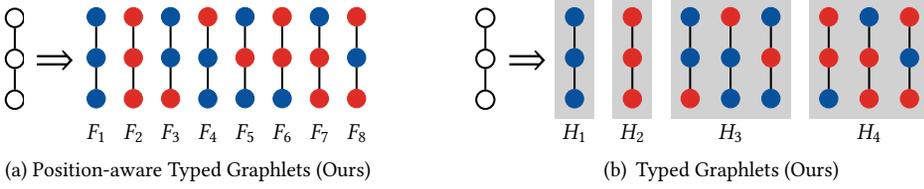
\begin{figure}[h!]

\tikzstyle{background-page}=[rectangle,
fill=gray!30,
inner sep=0.5cm,
rounded corners=0mm]  
\tikzstyle{background-white}=[rectangle,
fill=white,
inner ysep=0.5cm,
rounded corners=0mm]  

\tikzstyle{background-page}=[rectangle,
fill=gray!30,
fill=white,
outer ysep=0.1cm,
outer xsep=0.0cm,
inner xsep=0.2cm,
inner ysep=0.5cm,
rounded corners=0mm]  
\tikzstyle{background-white}=[rectangle,
fill=white,
inner ysep=0.5cm,
rounded corners=0mm]

\scalebox{1.0}{
\subfigure[\!\vspace{-8mm}Position-aware Typed Graphlets (Ours)\!]{
\scalebox{0.96}{

\scalebox{0.28}{
\begin{tikzpicture}[-,>=latex,auto,node distance=2.0cm,thick,
main node/.style={circle,draw=black,fill=white,draw,font=\sffamily\Huge\bfseries,text=black,minimum width=0.9cm, line width=1mm},
]

\node[main node] (1) {}; 
\node[main node] (2) [below of=1] {}; 
\node[main node] (3) [below of=2] {}; 

\tikzstyle{LabelStyle}=[below=3pt]
\path[every node/.style={font=\sffamily}] 
(1) edge [line width=1.0mm, left] node [above left] {} (2) 
(2) edge [line width=1.0mm, left] node[below left] {} (3);

\begin{pgfonlayer}{background}
		\node [background-white, 
fit=(1) (2) (3),
label=below:\fontsize{32}{32}\selectfont 
\textcolor{white}{$F_1$}
] {};
\end{pgfonlayer} 

\end{tikzpicture}
}
\hspace{-2.7mm}
\scalebox{0.28}{
\begin{tikzpicture}[-,>=latex,auto,node distance=2.0cm,thick,
main node/.style={circle,draw=black,fill=white,draw,font=\sffamily\Huge\bfseries,text=black,minimum width=0.9cm},
white node/.style={draw=white,draw,font=\sffamily\Huge\bfseries,text=black,minimum width=0.9cm}
]
\node[white node] (1) {};
\node[white node] (2) [below of=1] {};
\node[white node] (3) [below of=2] {};
\node[white node] (4) [right of=2, left=30pt, below=10pt, above=0.05pt] {\vspace{2mm}\fontsize{56}{56}\selectfont $\Rightarrow$};

\begin{pgfonlayer}{background}
		\node [background-white, 
fit=(1) (2) (3) (4),
label=below:\fontsize{32}{32}\selectfont 
\textcolor{white}{$F_1$}
] {};
\end{pgfonlayer} 

\end{tikzpicture}
}
\hspace{-2mm}
\scalebox{0.28}{
\begin{tikzpicture}[-,>=latex,auto,node distance=2.0cm,thick,
typeOne node/.style={circle,draw=typeOneColor,fill=typeOneColor,draw,font=\sffamily\Huge\bfseries,text=white,minimum width=0.9cm},
typeTwo node/.style={circle,draw=typeTwoColor,fill=typeTwoColor,draw,font=\sffamily\Huge\bfseries,text=white,minimum width=0.9cm},
typeThree node/.style={circle,draw=typeThreeColor,fill=typeThreeColor,draw,font=\sffamily\Huge\bfseries,text=white,minimum width=0.9cm},
]

\node[typeOne node] (1) {}; 
\node[typeOne node] (2) [below of=1] {}; 
\node[typeOne node] (3) [below of=2] {}; 

\tikzstyle{LabelStyle}=[below=3pt]
\path[every node/.style={font=\sffamily}] 
(1) edge [line width=1.0mm, left] node [above left] {} (2) 
(2) edge [line width=1.0mm, left] node[below left] {} (3);

\begin{pgfonlayer}{background}
\node [background-page, 
fit=(1) (2) (3),
label=below:\fontsize{32}{32}\selectfont $F_1$
] {};
\end{pgfonlayer} 
\end{tikzpicture}
}
\label{fig:typed-3-path-homo-typeOne-3colors}
\scalebox{0.28}{
\begin{tikzpicture}[-,>=latex,auto,node distance=2.0cm,thick,
typeOne node/.style={circle,draw=typeOneColor,fill=typeOneColor,draw,font=\sffamily\Huge\bfseries,text=white,minimum width=0.9cm},
typeTwo node/.style={circle,draw=typeTwoColor,fill=typeTwoColor,draw,font=\sffamily\Huge\bfseries,text=white,minimum width=0.9cm},
typeThree node/.style={circle,draw=typeThreeColor,fill=typeThreeColor,draw,font=\sffamily\Huge\bfseries,text=white,minimum width=0.9cm},
]

\node[typeTwo node] (1) {}; 
\node[typeTwo node] (2) [below of=1] {}; 
\node[typeTwo node] (3) [below of=2] {}; 

\tikzstyle{LabelStyle}=[below=3pt]
\path[every node/.style={font=\sffamily}] 
(1) edge [line width=1.0mm, left] node [above left] {} (2) 
(2) edge [line width=1.0mm, left] node[below left] {} (3);

\begin{pgfonlayer}{background}
\node [background-page, 
fit=(1) (2) (3),
label=below:\fontsize{32}{32}\selectfont $F_2$
] {};
\end{pgfonlayer} 
\end{tikzpicture}
}
\scalebox{0.28}{
\begin{tikzpicture}[-,>=latex,auto,node distance=2.0cm,thick,
typeOne node/.style={circle,draw=typeOneColor,fill=typeOneColor,draw,font=\sffamily\Huge\bfseries,text=white,minimum width=0.9cm},
typeTwo node/.style={circle,draw=typeTwoColor,fill=typeTwoColor,draw,font=\sffamily\Huge\bfseries,text=white,minimum width=0.9cm},
typeThree node/.style={circle,draw=typeThreeColor,fill=typeThreeColor,draw,font=\sffamily\Huge\bfseries,text=white,minimum width=0.9cm},
]

\node[typeOne node] (1) {}; 
\node[typeOne node] (2) [below of=1] {}; 
\node[typeTwo node] (3) [below of=2] {}; 

\tikzstyle{LabelStyle}=[below=3pt]
\path[every node/.style={font=\sffamily}] 
(1) edge [line width=1.0mm, left] node [above left] {} (2) 
(2) edge [line width=1.0mm, left] node[below left] {} (3);

\begin{pgfonlayer}{background}
\node [background-page, 
fit=(1) (2) (3),
label=below:\fontsize{32}{32}\selectfont $F_3$
] {};
\end{pgfonlayer} 
\end{tikzpicture}
}
\scalebox{0.28}{
\begin{tikzpicture}[-,>=latex,auto,node distance=2.0cm,thick,
typeOne node/.style={circle,draw=typeOneColor,fill=typeOneColor,draw,font=\sffamily\Huge\bfseries,text=white,minimum width=0.9cm},
typeTwo node/.style={circle,draw=typeTwoColor,fill=typeTwoColor,draw,font=\sffamily\Huge\bfseries,text=white,minimum width=0.9cm},
typeThree node/.style={circle,draw=typeThreeColor,fill=typeThreeColor,draw,font=\sffamily\Huge\bfseries,text=white,minimum width=0.9cm},
]

\node[typeTwo node] (1) {}; 
\node[typeOne node] (2) [below of=1] {}; 
\node[typeOne node] (3) [below of=2] {}; 

\tikzstyle{LabelStyle}=[below=3pt]
\path[every node/.style={font=\sffamily}] 
(1) edge [line width=1.0mm, left] node [above left] {} (2) 
(2) edge [line width=1.0mm, left] node[below left] {} (3);

\begin{pgfonlayer}{background}
\node [background-page, 
fit=(1) (2) (3),
label=below:\fontsize{32}{32}\selectfont $F_4$
] {};
\end{pgfonlayer} 
\end{tikzpicture}
}
\scalebox{0.28}{
\begin{tikzpicture}[-,>=latex,auto,node distance=2.0cm,thick,
typeOne node/.style={circle,draw=typeOneColor,fill=typeOneColor,draw,font=\sffamily\Huge\bfseries,text=white,minimum width=0.9cm},
typeTwo node/.style={circle,draw=typeTwoColor,fill=typeTwoColor,draw,font=\sffamily\Huge\bfseries,text=white,minimum width=0.9cm},
typeThree node/.style={circle,draw=typeThreeColor,fill=typeThreeColor,draw,font=\sffamily\Huge\bfseries,text=white,minimum width=0.9cm},
]

\node[typeOne node] (4) {}; 
\node[typeTwo node] (5) [below of=4] {}; 
\node[typeOne node] (6) [below of=5] {}; 

\tikzstyle{LabelStyle}=[below=3pt]
\path[every node/.style={font=\sffamily}] 

(4) edge [line width=1.0mm, left] node [above left] {} (5) 
(5) edge [line width=1.0mm, left] node[below left] {} (6);

\begin{pgfonlayer}{background}
\node [background-page, 
fit=(4) (5) (6),
label=below:\fontsize{32}{32}\selectfont $F_5$
] {};
\end{pgfonlayer} 
\end{tikzpicture}
}
\scalebox{0.28}{
\begin{tikzpicture}[-,>=latex,auto,node distance=2.0cm,thick,
typeOne node/.style={circle,draw=typeOneColor,fill=typeOneColor,draw,font=\sffamily\Huge\bfseries,text=white,minimum width=0.9cm},
typeTwo node/.style={circle,draw=typeTwoColor,fill=typeTwoColor,draw,font=\sffamily\Huge\bfseries,text=white,minimum width=0.9cm},
typeThree node/.style={circle,draw=typeThreeColor,fill=typeThreeColor,draw,font=\sffamily\Huge\bfseries,text=white,minimum width=0.9cm},
]

\node[typeTwo node] (7) {}; 
\node[typeTwo node] (8) [below of=7] {}; 
\node[typeOne node] (9) [below of=8] {}; 

\tikzstyle{LabelStyle}=[below=3pt]
\path[every node/.style={font=\sffamily}] 
(7) edge [line width=1.0mm, left] node [above left] {} (8) 
(8) edge [line width=1.0mm, left] node[below left] {} (9);

\begin{pgfonlayer}{background}
\node [background-page, 
fit=(7) (8) (9),
label=below:\fontsize{32}{32}\selectfont $F_6$
] {};
\end{pgfonlayer} 
\end{tikzpicture}
}
\scalebox{0.28}{
\begin{tikzpicture}[-,>=latex,auto,node distance=2.0cm,thick,
typeOne node/.style={circle,draw=typeOneColor,fill=typeOneColor,draw,font=\sffamily\Huge\bfseries,text=white,minimum width=0.9cm},
typeTwo node/.style={circle,draw=typeTwoColor,fill=typeTwoColor,draw,font=\sffamily\Huge\bfseries,text=white,minimum width=0.9cm},
typeThree node/.style={circle,draw=typeThreeColor,fill=typeThreeColor,draw,font=\sffamily\Huge\bfseries,text=white,minimum width=0.9cm},
]

\node[typeOne node] (7) {}; 
\node[typeTwo node] (8) [below of=7] {}; 
\node[typeTwo node] (9) [below of=8] {}; 

\tikzstyle{LabelStyle}=[below=3pt]
\path[every node/.style={font=\sffamily}] 
(7) edge [line width=1.0mm, left] node [above left] {} (8) 
(8) edge [line width=1.0mm, left] node[below left] {} (9);

\begin{pgfonlayer}{background}
\node [background-page, 
fit=(7) (8) (9),
label=below:\fontsize{32}{32}\selectfont $F_7$
] {};
\end{pgfonlayer} 
\end{tikzpicture}
}
\scalebox{0.28}{
\begin{tikzpicture}[-,>=latex,auto,node distance=2.0cm,thick,
typeOne node/.style={circle,draw=typeOneColor,fill=typeOneColor,draw,font=\sffamily\Huge\bfseries,text=white,minimum width=0.9cm},
typeTwo node/.style={circle,draw=typeTwoColor,fill=typeTwoColor,draw,font=\sffamily\Huge\bfseries,text=white,minimum width=0.9cm},
typeThree node/.style={circle,draw=typeThreeColor,fill=typeThreeColor,draw,font=\sffamily\Huge\bfseries,text=white,minimum width=0.9cm},
]

\node[typeTwo node] (10) {}; 
\node[typeOne node] (11) [below of=10] {}; 
\node[typeTwo node] (12) [below of=11] {}; 

\tikzstyle{LabelStyle}=[below=3pt]
\path[every node/.style={font=\sffamily}] 

(10) edge [line width=1.0mm, left] node [above left] {} (11) 
(11) edge [line width=1.0mm, left] node[below left] {} (12);

\begin{pgfonlayer}{background}
\node [background-page, 
fit=(10) (11) (12),
label=below:\fontsize{32}{32}\selectfont $F_8$
] {};
\end{pgfonlayer} 
\end{tikzpicture}
}
\label{fig:position-aware-typed-graphlet}
}
}

\tikzstyle{background-page}=[rectangle,
fill=gray!30,
outer ysep=0.1cm,
inner sep=0.5cm,
rounded corners=0mm]  
\tikzstyle{background-white}=[rectangle,
fill=white,
inner ysep=0.5cm,
rounded corners=0mm]

\hspace{6.0mm}
\hfill
\subfigure[\vspace{-2mm}Typed Graphlets (Ours)]{
\scalebox{0.96}{

\scalebox{0.28}{
\begin{tikzpicture}[-,>=latex,auto,node distance=2.0cm,thick,
main node/.style={circle,draw=black,fill=white,draw,font=\sffamily\Huge\bfseries,text=black,minimum width=0.9cm, line width=1mm},
]

\node[main node] (1) {}; 
\node[main node] (2) [below of=1] {}; 
\node[main node] (3) [below of=2] {}; 

\tikzstyle{LabelStyle}=[below=3pt]
\path[every node/.style={font=\sffamily}] 
(1) edge [line width=1.0mm, left] node [above left] {} (2) 
(2) edge [line width=1.0mm, left] node[below left] {} (3);

\begin{pgfonlayer}{background}
		\node [background-white, 
fit=(1) (2) (3),
label=below:\fontsize{32}{32}\selectfont 
\textcolor{white}{$H_1$}
] {};
\end{pgfonlayer} 

\end{tikzpicture}
}
\hspace{-2.7mm}
\scalebox{0.28}{
\begin{tikzpicture}[-,>=latex,auto,node distance=2.0cm,thick,fill=none,
main node/.style={circle,draw=black,fill=none,draw,font=\sffamily\Huge\bfseries,text=black,minimum width=0.9cm},
white node/.style={draw=white,fill=none,draw,font=\sffamily\Huge\bfseries,text=black,minimum width=0.9cm}
]
\node[white node] (1) {};
\node[white node] (2) [below of=1] {};
\node[white node] (3) [below of=2] {};
\node[white node] (4) [right of=2, left=30pt, below=10pt, above=0.05pt] {\vspace{2mm}\fontsize{56}{56}\selectfont $\Rightarrow$};

\begin{pgfonlayer}{background}
		\node [background-white, 
fit=(1) (2) (3) (4),
label=below:\fontsize{32}{32}\selectfont 
\textcolor{white}{$H_1$}
] {};
\end{pgfonlayer} 

\end{tikzpicture}
}
\hspace{-2mm}
\scalebox{0.28}{
\begin{tikzpicture}[-,>=latex,auto,node distance=2.0cm,thick,
typeOne node/.style={circle,draw=typeOneColor,fill=typeOneColor,draw,font=\sffamily\Huge\bfseries,text=white,minimum width=0.9cm},
typeTwo node/.style={circle,draw=typeTwoColor,fill=typeTwoColor,draw,font=\sffamily\Huge\bfseries,text=white,minimum width=0.9cm},
typeThree node/.style={circle,draw=typeThreeColor,fill=typeThreeColor,draw,font=\sffamily\Huge\bfseries,text=white,minimum width=0.9cm},
]

\node[typeOne node] (1) {}; 
\node[typeOne node] (2) [below of=1] {}; 
\node[typeOne node] (3) [below of=2] {}; 

\tikzstyle{LabelStyle}=[below=3pt]
\path[every node/.style={font=\sffamily}] 
(1) edge [line width=1.0mm, left] node [above left] {} (2) 
(2) edge [line width=1.0mm, left] node[below left] {} (3);

\begin{pgfonlayer}{background}
\node [background-page, 
fit=(1) (2) (3),
label=below:\fontsize{32}{32}\selectfont $H_1$
] {};
\end{pgfonlayer} 
\end{tikzpicture}
}
\label{fig:typed-3-path-homo-typeOne-3colors}
\hspace{0.3mm}
\scalebox{0.28}{
\begin{tikzpicture}[-,>=latex,auto,node distance=2.0cm,thick,
typeOne node/.style={circle,draw=typeOneColor,fill=typeOneColor,draw,font=\sffamily\Huge\bfseries,text=white,minimum width=0.9cm},
typeTwo node/.style={circle,draw=typeTwoColor,fill=typeTwoColor,draw,font=\sffamily\Huge\bfseries,text=white,minimum width=0.9cm},
typeThree node/.style={circle,draw=typeThreeColor,fill=typeThreeColor,draw,font=\sffamily\Huge\bfseries,text=white,minimum width=0.9cm},
]

\node[typeTwo node] (1) {}; 
\node[typeTwo node] (2) [below of=1] {}; 
\node[typeTwo node] (3) [below of=2] {}; 

\tikzstyle{LabelStyle}=[below=3pt]
\path[every node/.style={font=\sffamily}] 
(1) edge [line width=1.0mm, left] node [above left] {} (2) 
(2) edge [line width=1.0mm, left] node[below left] {} (3);

\begin{pgfonlayer}{background}
\node [background-page, 
fit=(1) (2) (3),
label=below:\fontsize{32}{32}\selectfont $H_2$
] {};
\end{pgfonlayer} 
\end{tikzpicture}
}
\hspace{0.3mm}
\scalebox{0.28}{
\begin{tikzpicture}[-,>=latex,auto,node distance=2.0cm,thick,
typeOne node/.style={circle,draw=typeOneColor,fill=typeOneColor,draw,font=\sffamily\Huge\bfseries,text=white,minimum width=0.9cm},
typeTwo node/.style={circle,draw=typeTwoColor,fill=typeTwoColor,draw,font=\sffamily\Huge\bfseries,text=white,minimum width=0.9cm},
typeThree node/.style={circle,draw=typeThreeColor,fill=typeThreeColor,draw,font=\sffamily\Huge\bfseries,text=white,minimum width=0.9cm},
]

\node[typeOne node] (1) {}; 
\node[typeOne node] (2) [below of=1] {}; 
\node[typeTwo node] (3) [below of=2] {}; 

\node[typeTwo node] (4) [right of=1] {}; 
\node[typeOne node] (5) [below of=4] {}; 
\node[typeOne node] (6) [below of=5] {}; 

\node[typeOne node] (7) [right of=4] {}; 
\node[typeTwo node] (8) [below of=7] {}; 
\node[typeOne node] (9) [below of=8] {}; 

\tikzstyle{LabelStyle}=[below=3pt]
\path[every node/.style={font=\sffamily}] 
(1) edge [line width=1.0mm, left] node [above left] {} (2) 
(2) edge [line width=1.0mm, left] node[below left] {} (3)

(4) edge [line width=1.0mm, left] node [above left] {} (5) 
(5) edge [line width=1.0mm, left] node[below left] {} (6)

(7) edge [line width=1.0mm, left] node [above left] {} (8) 
(8) edge [line width=1.0mm, left] node[below left] {} (9);

\begin{pgfonlayer}{background}
\node [background-page, 
fit=(1) (2) (3) (4) (5) (6) (7) (8) (9),
label=below:\fontsize{32}{32}\selectfont $H_3$
] {};
\end{pgfonlayer} 
\end{tikzpicture}
}
\hspace{0.3mm}
\scalebox{0.28}{
\begin{tikzpicture}[-,>=latex,auto,node distance=2.0cm,thick,
typeOne node/.style={circle,draw=typeOneColor,fill=typeOneColor,draw,font=\sffamily\Huge\bfseries,text=white,minimum width=0.9cm},
typeTwo node/.style={circle,draw=typeTwoColor,fill=typeTwoColor,draw,font=\sffamily\Huge\bfseries,text=white,minimum width=0.9cm},
typeThree node/.style={circle,draw=typeThreeColor,fill=typeThreeColor,draw,font=\sffamily\Huge\bfseries,text=white,minimum width=0.9cm},
]

\node[typeTwo node] (4) {}; 
\node[typeTwo node] (5) [below of=4] {}; 
\node[typeOne node] (6) [below of=5] {}; 

\node[typeOne node] (7) [right of=4] {}; 
\node[typeTwo node] (8) [below of=7] {}; 
\node[typeTwo node] (9) [below of=8] {}; 

\node[typeTwo node] (10) [right of=7] {}; 
\node[typeOne node] (11) [below of=10] {}; 
\node[typeTwo node] (12) [below of=11] {}; 

\tikzstyle{LabelStyle}=[below=3pt]
\path[every node/.style={font=\sffamily}] 

(4) edge [line width=1.0mm, left] node [above left] {} (5) 
(5) edge [line width=1.0mm, left] node[below left] {} (6)

(7) edge [line width=1.0mm, left] node [above left] {} (8) 
(8) edge [line width=1.0mm, left] node[below left] {} (9)

(10) edge [line width=1.0mm, left] node [above left] {} (11) 
(11) edge [line width=1.0mm, left] node[below left] {} (12);

\begin{pgfonlayer}{background}
\node [background-page, 
fit=(4) (5) (6) (7) (8) (9) (10) (11) (12),
label=below:\fontsize{32}{32}\selectfont $H_4$
] {};
\end{pgfonlayer} 
\end{tikzpicture}
}

\hspace{5.0mm}
\label{fig:typed-graphlet-example}
}
}

} 

\vspace{-2mm}
\caption{
\textbf{Position-aware Typed Graphlets (Def.~\ref{def:position-aware-typed-graphlet-instance}) and Typed Graphlets (Def.~\ref{def:typed-graphlet-instance}).}
This intuitive example shows the difference between position-aware typed graphlets \emph{and} typed graphlets that do not impose the constraint that types coincide via the isomorphism.
}
\label{fig:position-aware-typed-graphlets}
\end{figure}

\subsection{Algorithm} \label{sec:position-aware-typed-graphlets-algorithm}
To count position-aware typed graphlets, we only need a slight modification to the previous algorithm.
Notice the algorithmic difference between Definition~\ref{def:typed-graphlet-instance} and position-aware typed graphlets defined in Definition~\ref{def:position-aware-typed-graphlet-instance} is that to count typed graphlets, we only need to consider the types involved in the graphlet, and not the nodes/edges that correspond to those types (\ie, types coincide via the isomorphism).
In other words, typed graphlets formally defined in Definition~\ref{def:typed-graphlet-instance} can be viewed as ignoring the ``order'' of the types with respect to their assignment to nodes in the induced subgraph.
Recall that Section~\ref{sec:typed-motif-hash-function} used a function to map all possible orderings of a given multiset of $L$ types to a single hash value (using a precomputed hash table).
This allowed us to merge all such ``position-aware typed graphlets'' that have the same multiset of types (Figure~\ref{fig:position-aware-typed-graphlet}) to the appropriate typed graphlet (Figure~\ref{fig:typed-graphlet-example}).
For instance, the counts of the position-aware typed graphlets $\{F_3, F_4, F_5\}$ in Figure~\ref{fig:position-aware-typed-graphlet} are all mapped to the same typed graphlet $H_3$ shown in Figure~\ref{fig:typed-graphlet-example}.
Therefore, the framework (Algorithm~\ref{alg:typed-motifs-exact}) actually counts position-aware typed graphlets and uses a mapping function to obtain a single hash value for all such type orderings of some multiset of types, which allows the position-aware typed graphlet counts of a typed graphlet to be merged into a single count.
To count position-aware typed graphlets, we simply remove the lookup table that maps the appropriate position-aware typed graphlets to the corresponding typed graphlet.
Interestingly, this equates to slightly less work required for computing position-aware typed graphlets.
As such, the analysis in Section~\ref{sec:complexity-analysis} and elsewhere clearly holds for both position-aware typed graphlets and typed graphlets.

\begin{Property}
Given an untyped induced subgraph $H$, let $T_H$ denote the number of typed graphlets of $H$ in $G$, and let $P_H$ denote the number of position-aware typed graphlets of $H$ in $G$.
Then $P_H \geq T_H$.
\end{Property}
\begin{proof}
If $P_H = T_H$, then for the position-aware typed graphlets that map to a specific typed graphlet 
(\eg, $\{F_3,F_4,F_5\}$ maps to $H_3$ in Figure~\ref{fig:position-aware-typed-graphlets}), only one position-aware typed graphlet must have nonzero count in $G$, and this must hold for all typed graphlets in $G$.
Otherwise, if there exists a typed graphlet with more than one position-aware typed graphlet with nonzero count, then $P_H > T_H$.
\end{proof}

\section{Global Typed Graphlet Counts} \label{sec:global-typed-graphlet-counting}
\noindent
While Section~\ref{sec:framework} focused on counting typed graphlets locally for each edge $(i,j) \in E$ in $G$, one may also be interested in the total counts of each typed graphlet in $G$.
More formally,
\begin{Problem}[Global Typed Graphlets] \label{prob:global-typed-graphlet-counting}
Given a graph $G$ with $L$ types, the global typed graphlet counting problem is to find the set of all typed graphlets that occur in $G$ along with their corresponding frequencies.
This work focuses on computing all $\{2,3,4\}$-node typed graphlet counts for $G$.
\end{Problem}\noindent
A general equation for solving the above problem for any arbitrary \emph{typed graphlet} $H$ is given below.
Let $H$ denote an arbitrary typed graphlet and $\vx$ be an $M$-dimensional vector of counts of $H$ for every edge $(i,j) \in E$, then the frequency of $H$ in $G$ is:
\begin{equation}\label{eq:global-typed-graphlet-counts}
C_H \,= \, \frac{1}{|E(H)|}\; \vx^{\!\top}\ve 
\end{equation}\noindent
where $|E(H)|$ is the number of edges in the typed graphlet $H$ and $\ve = [ \, 1\; \cdots \; 1 \,]$ is an $M$-dimensional vector of all 1's.

\section{Theoretical Analysis} \label{sec:complexity-analysis}
\noindent
First, we show the relationship between the count of an untyped graphlet $H$ in $G$ and the count of all typed graphlets in $G$ with induced subgraph $H$.
\begin{Proposition} \label{prop:untyped-graphlet-equal-sum-of-all-typed-graphlets}
Let $\vx$ denote the vector of counts for any untyped graphlet $H \in \mathcal{H}$ (\eg, 4-cycle).
Further, let $\mX$ denote a $M \times T_{H}$ matrix of typed graphlet counts for graphlet $H$ where $T_{H}$ denotes the number of typed graphlets that arise from $L$ types.
Then the following holds:
\begin{equation}
C = \sum_{i=1}^{M} x_i = \sum_{i=1}^{M} \sum_{j=1}^{T_{H}} X_{ij}
\end{equation}
\end{Proposition}\noindent
Consider the counts of an untyped graphlet for a single edge.
The above demonstrates that these counts are partitioned among the set of typed graphlets that arise from the untyped graphlet when types are considered.
Let $\vp \in \RR^{T_{H}}$ denote a typed graphlet probability distribution ($\vp^T \ve = 1$), the entropy 
(average information content) of $\vp$ is
$\mathbb{H}(\vp) = -\sum_{i} \; p_i \log p_i$.
Hence, $\mathbb{H}(\vp)$ quantifies the amount of information in the relative frequencies of the typed graphlets (of a given graphlet $H \in \mathcal{H}$).
In the case of untyped graphlets, the $C = \sum_{i=1}^{M} x_i$ untyped graphlets are assumed to belong to a single homogeneous graphlet where all nodes are of the same type.
This matches exactly the information we have if types are not considered.
\begin{Proposition} \label{prop:info-gain}
Assume $\vp \in \RR^{T_{H}}$ is an arbitrary typed graphlet probability distribution such that $p_i<1$, $\forall i$ 
and $\vq$ is the untyped graphlet distribution where $q_i=1$ and $q_j=0, \forall j\not=i$, then $\mathbb{H}(\vp) > \mathbb{H}(\vq)$.
\end{Proposition}\noindent
This implies that typed graphlets contain more information than untyped graphlets.
The proof is straightforward.

\subsection{Time Complexity} \label{sec:time-complexity}
\noindent
We first introduce two properties that are useful to understand the complexity.
\begin{Property} \label{prop:1}
\begin{equation}
d_i + d_j = 2|T_{ij}| + |S_i| + |S_{j}|
\end{equation}\noindent
where $d_i = |\Gamma_i|$,\; $d_j = |\Gamma_j|$,\; $T_{ij} = \Gamma_i \cap \Gamma_j$,\; $S_i = \Gamma_i \setminus T_{ij}$,\; and $S_j = \Gamma_j \setminus T_{ij}$.
\end{Property}

\begin{Property} \label{prop:2}
The space required to store $T_{ij}$, $S_{i}$, and $S_{j}$ is less than $d_i+d_j$ iff $|T_{ij}|>0$.
\end{Property}
This is straightforward to see since $|S_{i}|+|S_{j}|=|S_{i} \cup S_{j}|$ always holds. 
However, if $|T_{ij}|=0$, then $|S_{i}|+|S_{j}|=d_i+d_j$.
Hence, triangles represent the smallest clique, and as shown in~\cite{rossi2018compressing-graphs-cliques} can be used to compress the graph.
As the density of the graph increases, more triangles are formed, and therefore less space is used.
Notice that the worst case is also unlikely to occur because of this fact.
For instance, suppose $d_i = \Delta$, $d_j = \Delta$, and $\Delta = n$ (worst case), then $|T_{ij}|=d_i=d_j$, and $|S_{i}|=0$, $|S_{j}|=0$.
Furthermore, if $|S_{i}|=n$, then $|S_{j}|=0$ and $|T_{ij}|=0$ must hold.
Obviously, if node $i$ is connected to all $n$ nodes, then any node $k \in \Gamma_{j}$ must form a triangle with $i$ ($k \in T_{ij}$).
For any node with maximum degree $\Delta$, there is very low probability that $|T_{ij}|=0$, which implies $|T_{ij}|>0$, $|S_{i}|<\Delta$ and $|S_{j}|<\Delta$. 

\subsubsection{Typed 3-Node Graphlets}
\noindent
\vspace{-2mm}
\begin{thm}\label{lem:time-complexity-3-node-graphlets}
The worst-case time complexity for counting all 3-node typed graphlets for a given edge $(i,j) \in E$ is:
\begin{equation}
\mathcal{O}(2|\Gamma_i| + |\Gamma_j|) = \mathcal{O}(\Delta)
\end{equation}\noindent
where $|\Gamma_i|$ and $|\Gamma_j|$ denote the number of nodes connected to node $i$ and $j$, respectively.
Further, $\Delta$ is the maximum degree in $G$.
\end{thm} 
\begin{proof}
It takes at most $\mathcal{O}(|\Gamma_i|+|\Gamma_j|)$ time to compute typed triangles (\ie, $T_{ij}^t$, for all $t=1,\ldots,L$) by hashing neighbors of $i$ in $\mathcal{O}(|\Gamma_i|)$ time, 
and then checking if each node $w \in \Gamma_j$ is hashed or not, taking $\mathcal{O}(|\Gamma_j|)$ time.
Similarly, if $w \in \Gamma_j$ is not hashed, then $S_j^t\leftarrow S_j^t \cup \{w\}$ where $t = \phi_w$.
Now all that remains is computing $S_i^t$, for all $t$.
Notice $|S_i^t| = |\Gamma_i^t| - |T_{ij}^t|$, for all $t=1,\ldots,L$.
\end{proof}

\subsubsection{Typed 4-Node Graphlets}
\noindent
We first provide the time complexity of deriving path-based and triangle-based graphlet orbits in Lemma~\ref{lem:time-complexity-path-based}-\ref{lem:time-complexity-triangle-based}, and then give the total time complexity of all 3 and 4-node typed graphlets in Theorem~\ref{thm:time-complexity-4-node-graphlets} based on these results.
Note that Lemma~\ref{lem:time-complexity-path-based}-\ref{lem:time-complexity-triangle-based} includes the time required to derive all typed 3-node typed graphlets.

\begin{lemma} \label{lem:time-complexity-path-based}
For a single edge $(i,j) \in E$, the worst-case time complexity for deriving all \emph{typed path-based graphlet orbits} is:
\begin{align} \label{eq:time-complexity-path-based}
\mathcal{O}\Big(\Delta \big(|S_i| + |S_j| \big) \Big) 
\end{align}\noindent
\end{lemma}\noindent
Note $|S_{i}|\Delta \geq \sum_{k \in S_{i}} d_k$ and $|S_{j}|\Delta \geq \sum_{k \in S_{j}} d_k$.

\begin{lemma} \label{lem:time-complexity-triangle-based}
For a single edge $(i,j) \in E$, the worst-case time complexity for deriving all \emph{typed triangle-based graphlet orbits} is:
\begin{equation} \label{eq:time-complexity-triangle-based}
\mathcal{O}\big(\Delta |T_{ij}| \big) 
\end{equation}\noindent
\end{lemma}\noindent
Notice $|T_{ij}|\Delta \geq |T_{ij}|\Delta_T \geq \sum_{k \in T_{ij}} d_k$ where $\Delta$ is the maximum degree of a node in $G$ and $\Delta_T$ is the maximum degree of a node in $T_{ij}$.
Thus, $|T_{ij}|\Delta$ only occurs iff $\forall k \in T_{ij}$, $d_k = \Delta$ where $\Delta = $ maximum degree of a node in $G$.
In sparse real-world graphs, $T_{ij}$ is likely to be smaller than $S_i$ and $S_j$ as triangles are typically more rare than 3-node paths.
Conversely, $T_{ij}$ is also more likely to contain high degree nodes, as nodes with larger degrees are obviously more likely to form triangles than those with small degrees.

From Lemma~\ref{lem:time-complexity-path-based}-\ref{lem:time-complexity-triangle-based}, we have the following:
\begin{thm} \label{thm:time-complexity-4-node-graphlets}
For a single edge $(i,j) \in E$, the worst-case time complexity for deriving all 3 and 4-node typed graphlet orbits is:
\begin{equation} \label{eq:time-complexity-fast-alg-overall}
\mathcal{O}\big(\Delta \big(|S_i| + |S_j| + |T_{ij}|\big)\big) 
\end{equation}\noindent
\end{thm}\noindent

\begin{proof}
The time complexity of each step is provided below.
Hashing all neighbors of node $i$ takes $\mathcal{O}(|\Gamma_i|)$.
Recall from Lemma~\ref{lem:time-complexity-3-node-graphlets} that counting all 3-node typed graphlets takes $\mathcal{O}(2|\Gamma_i| + |\Gamma_j|) = \mathcal{O}(\Delta)$ time for an edge $(i,j) \in E$.
This includes the time required to derive the number of typed 3-node stars and typed triangles for all types $t=1,\ldots,L$.
This information is needed to derive the remaining typed graphlet orbit counts in constant time.
Next, Algorithm~\ref{alg:typed-path-based-motifs-exact} is used to derive a few path-based typed graphlet orbit counts taking $\mathcal{O}(\Delta (|S_i| + |S_j|))$ time in the worst-case.
Similarly, Algorithm~\ref{alg:typed-triangle-based-motifs-exact} is used to derive a few triangle-based typed graphlet orbit counts taking in the worst-case $\mathcal{O}(\Delta |T_{ij}|)$ time.
As an aside, updating the count of a typed graphlet count is $o(1)$ (Algorithm~\ref{alg:updated-typed-motifs}).

Now, we derive the remaining typed graphlet orbit counts in constant time (Line~\ref{algline:main-alg-for-type-pair}-\ref{algline:main-alg-derive-remaining-typed-graphlet-orbits-constant-time}).
Since each type pair leads to different typed graphlets, we must iterate over at most $L(L-1)/2+L$ type pairs.
For each pair of types selected, we derive the typed graphlet orbit counts in $o(1)$ constant time via Eq.~\ref{eq:typed-4-path-center-orbit}-\ref{eq:typed-4-chordal-cycle-center-orbit} (See Line~\ref{algline:main-alg-for-type-pair}-\ref{algline:main-alg-derive-remaining-typed-graphlet-orbits-constant-time}).
Furthermore, the term involving $L$ is for the worst-case when there is at least one node in all $L$ sets (\ie, at least one node of every type $L$).
Nevertheless, since $L$ is a small constant, $L(L-1)/2+L$ is negligible.
Therefore, for a single edge, the worst-case time complexity is $\mathcal{O}(\Delta(|S_i|+|S_j|+|T_{ij}|))$.

Let $\bar{T}$ and $\bar{S}$ denote the average number of triangle and 3-node stars incident to an edge in $G$.
More formally, $\bar{T} = \frac{1}{M}\sum_{(ij) \in E} |T_{ij}|$ and $\bar{S} = \frac{1}{M} \sum_{(ij) \in E} |S_i|+|S_j|$.
The total worst-case time complexity for all $M$ edges is $\mathcal{O}(M\Delta(\bar{S}+\bar{T}))$.
Note that obviously $\bar{S}M = \sum_{(ij) \in E} |S_i|+|S_j|$ and $\bar{T}M = \sum_{(ij) \in E} |T_{ij}|$.
\end{proof}

\begin{cor}\label{lem:time-complexity}
The worst-case time complexity of counting typed graphlets using Algorithm~\ref{alg:typed-motifs-exact} matches the worst-case time complexity of the best known untyped graphlet counting algorithm.
\end{cor}

\begin{proof}
From Theorem~\ref{thm:time-complexity-4-node-graphlets} we have that $\mathcal{O}\big(\Delta \big(|S_i| + |S_j| + |T_{ij}|\big)\big)$, which is exactly the time complexity of the best known untyped graphlet counting algorithm (PGD~\cite{pgd,pgd-kais}).
\end{proof}

\subsection{Space Complexity} \label{sec:space-complexity}
\noindent
Since our approach generalizes to graphs with an arbitrary number of types $L$, the specific set of typed graphlets is unknown.
As demonstrated in Table~\ref{table:typed-graphlets-example}, it is impractical to store the counts of all possible $k$-node typed graphlets for any graph of reasonable size as typically done in traditional methods for untyped graphlets~\cite{rage,pgd}.
The space complexity required to store the counts of all possible typed graphlets is at least:
\begin{equation} \label{eq:space-complexity-of-other-methods}
\mathcal{O}(MT_{\max})
\end{equation}\noindent
where $M=|E|$ is the number of edges in $G$ and $T_{\max}$ is the number of different possible typed graphlets with $L$ types.
Thus, $MT_{\max}$ is the total space to store $M$ vectors of length $T_{\max}$, \ie, one $T_{\max}$-dimensional vector per edge.
To understand the space requirements 
and how it is impractical for any moderately sized graph, 
suppose we have a graph with $M=10,000,000$ edges and $L=7$ types.
Counting all 3- and 4-node typed graphlet orbits for every edge would require $90.72$ GB of space to store the large $MT_{\max}$ matrix (assuming 4 bytes per count/entry).
This is obviously impractical for any graph of even moderate size.
In contrast, Algorithm~\ref{alg:typed-motifs-exact} is orders of magnitude more space-efficient.
\begin{lemma}\label{lem:space-complexity}
The space complexity of typed graphlets (Algorithm~\ref{alg:typed-motifs-exact}) is $\mathcal{O}(M\bar{T})$.
\end{lemma}\noindent
\begin{proof}
For an edge $(i,j) \in E$, it takes $|\mathcal{X}_{ij}|$ space to store the counts of the nonzero typed graphlets.
Let $\bar{T} = \frac{1}{M} \sum_{(ij) \in E} |\mathcal{X}_{ij}|$ denote the average number of typed graphlets with nonzero counts per edge.
Therefore, the total space required to store the nonzero typed graphlet counts for all $M=|E|$ edges is only $\mathcal{O}(M\bar{T})$.
The space of all other data structures used in Algorithm~\ref{alg:typed-motifs-exact} is small in comparison, \eg, $\Psi$ takes at most $\mathcal{O}(|V|)$ space, whereas $T_{ij}$, $S_{i}$, and $S_{j}$ take $\mathcal{O}(\Delta)$ space in the worst-case (by Property~\ref{prop:relationship-between-typed-sets-and-untyped-sets}) and can be reused for every edge.
In addition, the size of $\vx$ is independent of the graph size ($|V|+|E|$) and can also be reused.
\end{proof}

From Lemma~\ref{lem:space-complexity}, it is straightforward to see that
\begin{equation} \label{eq:total-actual-space}
\mathcal{O}(M\bar{T}) \; \ll \; \mathcal{O}(MT_{\max})
\end{equation}\noindent
The space required by the proposed approach (Algorithm~\ref{alg:typed-motifs-exact}) is nearly-optimal and orders of magnitude lower than methods used for colored graphlets such as GC~\cite{gu2018heterAlignment}, which by definition solve a much simpler problem since there are strictly fewer colored graphlet counts to store.
This is also shown empirically in Table~\ref{table:space-results}.

\begin{table*}[h!]
\centering
\caption{
Runtime results for counting typed graphlets (ours) compared to state-of-the-art methods for colored graphlets (which is a different but simpler problem).
Since these methods are unable to handle large or even medium-sized graphs as shown below, we include a number of very small graphs (\eg, cora, citeseer, webkb) for comparison; and count all $\{2,3,4\}$-node typed graphlets (ours) and colored graphlets.
Note $\Delta=$ max node degree; $|\mathcal{T}_V|=$ number of node types; $|\mathcal{T}_E|=$ number of edge types.
}
\vspace{-3mm}
\label{table:runtime-perf}
\renewcommand{\arraystretch}{1.10} 
\renewcommand{\arraystretch}{1.05} 
\small
\setlength{\tabcolsep}{5.0pt} 
\begin{tabularx}{1.0\linewidth}{@{}
r H lllH cc HH
lXX HXH
H HHH H
@{}
}
\toprule
&&&&&&&
&&&
\multicolumn{5}{c}{\sc seconds} 
\\
\cmidrule(l{0pt}r{5pt}){11-15}

& 
& 
& 
& 
& 
&
& 
& &&
& 
& 
& 
& 
\textbf{Typed} 
\\

& 
& 
$|V|$ & $|E|$ & $\Delta$ & 
&
$|\mathcal{T}_V|$ & 
$|\mathcal{T}_E|$ 
& &&
\textbf{GC} & 
\textbf{ESU} & 
\textbf{G-Tries} & 
& 
\textbf{Graphlets} & 
\\
\midrule

\textsf{citeseer}
& 
& 3.3k & 4.5k & 99 & 
& 6 & 21
&      && 
46.27 & 
5937.75 &
144.08 & 
& 
\textbf{0.022} 
\\

\textsf{cora}
& 
& 2.7k & 5.3k & 168 & 
& 7 & 28
&      && 
467.20 & 
10051.07 & 
351.40 & 
& 
\textbf{0.032}  
\\

\textsf{fb-relationship} 
& 
& 7.3k & 44.9k & 106 & 
& 6 & 20
&      && 
1374.60 & 
54,837.69 & 
3789.17 & 
& 
\textbf{0.701} 
\\

\textsf{web-polblogs}
& 
& 1.2k & 16.7k & 351 & 
& 2 & 1
&      && 
28,986.70 & 
26,577.10 & 
1,563.04 & 
& 
\textbf{1.055} 
\\

\textsf{ca-DBLP}
& 
& 2.9k & 11.3k & 69 &
& 3 & 3
&      && 
149.20 & 
1,188.11 & 
18.90 & 
& 
\textbf{0.100} 
\\

\textsf{inf-openflights}
& 
& 2.9k & 15.7k & 242 & 
& 2 & 2
&      && 
9262.20 & 
18,839.36 & 
458.01 & 
& 
\textbf{0.578} 
\\

\textsf{soc-wiki-elec}
& 
& 7.1k & 100.8k & 1.1k & 
& 2 & 2
&      && 
ETL & 
ETL & 
26,468.85 & 
& 
\textbf{5.316} 
\\

\textsf{webkb}
& 
& 262 & 459 & 122 &  
& 5 & 14
&      && 
85.82 & 
7,158.10 & 
187.22 & 
& 
\textbf{0.006} 
\\

\textsf{terrorRel}
& 
& 881 & 8.6k & 36 & 
& 2 & 3
&      && 
192.6 & 
3130.7 & 
241.1 & 
& 
\textbf{0.039} 
\\

\textsf{pol-retweet} 
& 
& 18.5k & 48.1k & 786 & 
& 2 & 3
&      && 
ETL & 
ETL & 
ETL & 
&
\textbf{0.296} 
\\

\textsf{web-spam}
&
& 9.1k & 465k & 3.9k & 
& 3 & 6
&      && 
ETL & 
ETL & 
ETL & 
&
\textbf{210.97} 
\\

\midrule
\textsf{movielens}
& 
& 28.1k & 170.4k & 3.6k & 
& 3 & 3
&      && 
ETL & 
ETL & 
ETL & 
& 
\textbf{5.23}  
\\

\textsf{citeulike} 
& 
& 907.8k & 1.4M & 11.2k & 
& 3 & 2 
&      && 
ETL & 
ETL & 
ETL & 
& 
\textbf{126.53} 
\\

\textsf{yahoo-msg} 
& 
& 100.1k & 739.8k & 9.4k & 
& 2 & 2
&      && 
ETL & 
ETL & 
ETL & 
&
\textbf{35.22} 
\\

\textsf{dbpedia} 
& 
& 495.9k & 921.7k & 24.8k & 
& 4 & 3 
&      && 
ETL & 
ETL & 
ETL & 
& 
\textbf{56.02} 
\\

\textsf{digg} 
& 
& 217.3k & 477.3k & 219 & 
& 2 & 2 
&      && 
ETL & 
ETL & 
ETL & 
& 
\textbf{5.592} 
\\

\textsf{bibsonomy} 
& 
& 638.8k & 1.2M & 211 & 
& 3 & 3 
&      && 
ETL & 
ETL & 
ETL & 
& 
\textbf{3.631} 
\\

\textsf{epinions}
& 
& 658.1k & 2.6M & 775 &
& 2 & 2
&      && 
ETL & 
ETL & 
ETL & & 
\textbf{85.27} 
\\

\textsf{flickr} 
& 
& 2.3M & 6.8M & 216 & 
& 2 & 2 
&      && 
ETL & 
ETL & 
ETL & 
& 
\textbf{120.79} 
\\

\textsf{orkut}
& 
& 6M & 37.4M & 166 & 
& 2 & 2
&      && 
ETL & 
ETL & 
ETL & 
& 
\textbf{1241.01} 
\\

\midrule

\textsf{ER (10K,0.001)}
& 
& 10k & 50.1k & 26 & 
& 5 & 15
&      && 
183.32 & 
5,399.14 & 
241.27 & 
& 
\textbf{0.48} 
\\

\textsf{CL (1.8)}
& 
& 9.2k & 44.2k & 218 & 
& 5 & 15
&      && 
31,668 & 
45,399.14 & 
5,241.27 & 
& 
\textbf{1.46} 
\\

\textsf{KPGM (log 12,14)}
& 
& 3.3k & 43.2k & 1.3k & 
& 5 & 15
&       && 
ETL & 
ETL & 
63,843.86 & 
& 
\textbf{8.94} 
\\

\textsf{SW (10K,6,0.3)}
& 
& 10k & 30k & 12 & 
& 5 & 15
&    && 
21.48 & 
5,062.67 & 
206.92 & 
& 
\textbf{0.24} 
\\

\bottomrule
\multicolumn{14}{l}{\footnotesize\scriptsize $^{*}$ ETL = Exceeded Time Limit (24 hours / 86,400 seconds)} \\
\end{tabularx}
\end{table*}

\begin{table*}[t!]
\centering
\caption{
Runtime speedup results.
Note ``$\infty$'' indicates the baseline method (GC, ESU, or G-Tries) did not terminate within 24 hours and thus the precise speedup is unknown.
$|\mathcal{T}_V|=$ number of node types; $|\mathcal{T}_E|=$ number of edge types.
}
\vspace{-3mm}
\label{table:runtime-speedup-perf}
\renewcommand{\arraystretch}{1.10} 
\renewcommand{\arraystretch}{1.05} 
\small
\setlength{\tabcolsep}{5.5pt} 
\begin{tabularx}{0.98\linewidth}{@{}
r H 
ccH H
cc HH
HHH HH r
XXX H
}
\toprule
&&&&&&&
&&&
\multicolumn{5}{H}{} &&
\multicolumn{3}{c}{\sc speedup (typed graphlets vs.)}
\\
\cmidrule(l{5pt}r{5pt}){16-19}

& 
& 
$|V|$ & $|E|$ &  & 
&
$|\mathcal{T}_V|$ & 
$|\mathcal{T}_E|$ 
& && 
&&&
& 
\;\; &\;\;&
\textbf{GC} & 
\textbf{ESU} & 
\textbf{G-Tries} & 
\\
\midrule

\textsf{citeseer}
& 
& 3.3k & 4.5k & & 
& 6 & 21
&      && 
&&& &&&
2103x & 
269897x  & 
6549x & 
\\

\textsf{cora}
& 
& 2.7k & 5.3k & & 
& 7 & 28
&      && 
&&& &&&
14600x & 
314095x  & 
10981x & 
\\

\textsf{fb-relationship} 
& 
& 7.3k & 44.9k & & 
& 6 & 20
&      && 
&&& &&&
1960x & 
78227x & 
5405x & 
\\

\textsf{web-polblogs}
& 
& 1.2k & 16.7k & & 
& 2 & 1
&      && 
&&& &&&
27475x & 
25191x & 
1481x & 
\\

\textsf{ca-DBLP}
& 
& 2.9k & 11.3k & & 
& 3 & 3
&      && 
&&& &&& 
1492x & 
11881x & 
189x & 
\\

\textsf{inf-openflights}
& 
& 2.9k & 15.7k & & 
& 2 & 2
&      && 
&&& &&&
16024x & 
32594x & 
792x & 
\\

\textsf{soc-wiki-elec}
& 
& 7.1k & 100.8k & & 
& 2 & 2
&      && 
&&& &&&
$\infty$ & 
$\infty$ & 
45793x & 
\\

\textsf{webkb}
& 
& 262 & 459 & &  
& 5 & 14
&      && 
&&& &&&
14303x & 
1193016x & 
31203x & 
\\

\textsf{terrorRel}
& 
& 881 & 8.6k & & 
& 2 & 3
&      && 
&&& &&&
4938x & 
80274x & 
6182x & 
\\

\textsf{pol-retweet} 
& 
& 18.5k & 48.1k & & 
& 2 & 3
&      && 
&&& &&&
$\infty$ & 
$\infty$ & 
$\infty$ & 
\\

\textsf{web-spam}
&
& 9.1k & 465k & & 
& 3 & 6
&      && 
&&& &&&
$\infty$ & 
$\infty$ & 
$\infty$ & 
\\

\midrule
\textsf{movielens}
& 
& 28.1k & 170.4k & &  
& 3 & 3
&      && 
&&& &&&
$\infty$ & 
$\infty$ & 
$\infty$ & 
\\

\textsf{citeulike} 
& 
& 907.8k & 1.4M & &  
& 3 & 2 
&      && 
&&& &&&
$\infty$ & 
$\infty$ & 
$\infty$ & 
\\

\textsf{yahoo-msg} 
& 
& 100.1k & 739.8k & & 
& 2 & 2
&      && 
&&& &&&
$\infty$ & 
$\infty$ & 
$\infty$ & 
\\

\textsf{dbpedia} 
& 
& 495.9k & 921.7k & & 
& 4 & 3 
&      && 
&&& &&&
$\infty$ & 
$\infty$ & 
$\infty$ & 
\\

\textsf{digg} 
& 
& 217.3k & 477.3k & &  
& 2 & 2 
&      && 
&&& &&&
$\infty$ & 
$\infty$ & 
$\infty$ & 
\\

\textsf{bibsonomy} 
& 
& 638.8k & 1.2M & & 
& 3 & 3 
&      && 
&&& &&&
$\infty$ & 
$\infty$ & 
$\infty$ & 
\\

\textsf{epinions}
& 
& 658.1k & 2.6M & & 
& 2 & 2
&      && 
&&& &&&
$\infty$ & 
$\infty$ & 
$\infty$ & 
\\

\textsf{flickr} 
& 
& 2.3M & 6.8M & & 
& 2 & 2 
&      && 
&&& &&&
$\infty$ & 
$\infty$ & 
$\infty$ & 
\\

\textsf{orkut}
& 
& 6M & 37.4M & & 
& 2 & 2
&      && 
&&& &&&
$\infty$ & 
$\infty$ & 
$\infty$ & 
\\

\midrule

\textsf{ER (10K,0.001)}
& 
& 10k & 50.1k & & 
& 5 & 15
&      && 
&&& &&&
381x & 
11248x & 
502x & 
\\

\textsf{CL (1.8)}
& 
& 9.2k & 44.2k & & 
& 5 & 15
&      && 
&&& &&&
21690x & 
31095x & 
3589x & 
\\

\textsf{KPGM (log 12,14)}
& 
& 3.3k & 43.2k & & 
& 5 & 15
&       && 
&&& &&&
$\infty$ & 
$\infty$ & 
7141x & 
\\

\textsf{SW (10K,6,0.3)}
& 
& 10k & 30k & & 
& 5 & 15
&    && 
&&& &&&
89x & 
21094x & 
862x & 
\\

\bottomrule
\\
\end{tabularx}
\end{table*}

\section{Experiments} \label{sec:exp}
\noindent
The experiments are designed to investigate the runtime performance (Section~\ref{sec:exp-comparison}), space-efficiency (Section~\ref{sec:exp-space-efficiency}), parallelization (Section~\ref{sec:exp-parallel-scaling}), and scalability (Section~\ref{sec:exp-scalability}-\ref{sec:exp-syn-graph-exp}) of the proposed approach.
Results for position-aware typed graphlets are provided in Section~\ref{sec:position-aware-typed-graphlets-results}.
We also demonstrate the utility of \emph{typed graphlets} for two important use cases: 
(i) exploratory analysis/mining (Section~\ref{sec:exploratory-analysis}) and for 
(ii) improving a real-world predictive modeling application (Section~\ref{sec:exp-link-pred}).
To demonstrate the \emph{effectiveness} of the approach, we use a variety of heterogeneous and attributed network data from different application domains.
All data can be accessed at NetworkRepository~\cite{nr}.

\newcommand{\GraphletFigScale}{0.05}
\begin{table}[h!]
\centering
\renewcommand{\arraystretch}{0.95} 
\caption{
Comparing the number of unique \emph{typed graphlets} that occur for each induced subgraph (\eg, there are 40 typed triangles with different type structures in citeseer).
}
\label{table:unique-typed-motif-occur}
\vspace{-3mm}
\small
\begin{tabularx}{1.0\linewidth}{l HrH H cX XX XXXXXX HHH HHHHH@{}}
\toprule
\textbf{Network data}  &   
&  $|E|$  &  && $|\mathcal{T}_{V}|$ & $|\mathcal{T}_{E}|$ &
\includegraphics[scale=0.8]{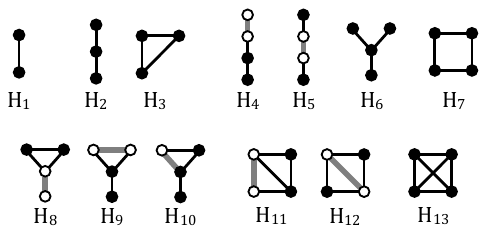} &
\includegraphics[scale=0.8]{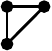} &
\includegraphics[scale=0.15]{fig5.pdf} &
\includegraphics[scale=0.8]{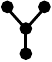} &
\includegraphics[scale=0.8]{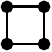} &
\includegraphics[scale=0.15]{fig8.pdf} &
\includegraphics[scale=0.14]{fig9.pdf} &
\includegraphics[scale=0.8]{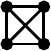} &&
\\
\midrule
\sffamily
\textsf{citeseer} & & 4.5k & & & 6 & 21 & 
56 & 40 & 124 & 119 & 66 & 98 & 56 & 19 & \\

\textsf{cora} & & 5.3k & & & 7 & 28 &
82 & 49 & 202 & 190 & 76 & 157 & 73 & 19 & \\

\textsf{fb-relationship} & & 44.9k & & & 6 & 20 &
50 & 47 & 112 & 109 & 85 & 106 & 89 & 77 & \\

\textsf{web-polblogs} & & 16.7k & & & 2 & 1 &
4 & 4 & 5 & 5 & 5 & 5 & 5 & 5 &  \\ 

\textsf{ca-DBLP} & & 11.3k & & & 3 & 3 & 
10 & 10 & 15 & 15 & 15 & 15 & 15 & 15 & \\

\textsf{inf-openflights} & & 15.7k & & & 2 & 2 &
4 & 4 & 5 & 5 & 5 & 5 & 5 & 5 & \\

\textsf{soc-wiki-elec} & & 100.8k & & & 2 & 2 &
4 & 4 & 5 & 5 & 5 & 5 & 5 & 5 & \\

\textsf{webkb} & & 459 & & & 5 & 14 &
31 & 21 & 59 & 59 & 23 & 51 & 32 & 8 & \\

\textsf{terrorRel} & & 8.6k & & & 2 & 3 &
4 & 4 & 5 & 0 & 4 & 5 & 5 & 5 & \\

\textsf{pol-retweet} &  & 48.1k & & & 2 & 3 & 
4 & 4 & 5 & 5 & 5 & 5 & 5 & 4 & \\

\textsf{web-spam} &  & 465k & & & 3 & 6 &
10 & 10 & 15 & 15 & 15 & 15 & 15 & 15 & \\

\midrule
\textsf{movielens} &  & 170.4k & & & 3 & 3 &
7 & 1 & 6 & 9 & 6 & 3 & 3 & 0 &  \\ 

\textsf{citeulike} & & 1.4M & & & 3 & 2 &
5 & 0 & 3 & 6 & 3 & 0 & 0 & 0 &  \\

\textsf{yahoo-msg} & & 739.8k & & & 2 & 2 &
3 & 2 & 3 & 4 & 3 & 3 & 3 & 2 &  \\

\textsf{dbpedia} & & 921.7k & & & 4 & 3 &
8 & 0 & 6 & 10 & 5 & 0 & 0 & 0 &  \\

\textsf{digg} & & 477.3k & & & 2 & 2 &
4 & 3 & 4 & 5 & 4 & 4 & 4 & 2 &  \\

\textsf{bibsonomy} & & 1.2M & & & 3 & 3 &
7 & 1 & 6 & 9 & 6 & 3 & 3 & 0 & \\

\textsf{epinions} & & 2.6M & & & 2 & 2 &
3 & 2 & 3 & 4 & 3 & 3 & 3 & 2 &  \\

\textsf{flickr} & & 6.8M & & & 2 & 2 &
3 & 2 & 3 & 4 & 3 & 3 & 3 & 2 &  \\

\textsf{orkut} & & 37.4M & & & 2 & 2 &
4 & 3 & 4 & 4 & 3 & 4 & 3 & 2 &  \\

\midrule
\textsf{ER (10K,0.001)} & & 50.1k & & & 5 & 15 &
35 & 30 & 70 & 70 & 69 & 66 & 1 & 0 & \\

\textsf{CL (1.8)} & & 44.2k & & & 5 & 15 &
35 & 35 & 70 & 70 & 70 & 70 & 70 & 68 & \\

\textsf{KPGM (log 12,14)} & & 43.2k & & & 5 & 15 &
35 & 35 & 70 & 70 & 70 & 70 & 70 & 70 & \\

\textsf{SW (10K,6,0.3)} & & 30k & & & 5 & 15 &
35 & 35 & 70 & 70 & 70 & 70 & 70 & 69 & \\

\bottomrule
\end{tabularx}
\end{table}

\subsection{Runtime Comparison} \label{sec:exp-comparison}
\noindent
Since this is the first work to propose and investigate typed graphlets, there are no existing methods for direct comparison.
Nevertheless, we compare the proposed framework to a few recent methods that focus on counting colored graphlets, which is a fundamentally simpler problem.
We first demonstrate how fast the proposed framework is for deriving typed graphlets by comparing the runtime (in seconds) of our approach against three methods for colored graphlets (a similar but simpler problem), namely,  
ESU (using fanmod)~\cite{fanmod}, G-Tries~\cite{ribeiro2014discovering}, and GC~\cite{gu2018heterAlignment}.
Note that these methods do not solve the typed graphlet problem (formally defined in Section~\ref{sec:typed-network-motifs}).
See Figure~\ref{fig:typed-graphlet-vs-colored-graphlet-key-differences} for an intuitive example of the key differences between colored graphlets and our proposed formalization called typed graphlets.
Since these methods are inherently serial (and difficult to parallelize), we use a serial version of the proposed approach for comparison.
We also note that the three methods count colored graphlets for every node whereas the proposed approach derives typed graphlets for every edge.
See Section~\ref{sec:related-work} for other key differences.
Nevertheless, these methods are used for comparison since they are the closest to our own work and solve conceptually simpler problems than the one described in this paper.\footnote{As an aside, we do not compare to methods for counting \emph{untyped graphlets} since these obviously solve a fundamentally different problem and thus are outside the scope of this work.
Furthermore, we also do not focus on graphs with edge types or counting typed graphlets with 5-nodes (or larger) as these are outside the scope of this paper and left for future research.}

For comparison, we use a wide variety of real-world graphs from different domains.
In Table~\ref{table:runtime-perf}, we report the time (in seconds) required by each method.
We also report the speedup obtained from our approach over the other methods in Table~\ref{table:runtime-speedup-perf}.
To be able to compare with the existing methods, we included a variety of very small graphs for which the existing methods could solve in a reasonable amount of time.
Note ETL indicates that a method did not terminate within 24 hours.
Strikingly, the existing methods are unable to handle medium to large graphs with hundreds of thousands or more nodes and edges as shown in Table~\ref{table:runtime-perf}.
Even small graphs can take hours to finish using existing methods (Table~\ref{table:runtime-perf}).
For instance, the small citeseer graph with only 3.3k nodes and 4.5k edges takes 46.27 seconds using the best existing method whereas ours finishes in a tiny fraction of a second, notably, $\nicefrac{2}{100}$ seconds.
This is about 2,100 times faster than the next best method.
Similarly, on the small cora graph with 2.7K nodes and 5.3K edges, GC takes 467 seconds whereas G-Tries takes 351 seconds.
However, our approach finishes counting all typed graphlets with $\{2,3,4\}$-nodes in only 0.03 seconds.
This is 10,000 times faster than the next best method.
Unlike existing methods, our approach is shown to be significantly faster and able to handle large-scale graphs.
The significant speedups obtained by our approach are largely due to the combinatorial relationships between the typed graphlets that we introduce in this work and leverage for deriving many of the typed graphlets in $o(1)$ time.
On flickr, our approach takes about 2 minutes to count the occurrences of all typed graphlets for all 6.8 million edges.
Across all graphs, the proposed method achieves significant speedups over the existing methods as shown in Table~\ref{table:runtime-speedup-perf}.
These results demonstrate the effectiveness of our approach for \emph{counting typed graphlets} in large real-world networks.
As such, the proposed approach brings new opportunities to leverage typed graphlets for real-world applications on much larger networks. 

The results in Table~\ref{table:runtime-perf}-\ref{table:runtime-speedup-perf} demonstrate the effectiveness of the proposed approach on a wide variety of heterogeneous and attributed network data from different application domains.
As an aside, the first 11 graphs in Table~\ref{table:runtime-perf}-\ref{table:runtime-speedup-perf} (\ie, citeseer, cora, fb-relationship, web-polblogs, ca-DBLP, inf-openflights, soci-wiki-elec, webkb, terrorRel, pol-retweet, and web-spam) are all attributed networks as the ``type'' corresponds to different attributes (\eg, political views, paper/research topic, gender, protein function, among others). 
This is in contrast to the other 9 real-world networks in Table~\ref{table:runtime-perf}-\ref{table:runtime-speedup-perf} (\ie, movielens, citeulike, yahoo-msg, dbpedia, digg, bibsonomy, epinions, flickr, orkut) where the types correspond to node or edge types such as users, movies, papers, ratings, among others.

Typed graphlet statistics are shown in Table~\ref{table:unique-typed-motif-occur}.
This includes the number of typed graphlets with nonzero counts for each induced subgraph.
For instance, in cora, there are 49 typed triangle graphlets with nonzero counts out of the 84 possible typed triangle graphlets that could actually occur with $L=7$ types.
From these results, we make an important observation.
In real-world graphs we observe that certain typed graphlets do not occur at all in the graph.
We define such typed graphlets that do not occur in $G$ as \emph{forbidden typed graphlets} as their appearance in the future would indicate something strong.
For instance, perhaps an anomaly or malicious activity.
Other interesting insights and applications of typed graphlets are discussed and explored further in Section~\ref{sec:exploratory-analysis} and Section~\ref{sec:exp-link-pred}.

We also generated synthetic graphs from 4 different graph models including:
Erd\H{o}s-R\'enyi (ER)~\cite{erdos1960evolution},
Chung-Lu (CL)~\cite{chung2002connected},
Kronecker Product Graph Model (KPGM)~\cite{leskovec2010kronecker}, and
Watts-Strogatz Small-World (SW) graph model~\cite{watts1998collective}.
Since these models generate graphs without types, we assign them uniformly at random such that $\frac{N}{L}$ nodes are assigned to every type. 
Unless otherwise mentioned, we set $L=5$.
Results are provided at the bottom of Table~\ref{table:runtime-perf}.
Just as before, we observe significant speedups across all graphs and methods as shown in Table~\ref{table:runtime-speedup-perf}.
Other experiments using synthetic graphs are discussed in Section~\ref{sec:exp-syn-graph-exp}.

\subsection{Space Efficiency Comparison} \label{sec:exp-space-efficiency}
\noindent
We theoretically showed the space complexity of our approach in Section~\ref{sec:space-complexity}.
In this section, we empirically investigate the space-efficiency of our approach compared to ESU (using fanmod)~\cite{fanmod}, G-Tries~\cite{ribeiro2014discovering}, and GC~\cite{gu2018heterAlignment}.
Table~\ref{table:space-results} reports the space used by each method for a variety of real-world graphs.
Strikingly, the proposed approach uses between 42x and 776x less space than existing methods as shown in Table~\ref{table:space-results}.
These results indicate that our approach is space-efficient and practical for large networks.
As an aside, typed graphlet statistics are also shown in Table~\ref{table:unique-typed-motif-occur}.

\begin{table}[h!]
\centering
\setlength{\tabcolsep}{4pt}
\renewcommand{\arraystretch}{1.2} 
\caption{
Comparing the \emph{space} used by the proposed typed graphlet approach.
Since there are no existing methods for typed graphlet, we compare against colored graphlet methods that solve a simpler problem.
}
\vspace{-2mm}
\label{table:space-results}
\small
\begin{tabularx}{1.00\linewidth}{lH @{} XX XX}
\toprule
& &
\multicolumn{1}{X}{\rotatebox{0}{\textsf{citeseer}}} & 
\multicolumn{1}{X}{\rotatebox{0}{\textsf{cora}}} & 
\multicolumn{1}{X}{\rotatebox{0}{\textsf{movielens}}} & 
\multicolumn{1}{X}{\rotatebox{0}{\textsf{web-spam}}}
\\ \midrule

\textbf{GC}
& &
30.1MB & 50.4MB & 
ETL & 
ETL 
\\

\textbf{ESU}
& & 
13.4MB & 46.2MB & 
ETL &
ETL 
\\

\textbf{G-Tries}
& & 
161.9MB & 448.6MB & 
ETL & 
ETL 
\\

\midrule
\textbf{Typed graphlets}
& &
\textbf{316KB} & \textbf{578KB} &
\textbf{22.5MB} &
\textbf{128.9MB} 
\\

\textbf{Position-aware typed graphlets}
& &
\textbf{417KB} & \textbf{806KB} &
\textbf{32.2MB} &
\textbf{192.1MB} 
\\

\bottomrule
\multicolumn{6}{l}{\footnotesize
$^{*}$ ETL = Exceeded Time Limit (24 hours / 86,400 seconds)} 
\\
\end{tabularx}
\end{table}

\begin{figure}[t!]
\centering
\includegraphics[width=0.40\linewidth]{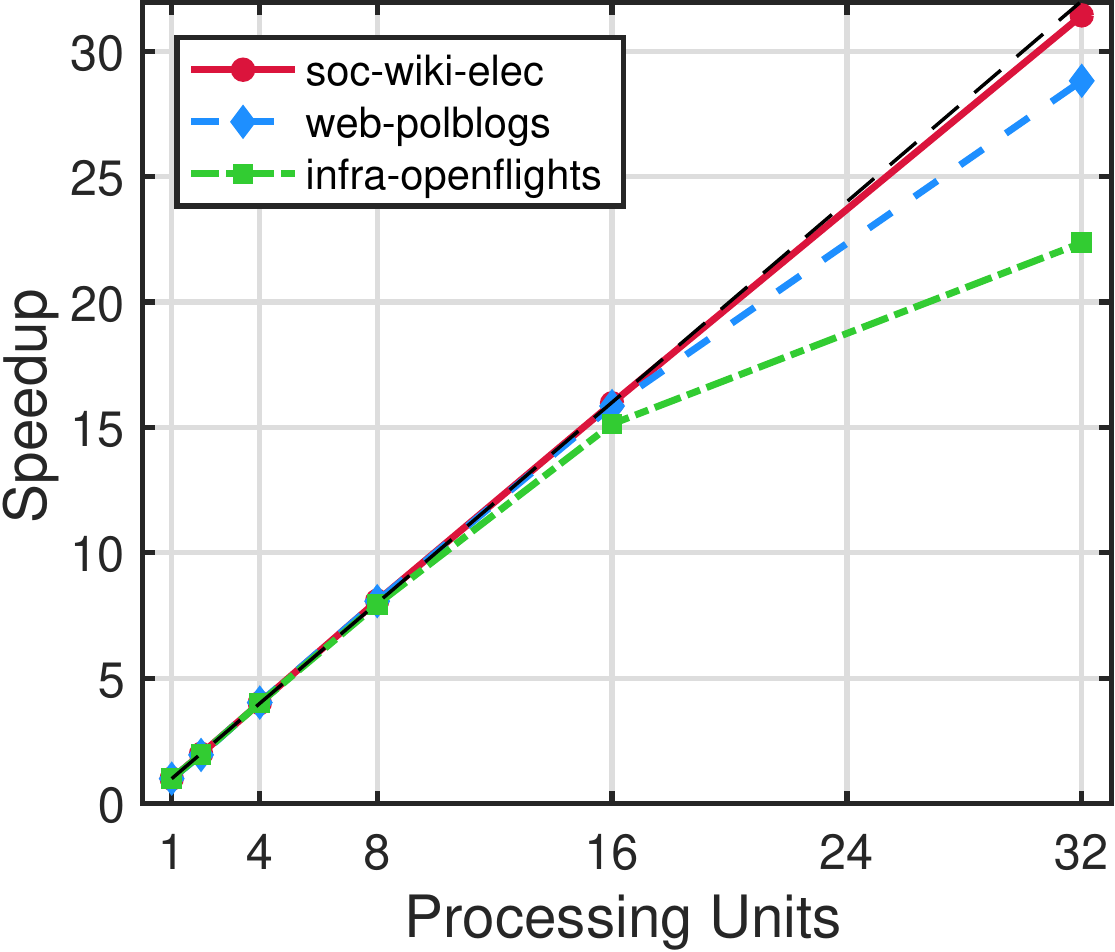}
\vspace{-2mm}
\caption{Parallel speedup of the proposed approach.
Notably, the approach exhibits nearly linear parallel scaling as the number of cores increases.
}
\label{fig:parallel-scaling}
\end{figure}

\subsection{Parallel Speedup} \label{sec:exp-parallel-scaling}
\noindent
This section evaluates the parallel scaling of the proposed approach.
As an aside, this work describes the first parallel approach for typed graphlet counting.
In these experiments, we used a two processor, Intel Xeon E5-2686 v4 system with 256 GB of memory.
None of the experiments came close to using all the memory.
Parallel speedup is simply $S_p = \frac{T_1}{T_p}$ where $T_1$ is the execution time of the sequential algorithm, and $T_p$ is the execution time of the parallel algorithm with $p$ processing units (cores).
In Figure~\ref{fig:parallel-scaling}, we observe nearly linear speedup as we increase the number of cores.
These results indicate the effectiveness of the parallel algorithm for counting typed graphlets in general heterogeneous graphs.

\subsection{Scalability} \label{sec:exp-scalability}
\noindent
To evaluate the scalability of the proposed framework as the size of the graph grows (\ie, number of nodes and edges increase),
we generate Erd\"{o}s-R\'{e}nyi graphs of increasing size (from 100 to 1 million nodes) such that each graph has an average degree of 10.
In Figure~\ref{fig:exp-runtime-ER}, we observe that our approach scales linearly as the number of nodes and edges grow large.
As an aside, our approach takes less than 2 minutes to derive all typed $\{2,3,4\}$-node graphlets for a large graph with 1 million nodes and 10 million edges. 
Note that existing methods are not shown in Figure~\ref{fig:exp-runtime-ER} since they are unable to handle medium to large-sized graphs as shown previously in Table~\ref{table:runtime-perf}.

\begin{figure}[h!]
\centering
\includegraphics[width=0.4\linewidth]{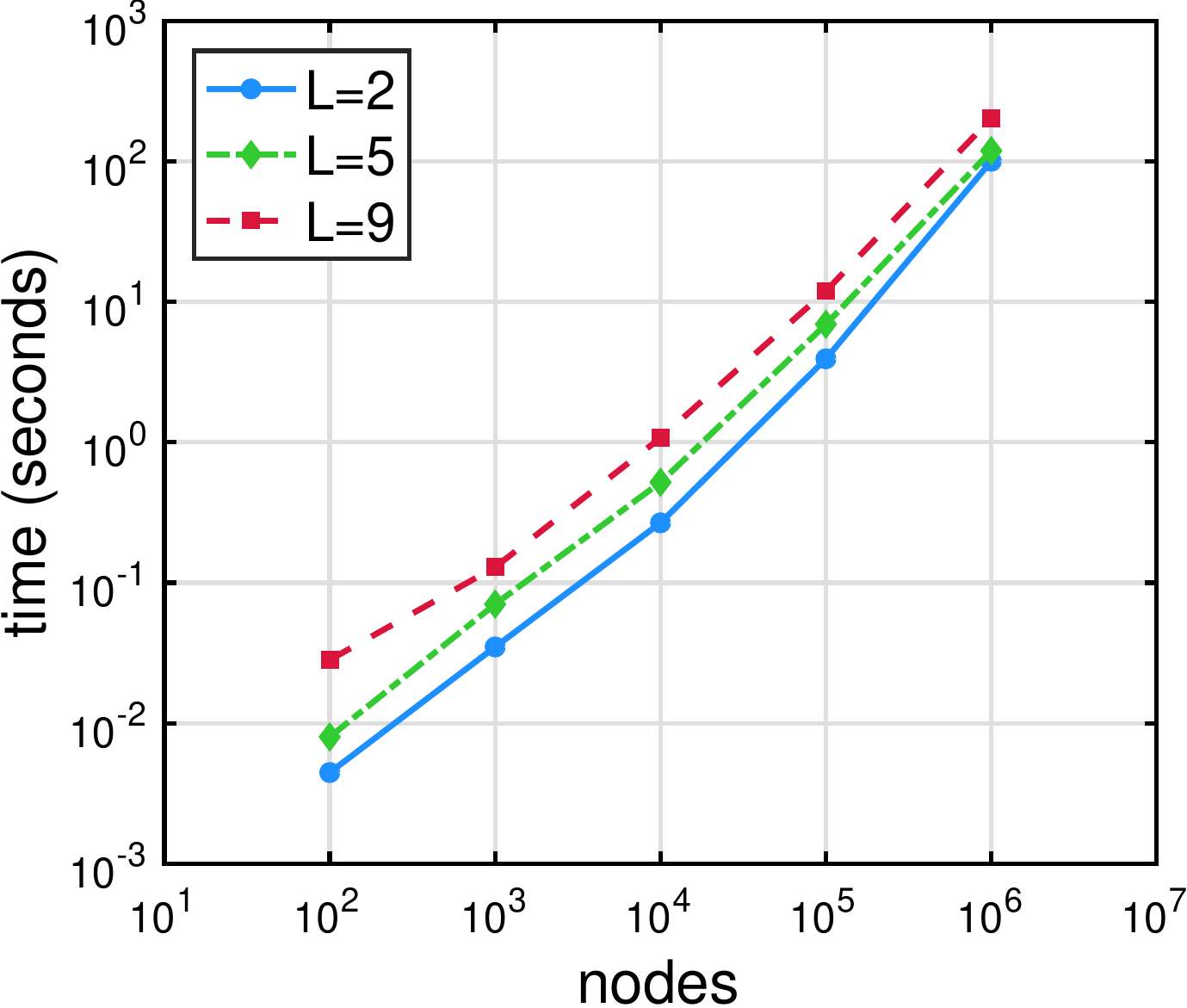}
\vspace{-2mm}
\caption{
Scalability of the proposed framework as the size of the graph increases.
Erd\protect\"{o}s-R\protect\'{e}nyi graphs with an average degree of 10 are used.
}
\label{fig:exp-runtime-ER}
\end{figure}

\subsection{Synthetic Graph Experiments} \label{sec:exp-syn-graph-exp}
In these experiments, we generate synthetic graphs.
For each graph, we vary the number of node types $L$ from 2 to 9, and measure the runtime performance as the number of node types increases as well as the impact in terms of space as $L$ increases.
Given $L \in \{2,\ldots,9\}$ node types, we assign types to nodes uniformly at random such that $\frac{N}{L}$ nodes are assigned to every type.

\begin{table}[h!]
\centering
\renewcommand{\arraystretch}{1.2} 
\setlength{\tabcolsep}{3.5pt}
\caption{Comparing the number of unique typed graphlets that occur for each induced subgraph as we vary the number of types $L$ using a KPGM graph with 3.3k nodes, 43.2k edges, average degree = 26, and max degree = 1.3K.
Speedup is shown in parenthesis.
}
\label{table:unique-typed-motif-occur-syn-KPGM}
\vspace{-3mm}
\small
\begin{tabularx}{0.72\linewidth}{Hc cc cccccc c H 
cc H 
HHHHHH@{}
}
\toprule

&& && &&&&&&&& \multicolumn{2}{c}{
\vspace{-3mm}
\textbf{time} (sec.)\quad\quad\;\;\;\;
}
\\

& $L$  &
\includegraphics[scale=0.8]{fig3.pdf} &
\includegraphics[scale=0.8]{fig4.pdf} &
\includegraphics[scale=0.15]{fig5.pdf} &
\includegraphics[scale=0.8]{fig6.pdf} &
\includegraphics[scale=0.8]{fig7.pdf} &
\includegraphics[scale=0.15]{fig8.pdf} &
\includegraphics[scale=0.14]{fig9.pdf} &
\includegraphics[scale=0.8]{fig10.pdf} &&&
\textbf{serial} & 
\textbf{parallel} - 4 cores & 
\\
\midrule

& \textbf{2} & 4 & 4 & 5 & 5 & 5 & 5 & 5 & 5 &&&
8.11 & 2.08 \;(\emph{3.89x}) \\

& \textbf{5} & 35 & 35 & 70 & 70 & 70 & 70 & 70 & 70 &&&
8.94 & 2.26 \;(\emph{3.95x})  \\

& \textbf{9} & 165 & 165 & 495 & 495 & 495 & 495 & 495 & 495 &&&
10.37 & 2.62 \;(\emph{3.95x})  \\

\bottomrule
\end{tabularx}
\end{table}

\subsubsection{Impact on Performance}
We first investigate the runtime performance of our approach as the number of types $L$ increases from 2 to 9.
We use both a serial and parallel implementation of our method for comparison.
Results are shown in Table~\ref{table:unique-typed-motif-occur-syn-KPGM}.
Notably, the parallel speedup of the parallel algorithm is constant regardless of $L$.
Therefore, it is not impacted by the increase in $L$.
Furthermore, the runtime of both the serial and parallel algorithm increases slightly as $L$ increases. 
Notice the additional work depends on the number of unique typed graphlets (the sum of columns 2-9 in Table~\ref{table:unique-typed-motif-occur-syn-KPGM}) and not directly on $L$ itself. 
The total amount of unique typed graphlets substantially increases as $L$ increases from $2$ to $9$ as shown in Table~\ref{table:unique-typed-motif-occur-syn-KPGM}.
This is primarily due to the random assignment of types to nodes.
However, in sparse real-world graphs the total unique typed graphlets is typically much smaller as shown in Table~\ref{table:unique-typed-motif-occur}.

\begin{figure}[h!]
\centering
\hspace{-2mm}
\subfigure[Runtime performance as the $\#$ of types increases]{
\includegraphics[width=0.4\linewidth]{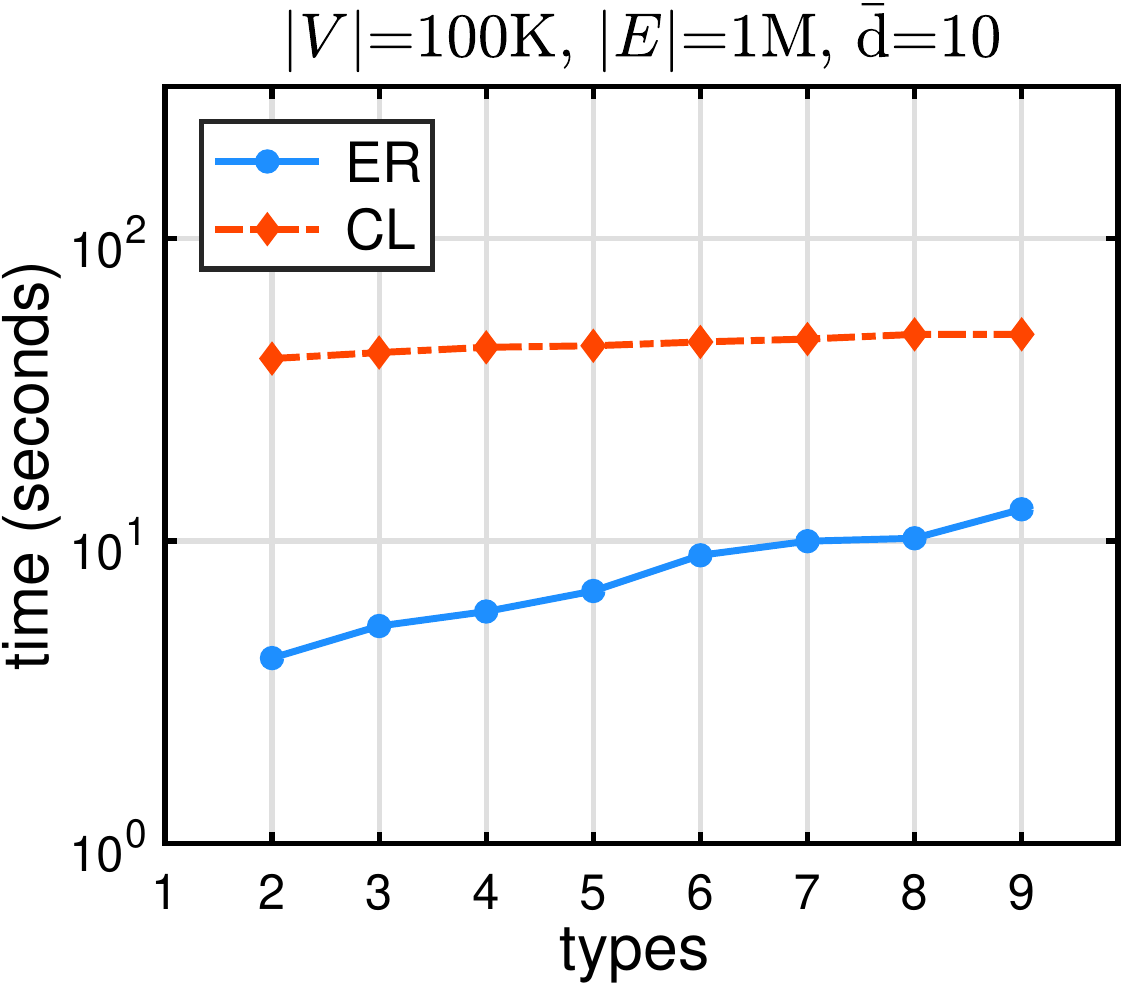}
\label{fig:runtime-ER-CL-vs-varyingNumTypes}
}
\hspace{8mm}
\subfigure[Space (MB) as the $\#$ of types increases]{
\includegraphics[width=0.4\linewidth]{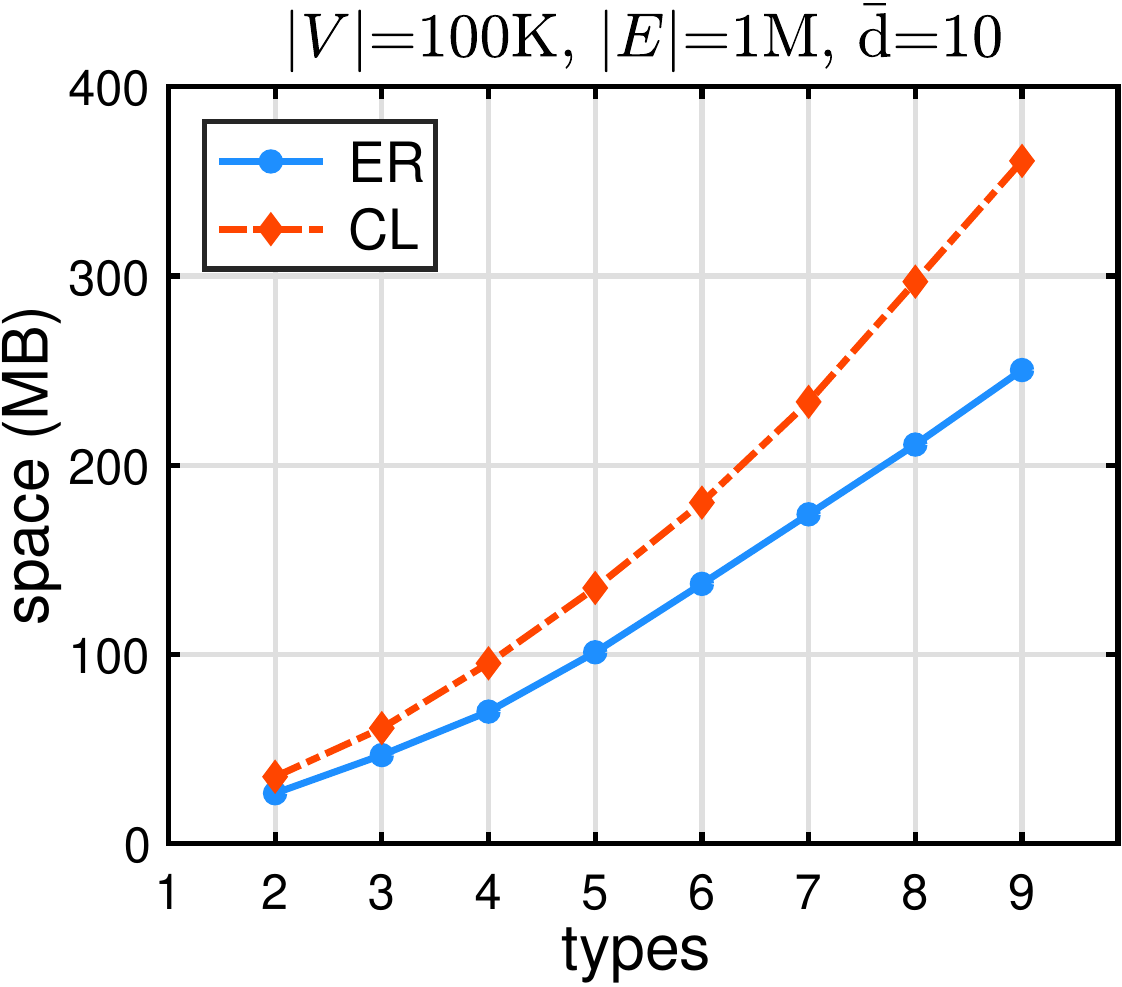}
\label{fig:space-ER-CL-vs-varyingNumTypes}
}

\caption{Comparing runtime performance \emph{and} space as the number of types increases.
\textbf{Left}: In the case of CL graphs with a skewed degree distribution, the runtime is nearly constant as the number of types increases.
\textbf{Right}: Both ER and CL have similar space requirements.
Since node types are randomly assigned to nodes, this represents a type of worst-case since every edge in $G$ is likely to have significantly more distinct typed graphlets compared to sparse real-world graphs (as shown in Table~\ref{table:unique-typed-motif-occur} and Section~\ref{sec:exploratory-analysis}).
Note $\bar{\mathrm{d}} = \frac{1}{n} \sum d_i$.
}
\label{fig:runtime-space-ER-CL-vs-varyingNumTypes}
\end{figure}

To further understand how the structure of the graph impacts runtime, we generate an ER and Chung-Lu (CL) graph with 100K nodes, 1M edges, and average degree 10.
We vary the number of types $L$ and assign types to nodes uniformly at random as discussed previously.
Notice that both the ER and CL graph are generated such that they each have 100K nodes, 1 million edges, with average degree 10.
However, both graphs are structurally very different.
For instance, the degrees among the nodes in the ER graph are more uniform whereas the degree of the nodes in CL are skewed such that a few nodes have very large degree while the others have relatively small degree.
We observe in Figure~\ref{fig:runtime-ER-CL-vs-varyingNumTypes} that for CL graphs with a skewed degree distribution, the runtime of the approach as the number of types increases is essentially constant.
This result is important as most real-world graphs also have a skewed degree distribution (social networks, web graphs, information networks, etc.)~\cite{Girvan2002,Faloutsos1999}.
However, even in the case where the degrees are more uniform across the nodes, our approach still performs well as shown in Figure~\ref{fig:runtime-ER-CL-vs-varyingNumTypes}.

\subsubsection{Impact on Space}
Figure~\ref{fig:space-ER-CL-vs-varyingNumTypes} shows the memory (space) required by our approach as the number of types increases from $L\in \{2,\ldots,9\}$.
Both ER and CL graphs are shown to have similar space requirements. 
This is likely due to the random assignment of types to nodes.
This assignment represents a type of worst case since every edge is likely to have significantly more distinct typed graphlets compared to sparse real-world graphs.
This difference can be seen in Table~\ref{table:unique-typed-motif-occur}.

\begin{table}[b!]
\centering
\renewcommand{\arraystretch}{0.95} 
\caption{
Comparing the number of unique \emph{position-aware typed graphlets} that occur for each induced subgraph.
}
\label{table:unique-position-aware-typed-motif-occur}
\vspace{-3mm}
\small
\begin{tabularx}{1.0\linewidth}{l HrH H cX XX XXXXXX HHH HHHHH@{}}
\toprule
\textbf{Network data}  &   
&  $|E|$  & && $|\mathcal{T}_{V}|$ & $|\mathcal{T}_{E}|$ &
\includegraphics[scale=0.8]{fig3.pdf} &
\includegraphics[scale=0.8]{fig4.pdf} &
\includegraphics[scale=0.15]{fig5.pdf} &
\includegraphics[scale=0.8]{fig6.pdf} &
\includegraphics[scale=0.8]{fig7.pdf} &
\includegraphics[scale=0.15]{fig8.pdf} &
\includegraphics[scale=0.14]{fig9.pdf} &
\includegraphics[scale=0.8]{fig10.pdf} &&
\\
\midrule
\sffamily

\textsf{citeseer} & & 4.5k && & 6 & 21 & 
212 & 114 & 1090 & 885 & 278 & 663 & 243 & 49 & \\

\textsf{cora} & & 5.3k && & 7 & 28 & 
299 & 132 & 1839 & 1626 & 289 & 1129 & 334 & 52 & \\

\textsf{fb-relationship} & & 44.9k && & 6 & 20 & 
199 & 180 & 1148 & 1105 & 761 & 1113 & 895 & 703 & \\

\textsf{web-polblogs} & & 16.7k && & 2 & 1 & 
8 & 8 & 16 & 16 & 16 & 16 & 16 & 16 & \\

\textsf{ca-DBLP} & & 11.3k && & 3 & 3 & 
27 & 27 & 81 & 81 & 81 & 81 & 81 & 81 & \\

\textsf{infra-openflights} & & 15.7k && & 2 & 2 & 
8 & 8 & 16 & 16 & 16 & 16 & 16 & 16 & \\

\textsf{soc-wiki-elec} & & 100.8k && & 2 & 2 & 
8 & 8 & 16 & 16 & 16 & 16 & 16 & 16 & \\

\textsf{webkb} & & 459 && & 5 & 14 & 
97 & 63 & 444 & 396 & 98 & 382 & 209 & 44 & \\

\textsf{terrorRel} & & 8.6k && & 2 & 3 & 
8 & 8 & 16 & 0 & 10 & 16 & 16 & 16 &  \\

\textsf{pol-retweet} & & 48.1k && & 2 & 3 & 
8 & 8 & 16 & 16 & 16 & 16 & 16 & 8 & \\

\textsf{web-spam} & & 465k && & 3 & 6 & 
27 & 27 & 81 & 81 & 81 & 81 & 81 & 81 & \\

\midrule
\textsf{movielens} & & 170.4k && & 3 & 3 & 
18 & 6 & 42 & 42 & 18 & 30 & 18 & 0 &  \\

\textsf{citeulike} & & 1.4M && & 3 & 1 & 
12 & 0 & 16 & 20 & 8 & 0 & 0 & 0 &  \\

\textsf{yahoo-msg} & & 739.8k && & 2 & 2 & 
4 & 3 & 7 & 8 & 5 & 7 & 6 & 3 & \\

\textsf{dbpedia} & & 921.7k && & 4 & 3 & 
20 & 0 & 36 & 36 & 14 & 0 & 0 & 0 & \\

\textsf{digg} & & 477.3k && & 2 & 2 & 
6 & 4 & 10 & 12 & 7 & 8 & 5 & 3 &  \\

\textsf{bibsonomy} & & 1.2M && & 3 & 3 & 
18 & 6 & 42 & 42 & 18 & 30 & 18 & 0 & \\

\textsf{epinions} & & 2.6M && & 2 & 2 & 
6 & 4 & 10 & 12 & 7 & 10 & 8 & 4 &  \\

\textsf{flickr} & & 6.8M && & 2 & 2 & 
6 & 4 & 10 & 12 & 7 & 10 & 8 & 4 &  \\

\textsf{orkut} & & 37.4M && & 2 & 2 & 
8 & 6 & 15 & 15 & 7 & 14 & 10 & 5 & \\

\midrule

\textsf{ER (10K,0.001)} & & 50.1k && & 5 & 15 & 
125 & 117 & 625 & 625 & 623 & 623 & 3 & 0 & \\

\textsf{CL (1.8)} & & 44.2k && & 5 & 15 & 
125 & 125 & 625 & 625 & 625 & 625 & 625 & 625 & \\

\textsf{KPGM (log 12,14)} & & 43.2k && & 5 & 15 & 
125 & 125 & 625 & 625 & 625 & 625 & 625 & 625 & \\

\textsf{SW (10K,6,0.3)} & & 30k && & 5 & 15 & 
125 & 125 & 625 & 625 & 625 & 625 & 625 & 623 & \\

\bottomrule
\end{tabularx}
\end{table}

\subsection{Position-Aware Typed Graphlet Results} \label{sec:position-aware-typed-graphlets-results}
In these experiments, we investigate position-aware typed graphlets introduced formally in Section~\ref{sec:position-aware-typed-graphlets}.
The difference between typed graphlets (Def.~\ref{def:typed-graphlet-instance}) and \emph{position-aware typed graphlets} (Def.~\ref{def:position-aware-typed-graphlet-instance}) is shown in Figure~\ref{fig:position-aware-typed-graphlets} using an intuitive example.
In Table~\ref{table:unique-position-aware-typed-motif-occur}, we show the number of unique position-aware typed graphlets that occur for each induced subgraph using both synthetic and real-world graphs from a wide range of domains.
Note the first eleven graphs in Table~\ref{table:unique-position-aware-typed-motif-occur} are attributed graphs whereas the next nine graphs are heterogeneous networks.

\begin{table*}[h!]
\centering
\caption{
Runtime results (in seconds) for counting position-aware typed graphlets (Def.~\ref{def:position-aware-typed-graphlet-instance}) compared to typed graphlets (Def.~\ref{def:typed-graphlet-instance}).
Best result is bold.
Note $\Delta=$ max node degree; $|\mathcal{T}_V|=$ number of node types.
}
\vspace{-3mm}
\label{table:runtime-perf-position-aware}
\renewcommand{\arraystretch}{1.10} 
\renewcommand{\arraystretch}{1.05} 
\small
\setlength{\tabcolsep}{6.5pt} 
\begin{tabularx}{1.0\linewidth}{@{}
r H lllH cH HH
HHH H
ll@{}
H
H 
HHH H
@{}
}
\toprule

& 
& 
$|V|$ & $|E|$ & $\Delta$ & 
&
\;$|\mathcal{T}_V|$\;\;  & 
& &&
&&& 
& 
\textbf{Typed graphlets} &
\textbf{Position-aware (Def.~\ref{def:position-aware-typed-graphlet-instance})} & 
\\
\midrule

\textsf{citeseer}
& 
& 3.3k & 4.5k & 99 & 
& 6 & 
&      && 
&    &&  
& 
\text{0.022} & \textbf{0.020} &
\\

\textsf{cora}
& 
& 2.7k & 5.3k & 168 & 
& 7 & 
&      && 
&    && 
& 
\text{0.032} & \textbf{0.031} & 
\\

\textsf{fb-relationship} 
& 
& 7.3k & 44.9k & 106 & 
& 6 & 
&      && 
&    && 
& 
\textbf{0.701} & \text{0.832} & 
\\

\textsf{web-polblogs}
& 
& 1.2k & 16.7k & 351 & 
& 2 & 
&      && 
&    && 
& 
\text{1.055} & \textbf{1.042} &
\\

\textsf{ca-DBLP}
& 
& 2.9k & 11.3k & 69 &
& 3 & 
&      && 
&    && 
& 
\textbf{0.100} & \text{0.115} & 
\\

\textsf{inf-openflights}
& 
& 2.9k & 15.7k & 242 & 
& 2 & 
&      && 
&    && 
& 
\text{0.578} & \textbf{0.562} &
\\

\textsf{soc-wiki-elec}
& 
& 7.1k & 100.8k & 1.1k & 
& 2 & 
&      && 
&    && 
& 
\text{5.316} & \textbf{4.939} & 
\\

\textsf{webkb}
& 
& 262 & 459 & 122 &  
& 5 & 
&      && 
&    && 
& 
\text{0.006} & \textbf{0.005} & 
\\

\textsf{terrorRel}
& 
& 881 & 8.6k & 36 & 
& 2 & 
&      && 
&    && 
& 
\textbf{0.039} & \text{0.048} & 
\\

\textsf{pol-retweet} 
& 
& 18.5k & 48.1k & 786 & 
& 2 & 
&      && 
&    && 
&
\text{0.296} & \textbf{0.289} &
\\

\textsf{web-spam}
&
& 9.1k & 465k & 3.9k & 
& 3 & 
&      && 
&    && 
&
\text{210.97} & \textbf{207.36} & 
\\

\midrule
\textsf{movielens}
& 
& 28.1k & 170.4k & 3.6k & 
& 3 & 
&      && 
&    && 
& 
\text{5.23} & \textbf{5.12} & 
\\

\textsf{citeulike} 
& 
& 907.8k & 1.4M & 11.2k & 
& 3 &  
&      && 
&    && 
& 
\text{126.53} & \textbf{125.66} &
\\

\textsf{yahoo-msg} 
& 
& 100.1k & 739.8k & 9.4k & 
& 2 & 
&      && 
&    && 
&
\text{35.22} & \textbf{35.08} & 
\\

\textsf{dbpedia} 
& 
& 495.9k & 921.7k & 24.8k & 
& 4 & 
&      && 
&    && 
& 
\text{56.02} & \textbf{53.36} &
\\

\textsf{digg} 
& 
& 217.3k & 477.3k & 219 & 
& 2 & 
&      && 
&    && 
& 
\text{5.592} & \textbf{5.578} & 
\\

\textsf{bibsonomy} 
& 
& 638.8k & 1.2M & 211 & 
& 3 &  
&      && 
&    && 
& 
\text{3.631} & \textbf{3.607} & 
\\

\textsf{epinions}
& 
& 658.1k & 2.6M & 775 &
& 2 & 
&      && 
&    && 
& 
\text{85.27} & \textbf{85.05} &
\\

\textsf{flickr} 
& 
& 2.3M & 6.8M & 216 & 
& 2 & 
&      && 
&    && 
& 
\text{120.79} & \textbf{112.45} &
\\

\textsf{orkut}
& 
& 6M & 37.4M & 166 & 
& 2 & 
&      && 
&    && 
& 
\text{1241.01} & \textbf{1236.21} & 
\\

\midrule

\textsf{ER (10K,0.001)}
& 
& 10k & 50.1k & 26 & 
& 5 & 
&      && 
&    && 
& 
\textbf{0.48} & \text{0.59} & 
\\

\textsf{CL (1.8)}
& 
& 9.2k & 44.2k & 218 & 
& 5 & 
&      && 
&    && 
& 
\textbf{1.46} & \text{1.65} & 
\\

\textsf{KPGM (log 12,14)}
& 
& 3.3k & 43.2k & 1.3k & 
& 5 & 
&       && 
&    && 
& 
\text{8.94} & \textbf{8.56} & 
\\

\textsf{SW (10K,6,0.3)}
& 
& 10k & 30k & 12 & 
& 5 & 
&    && 
&    && 
& 
\textbf{0.24} & \text{0.29} & 
\\

\bottomrule
\end{tabularx}
\end{table*}

\begin{figure}[b!]
\centering

\includegraphics[width=0.45\linewidth]{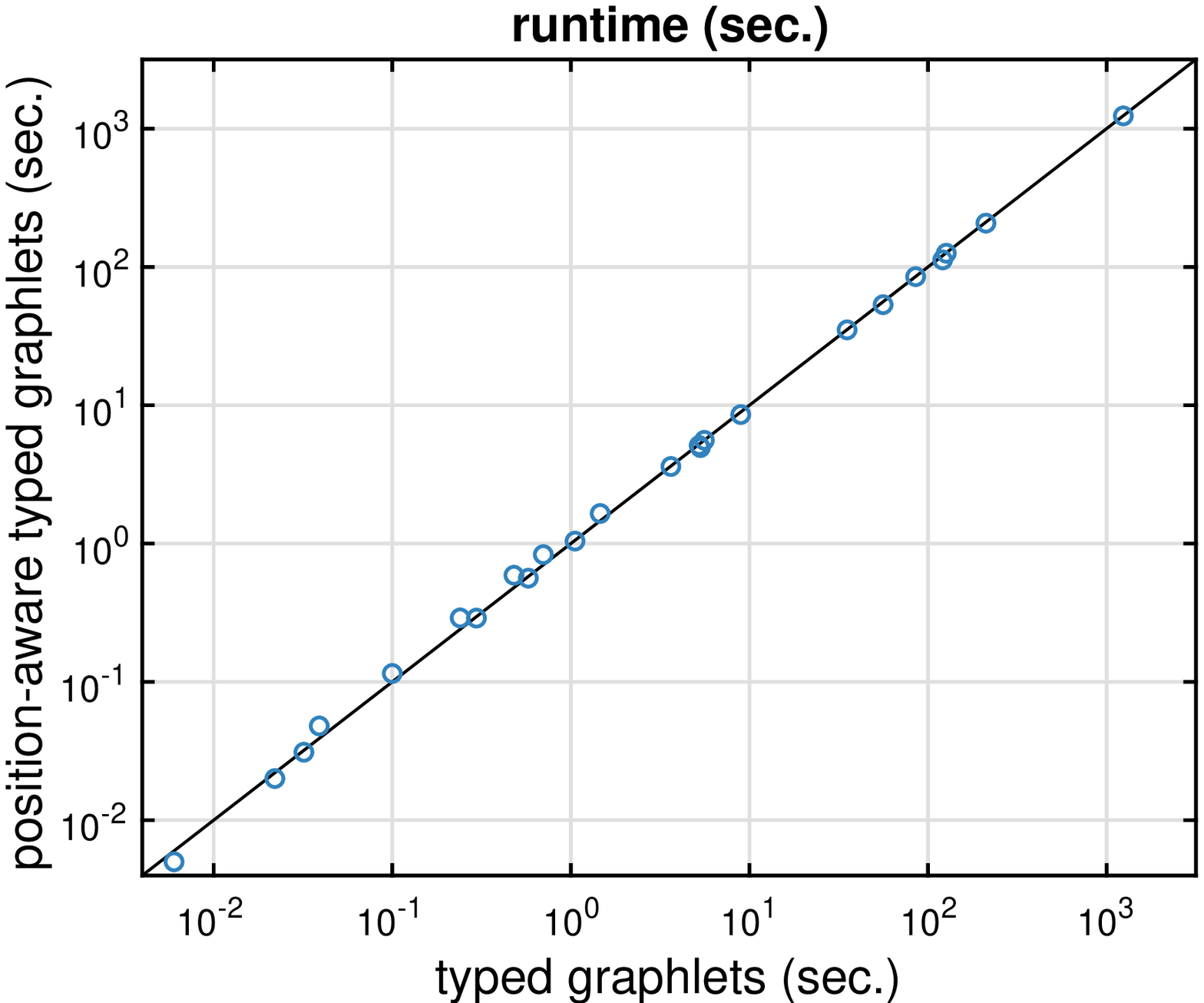}
\vspace{-2mm}
\caption{
Runtime (in seconds) for typed graphlets and position-aware typed graphlets. 
See text for discussion.
}
\label{fig:runtime-typed-graphlet-vs-position-aware}
\end{figure}

Runtime results for position-aware typed graphlets is shown in Table~\ref{table:runtime-perf-position-aware}.
We compare the runtime of position-aware typed graphlets to typed graphlets.
Notably, in all cases, the runtime of position-aware typed graphlets and typed graphlets is very close as shown in Table~\ref{table:runtime-perf-position-aware}.
Recall from Section~\ref{sec:position-aware-typed-graphlets} that this is what is expected since the 
algorithmic difference is trivial.
To better understand the differences in runtime, we show the runtime in seconds of typed graphlets (x-axis) vs. position-aware typed graphlets (y-axis) for all the graphs in Table~\ref{table:runtime-perf-position-aware}.
Note the diagonal line in Figure~\ref{fig:runtime-typed-graphlet-vs-position-aware} represents the expected runtime assuming that typed graphlets and position-aware typed graphlets have the same runtime performance.
In Figure~\ref{fig:runtime-typed-graphlet-vs-position-aware}, most of the graphs lie on the line, while a few have minor deviations in either direction.
Position-aware typed graphlets are typically faster to compute for most of the real-world graphs as shown in Table~\ref{table:runtime-perf-position-aware}.
However, in terms of the four synthetic graphs in Table~\ref{table:runtime-perf-position-aware} (last four graphs), we find that typed graphlets is faster for three of the four with the KPGM graph being the exception.
We also note that compared to the other three methods for counting the simpler notion of colored graphlets, both typed graphlets and position-aware typed graphlets remain significantly faster.

In Table~\ref{table:space-results}, we also report the space used by position-aware typed graphlets.
Note that by definition position-aware typed graphlets use at least as much space as typed graphlets, and typically use more space than typed graphlets as shown in Table~\ref{table:space-results}.
Despite that typed graphlets and position-aware typed graphlets are more complex than colored graphlets (and theoretically should use less space), the proposed framework for both uses significantly less space than these other methods.

\begin{figure}[h!]
\centering
\subfigure{\includegraphics[width=0.40\linewidth]{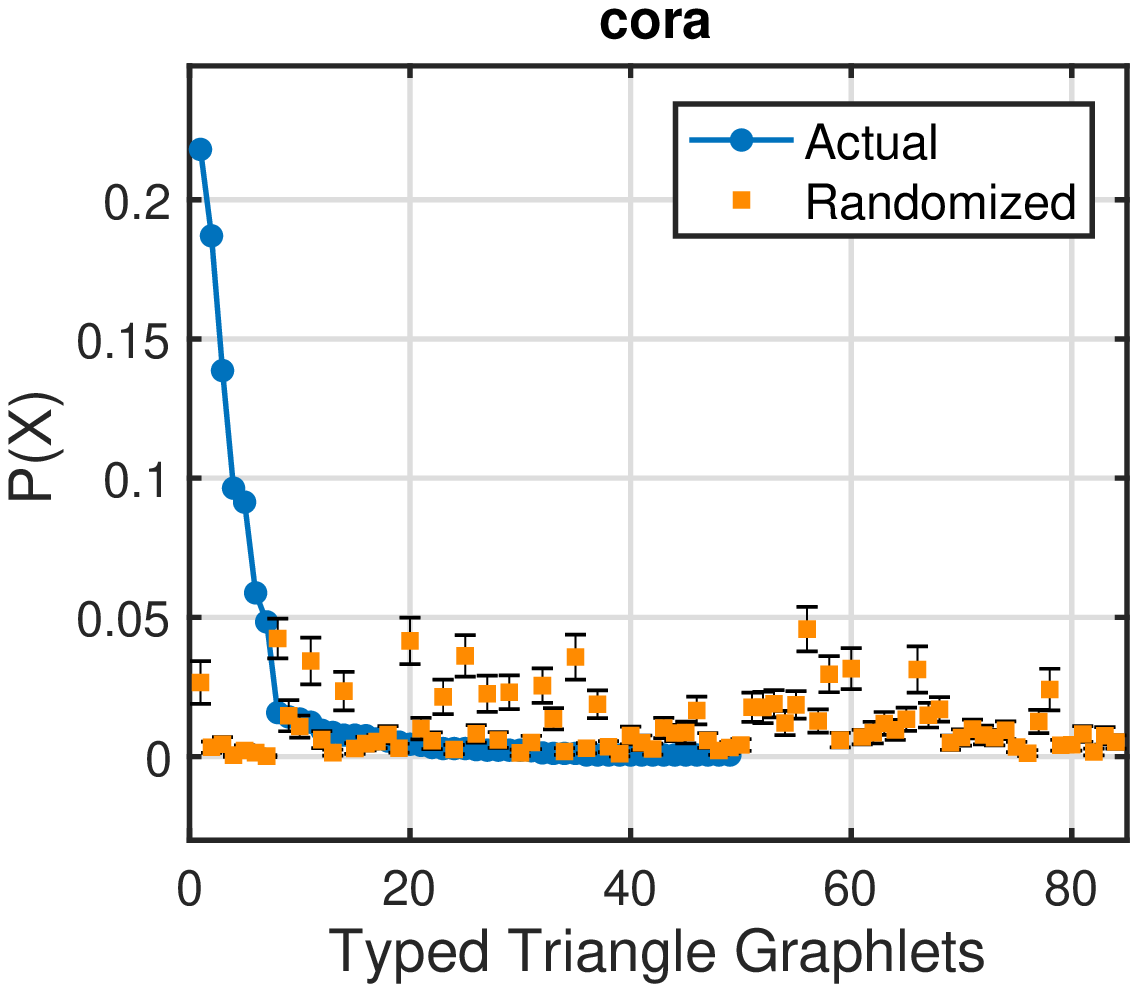}}
\hspace{8mm}
\subfigure{\includegraphics[width=0.405\linewidth]{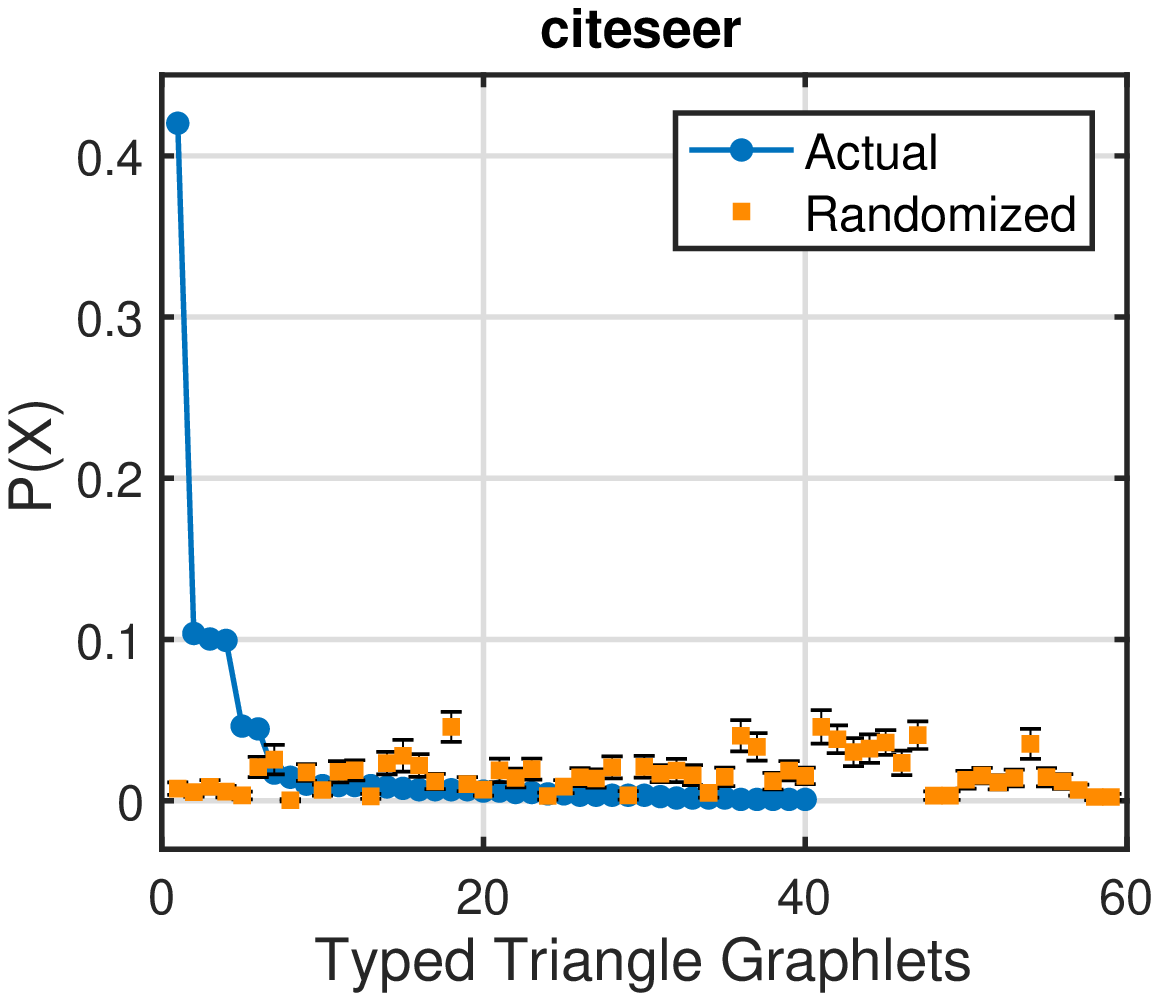}}

\vspace{-3mm}
\caption{
Comparing the actual typed triangle distribution to the randomized typed triangle distribution.
We compute 100 random permutations of the node types and run the approach on each permutation then average the resulting counts to obtain the mean randomized typed triangle distribution.
There are three key findings.
First, we observe a significant difference between the actual and randomized typed triangle distributions.
Second, many of the typed triangles that occur when the types are randomized, do not occur in the actual typed triangle distribution.
Third, we find the typed triangle distribution to be skewed 
as a few typed triangles occur very frequently while the vast majority have very few occurrences.
}
\label{fig:typed-tri-prob-dist-cora-citeseer}
\end{figure}

\begin{figure}[h!]
\centering
\includegraphics[width=0.9\linewidth]{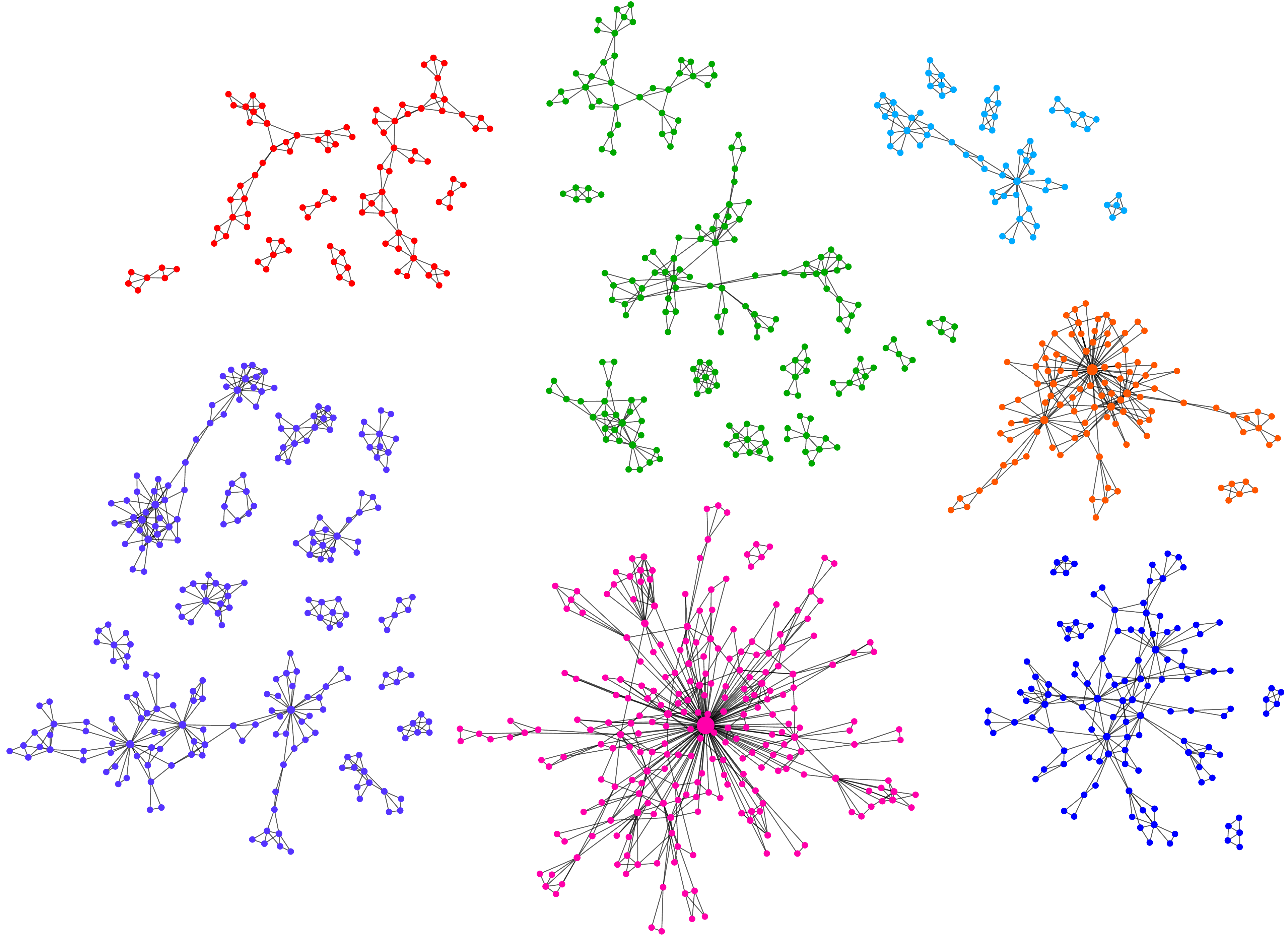}
\vspace{-2mm}
\caption{
A common prediction task in the cora citation network is to predict the research area (type) of a paper (node).
We visualize the network resulting from the edges with nonzero typed triangle counts and find something striking.
Typed triangles shatter the graph into many different components that are tightly connected and overwhelmingly homogeneous with respect to research area (type/label) of the nodes.
This can be used to filter noisy links from the graph to improve classification performance.
Node color encodes the research area of the papers.
}
\label{fig:cora-typed-motifs-homogeneous-types}
\end{figure}

\subsection{Exploratory Analysis} \label{sec:exploratory-analysis}
\noindent
This section demonstrates the use of heterogeneous graphlets for mining and exploratory analysis.

\subsubsection{Political retweets}
The political retweets data consists of 18,470 Twitter users classified into 2 types that encode the users political leanings (\ie, left, right). The graph has 61,157 links representing retweets.
There are 24,815 triangles in the political retweet network.
Triangles in this graph indicate that users retweeted by an individual also retweet each other (\ie, triangle = three users that have all mutually retweeted each other).
Triangles may represent users with similar interests.
However, triangles alone do not reveal any additional information about the users.
Another interesting question is as follows:
are users with a particular political leaning more likely to form retweet triangles with users of the same political leaning or vice-versa?
Unfortunately, untyped triangles alone cannot be used to answer such questions.
To answer such questions, typed graphlets are used by encoding the political leanings of a Twitter user as the type.\footnote{Typed graphlets can be used with any attribute.}
Interestingly, the 24,815 (untyped) triangles are distributed as follows:
\begin{center}
\begin{tabular}{l r@{} cccc @{}l r}
&&
\includegraphics[width=5mm]{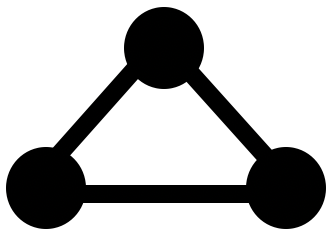} &
\includegraphics[width=5mm]{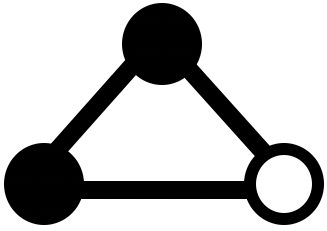} &
\includegraphics[width=5mm]{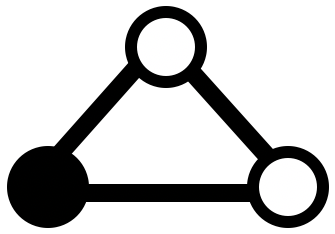} &
\includegraphics[width=5mm]{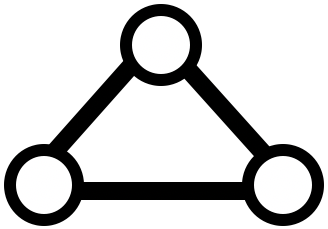} \\

\quad\quad\quad\quad\quad&
$\vp = \mathbf{\big[}\,$ &
0.608 & 0.003 & 0.001 & 0.388 &
$\,\big]$ &
$\quad\quad\quad\quad $
\\
\end{tabular}
\smallskip
\end{center}\noindent
Notably, we observe that $60.86\%$ and $38.79\%$ of the 24,815 triangles are formed among users with the same political leanings.
This implies that three users with the same political leanings are more likely to retweet each other than with users of different political leanings.
These results indicate the presence of homophily~\cite{mcpherson2001homophily} as users tend to retweet similar others.
Furthermore, these homogeneous typed triangles 
(\protect\includegraphics[width=4mm]{fig19.png}, 
\protect\includegraphics[width=4mm]{fig22.png}) 
account for $99.65\%$ of the 24,815 triangles.
Intuitively, this implies that the network consists of two tightly-knit communities of users of the same political leanings. The two communities are sparsely connected.
Typed triangles obviously contain significantly more information than untyped triangles.
This includes not only information about the local properties but also about the global structure of the network as shown above.
Obviously, untyped graphlets are unable to provide such insights as they do not encode the types, attribute values, or class labels associated with a graphlet.
They only reveal the structural information independent of any important external information associated with the node.

We also investigated typed 4-clique graphlets.
Strikingly, only 4 of the 5 typed 4-clique graphlets that arise from $2$ types actually occur in the graph.
In particular, the typed 4-clique graphlet with 2 right users and 2 left users does not even appear in the graph.
This typed graphlet might indicate collusion between individuals from different political parties or some other extremely rare anomalous activity.
The other typed 4-cliques that are extremely rare are the typed 4-clique graphlet with 3 right (left) users and a single left (right) user.

\begin{figure*}[h!]
\centering
\subfigure{\includegraphics[width=0.32\linewidth]{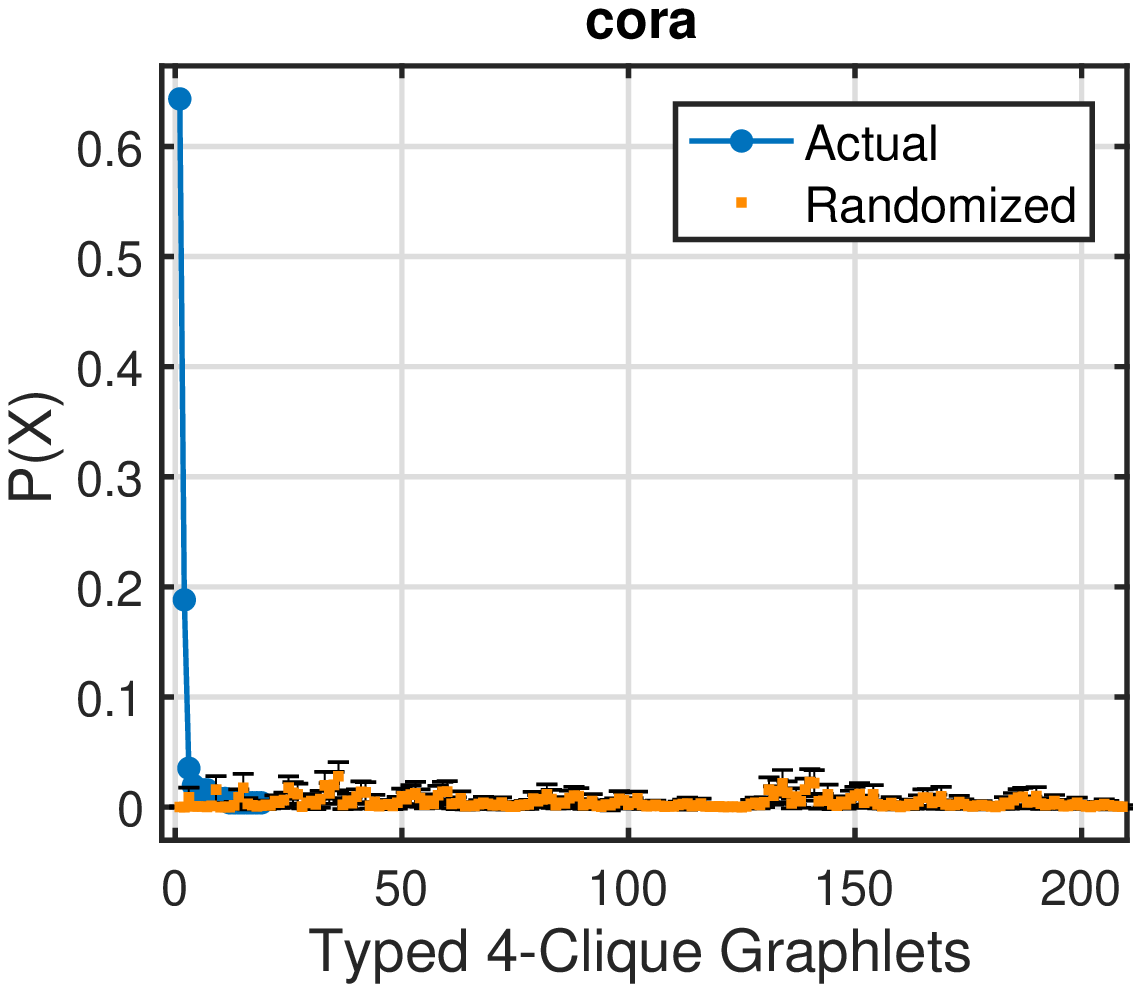}}
\hfill
\subfigure{\includegraphics[width=0.32\linewidth]{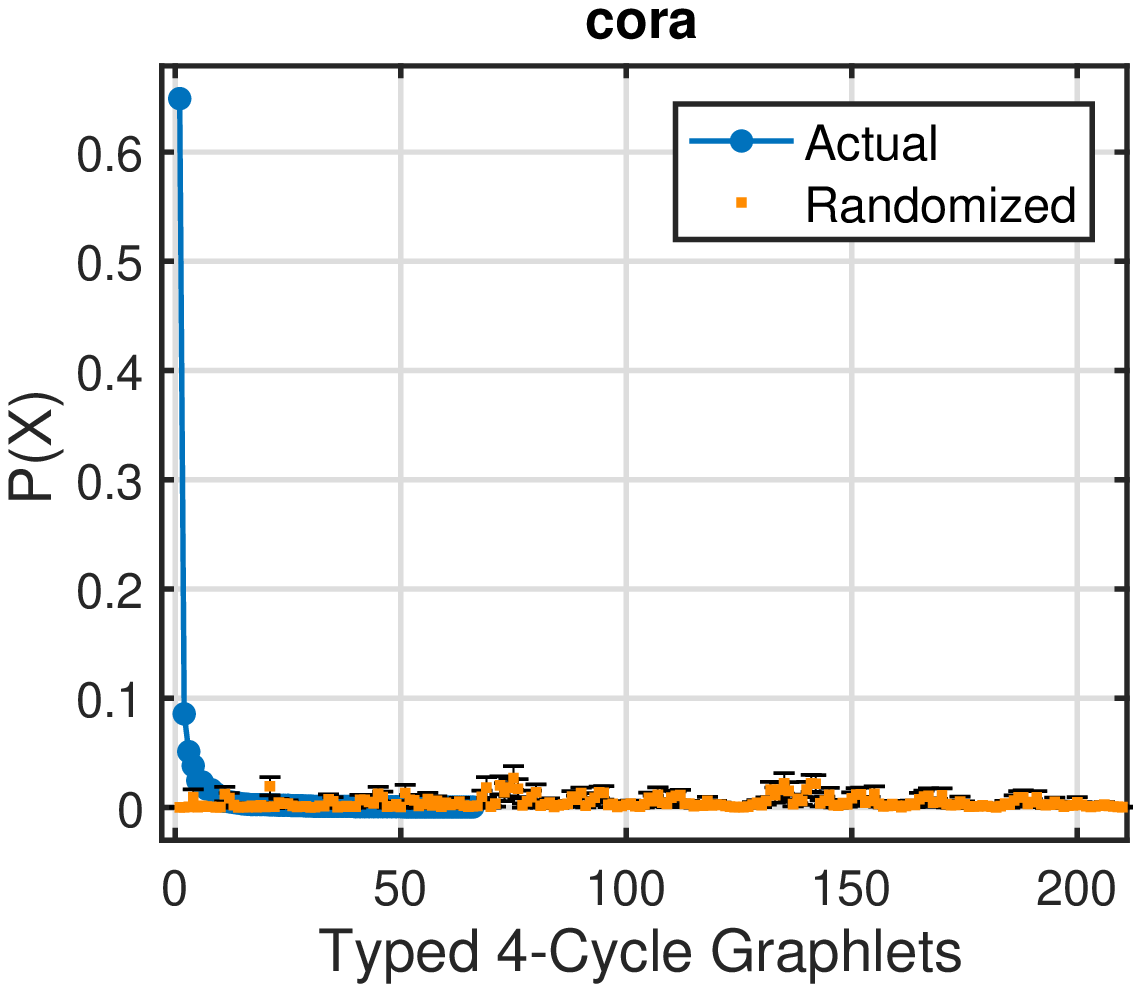}}
\hfill
\subfigure{\includegraphics[width=0.32\linewidth]{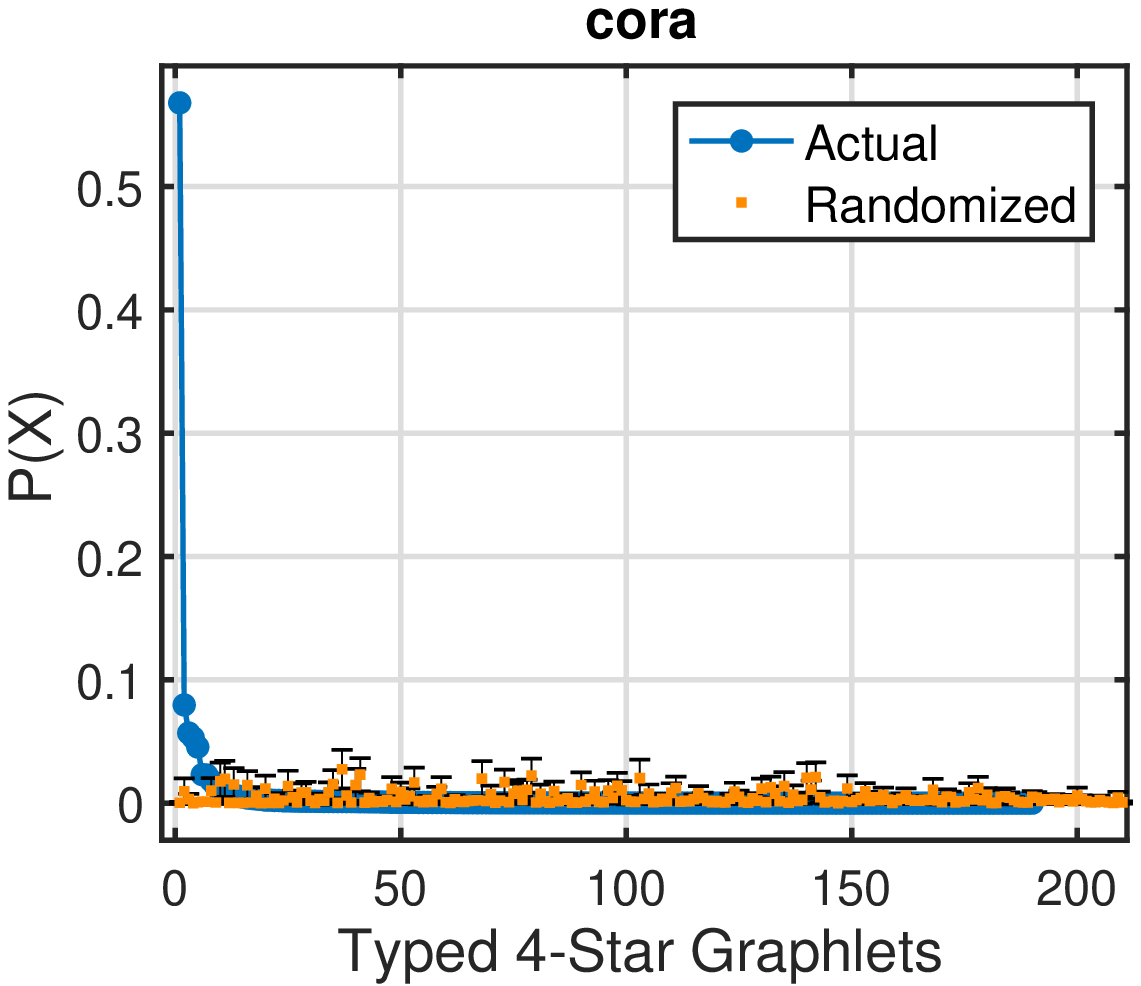}}

\subfigure{\includegraphics[width=0.32\linewidth]{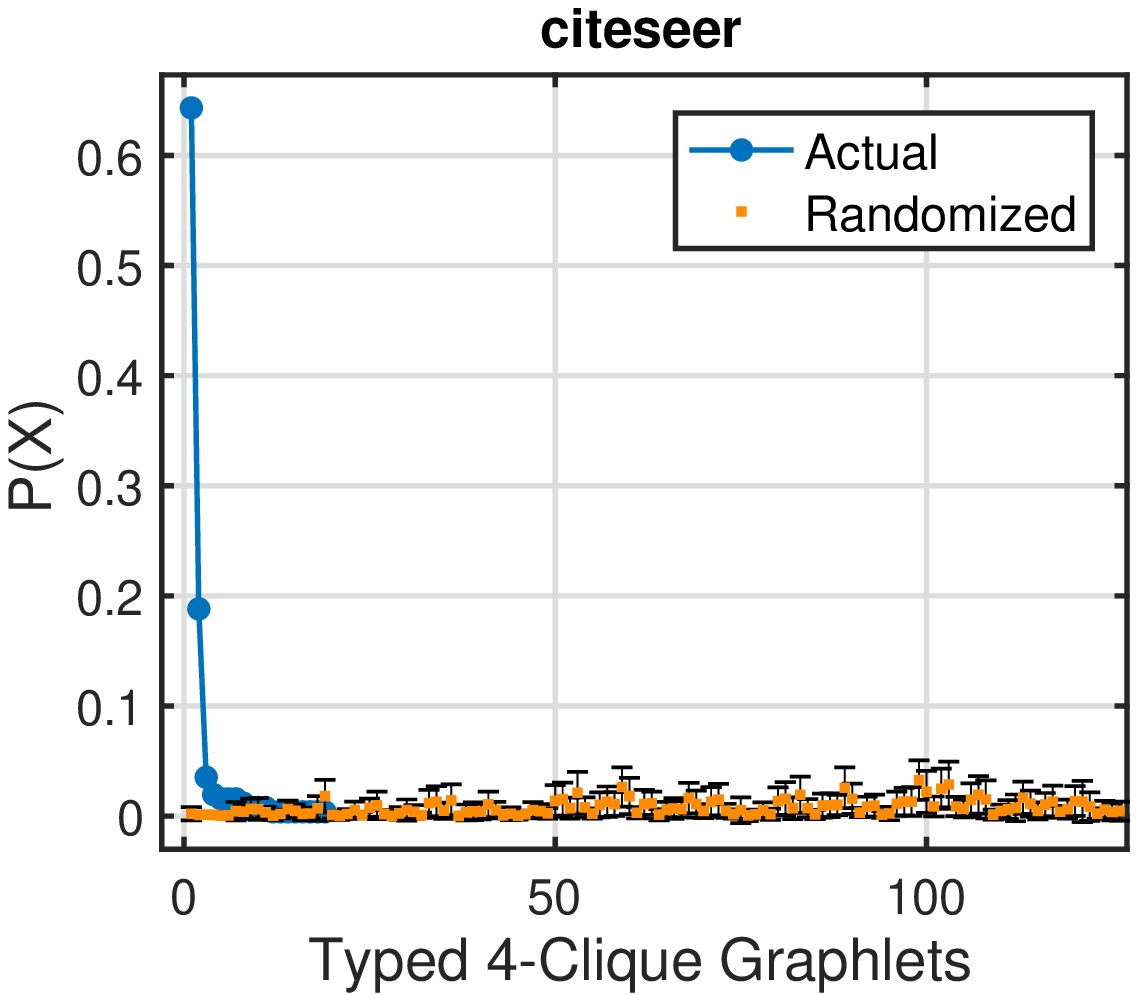}}
\hfill
\subfigure{\includegraphics[width=0.32\linewidth]{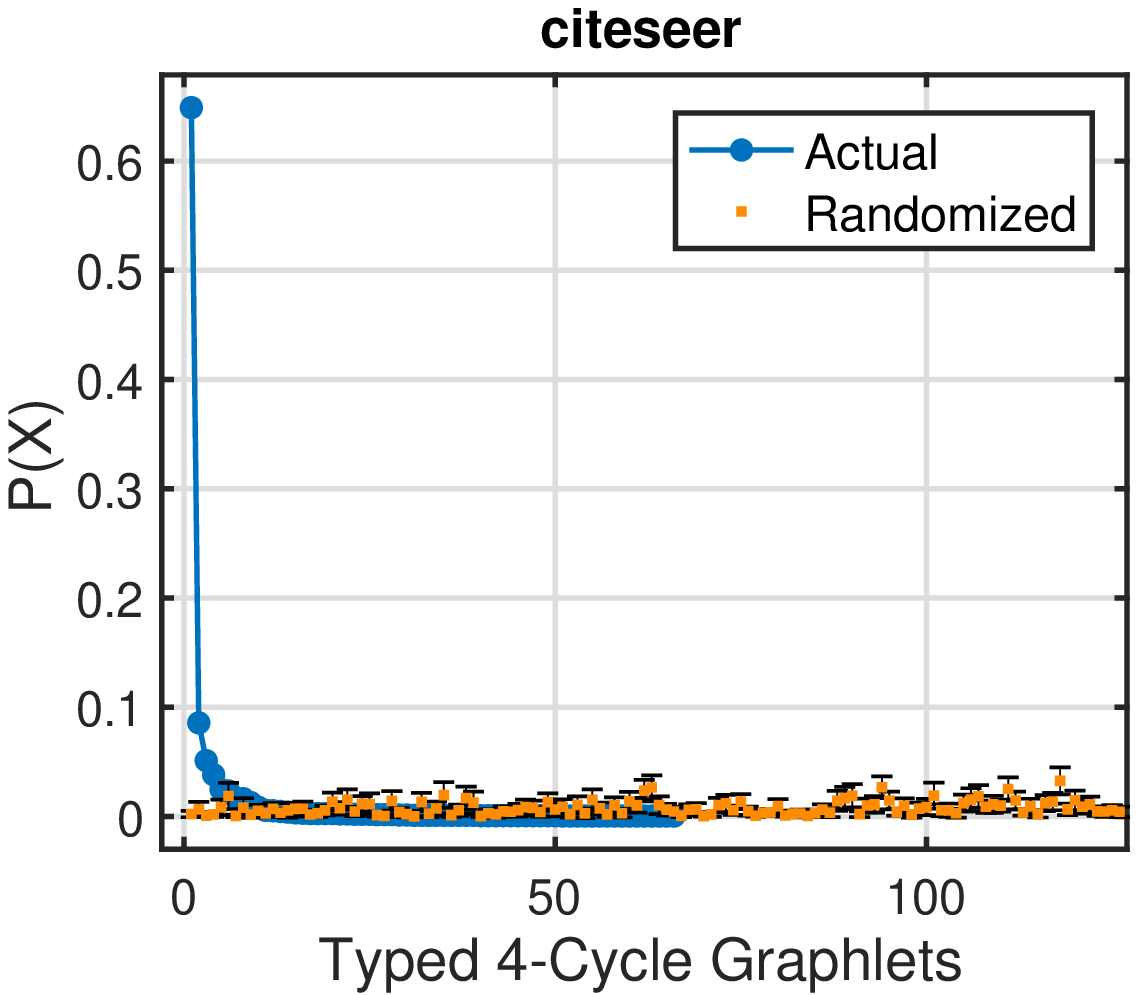}}
\hfill
\subfigure{\includegraphics[width=0.32\linewidth]{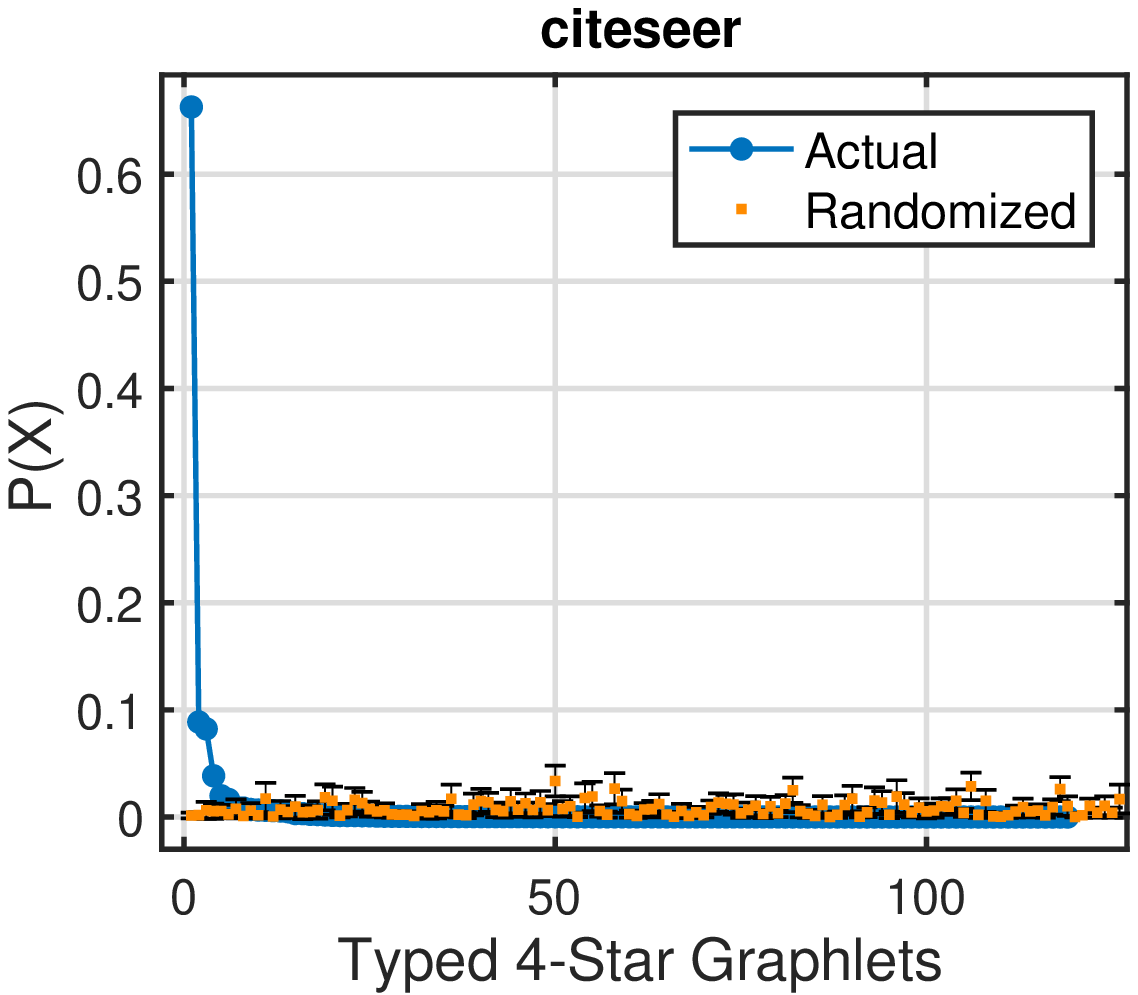}}

\vspace{-3mm}
\caption{
Comparing the actual \emph{typed} 4-clique, 4-cycle and 4-star graphlet distributions to the randomized typed distributions.
We compute 100 random permutations of the node types and run the approach on each permutation then average the resulting counts to obtain the mean randomized typed graphlet distribution.
There are three key findings.
First, we observe a significant difference between the actual and randomized typed graphlet distributions.
Second, many of the typed graphlets that occur when the types are randomized, do not occur in the actual typed graphlet distribution.
Third, we find the typed graphlet distribution to be extremely skewed as a few typed graphlets occur very frequently while the vast majority have very few occurrences (and many typed graphlets are even forbidden, in the sense that they do not occur at all in the graph).
}
\label{fig:typed-4-node-prob-dist-cora-citeseer}
\end{figure*}

\subsubsection{Cora citation network} \label{sec:exploratory-analysis-cora}
The Cora citation network consists of 2708 scientific publications classified into one of seven types (class labels) that indicate the paper topic.
The citation network consists of 5429 links.
Using the proposed heterogeneous graphlets, we find 129 typed 3-node graphlets among the 168 possible typed 3-node graphlets that could occur.
Notably, we observe the most frequent typed triangle graphlets are of a single type.
Indeed, the first 7 typed triangle graphlets with largest frequency in Figure~\ref{fig:typed-tri-prob-dist-cora-citeseer} are of a single type and account for 83.86\% of all typed triangle graphlets.
This finding indicates \emph{strong homophily} among nodes with similar types (Figure~\ref{fig:cora-typed-motifs-homogeneous-types}). 
Unlike untyped graphlets, typed graphlets simultaneously capture the labeling and structural properties that lie at the heart of homophily~\cite{mcpherson2001homophily,LaFond2010}.
Therefore, typed graphlets provide a principled foundation for studying homophily in social networks.
In Figure~\ref{fig:typed-tri-prob-dist-cora-citeseer}, we observe a large gap that clearly separates the 7 single-typed triangle graphlets from the other typed triangle graphlets with heterogeneous types.
Furthermore, only 49 out of the 84 possible typed triangle graphlets (Table~\ref{table:typed-graphlets-example}) actually occur in $G$.

In Figure~\ref{fig:typed-4-node-prob-dist-cora-citeseer}, we also investigate a variety of typed 4-node graphlet distributions (from most dense to least dense).
Strikingly, only 19 of the 210 possible typed 4-clique graphlets actually occur in $G$ when node types are randomly shuffled.
In the case of 4-node typed cycles, we observe 66 of the actual 210 possible typed 4-cycle graphlets appear when node types are randomly shuffled.

\subsubsection{Citeseer citation network}
The citeseer citation network consists of papers with citation links between them. 
Each paper is associated with one of six types representing the paper topic (\eg, ML).
The 5 most frequently occurring typed triangle graphlets are those with a single type (Figure~\ref{fig:typed-tri-prob-dist-cora-citeseer}).
Overall, these typed triangle graphlets account for $77.02\%$ of all typed triangle graphlets that occur in $G$.
This finding indicates \emph{strong homophily} among nodes.
Furthermore, since there are six types corresponding to paper topics, there are 56 potential typed triangle graphlets.
However, among the 56 possible typed triangle graphlets of six types, only 40 actually appear in $G$ as shown in Figure~\ref{fig:typed-tri-prob-dist-cora-citeseer}.
The others are forbidden typed triangle graphlets.
In addition, Figure~\ref{fig:typed-4-node-prob-dist-cora-citeseer} compares the actual typed 4-clique, 4-cycle, and 4-star distributions to the randomized distributions.
Notably, while 126 different typed 4-cliques appear when node types are randomly shuffled, only 19 distinct typed 4-cliques (with different type configurations) actually appear in $G$.
This finding is consistent with the cora citation network discussed in Section~\ref{sec:exploratory-analysis-cora}.

\subsection{Use Case: Link Prediction} \label{sec:exp-link-pred}
This section quantitatively demonstrates the effectiveness of typed graphlets for link prediction.
Given a partially observed graph $G$, 
the link prediction task is to predict the missing edges.
This general problem has applications in recommendation systems, \eg, recommending movies to users (movielens) or suggesting potential friends (yahoo), among other important applications.

{
\algblockdefx[parallel]{parfor}{endpar}[1][]{$\textbf{parallel for}$ #1 $\textbf{do}$}{$\textbf{end parallel}$}
\algrenewcommand{\alglinenumber}[1]{\!\!\fontsize{8}{8.5}\selectfont#1\;}
\begin{figure}[h!]
\begin{center}
\begin{algorithm}[H]
\caption{\; 
Higher-Order \emph{Typed Graphlet} Node Embeddings
}
\label{alg:higher-order-typed-graphlet-embeddings}
\begin{spacing}{1.35}
\small
\begin{algorithmic}[1]
\smallskip
\Require a graph $G$, typed graphlet $H$, embedding dimension $D$
\vspace{-0.1mm}
\Ensure Higher-order node embedding matrix $\mZ \in \RR^{N \times D}$ for $H$
\smallskip

\State $\!\big(\mW_{G^H}\big)_{ij} \!\leftarrow \# \text{ instances of H containing } i \text{ and } j, \; \forall (i,j) \in E$

\State $\mD_{G^H} \leftarrow \!\text{typed-graphlet degree matrix} \big(\mD_{G^H}\big)_{ii} \!= \!\sum_{j} \!\big(\mW_{G^H}\big)_{ij}$

\State $\vu_1, \vu_2, \ldots, \vu_D \leftarrow $ eigenvectors of $D$ smallest eigenvalues of $\mL_{G^{H}} = \eye - \mD_{G^H}^{-1/2}\mW_{G^H}\mD_{G^H}^{-1/2}$

\State $Z_{ij} \leftarrow U_{ij} \Big/ \sqrt{\sum_{j=1}^{D} U_{ij}^2}$

\State {\bf return} $\mZ = \big[\, \vz_1\;\,\, \vz_2\;\, \cdots \;\, \vz_n \,\big]^T \!\in \RR^{N \times D}$
\smallskip
\end{algorithmic}
\end{spacing}
\vspace{-1.mm}
\end{algorithm}
\end{center}
\vspace{-2.mm}
\end{figure}
}

\subsubsection{Higher-Order Typed Graphlet Embedding}
\label{sec:embeddings}
Algorithm~\ref{alg:higher-order-typed-graphlet-embeddings} summarizes the method for deriving \emph{higher-order typed graphlet node embeddings}.
In particular, given a typed-graphlet $H$ of interest,
Algorithm~\ref{alg:higher-order-typed-graphlet-embeddings} outputs a matrix $\mZ$ of node embeddings.
For graphs with many connected components, Algorithm~\ref{alg:higher-order-typed-graphlet-embeddings} is called for each connected component of the typed graphlet graph $G^{H}$ and the resulting embeddings are stored in the appropriate locations in the overall embedding matrix $\mZ$.

\subsubsection{Experimental Setup}
We evaluate the higher-order typed graphlet node embedding approach that explicitly leverages typed graphlets (Algorithm~\ref{alg:higher-order-typed-graphlet-embeddings}) against the following methods: DeepWalk (DW)~\cite{deepwalk}, LINE~\cite{line}, GraRep~\cite{grarep}, spectral embedding (untyped edge graphlet)~\cite{Long06spec},
and spectral embedding using untyped-graphlets.
All methods output ($D$=128)-dimensional node embeddings.
For DeepWalk (DW)~\cite{deepwalk}, we perform 10 random walks per node of length 80 as mentioned in~\cite{node2vec}.
For LINE~\cite{line}, we use 2nd-order proximity and perform 60 million samples.
For GraRep (GR)~\cite{grarep}, we use $(k=2)$-steps.
In contrast, the spectral embedding methods do not have any hyperparameters besides $D$ which is fixed for all methods.

\begin{table}[t!]
\centering
\setlength{\tabcolsep}{4pt}
\renewcommand{\arraystretch}{1.1} 
\caption{
Link prediction edge types and semantics.
The edge type predicted by the models is bold.
}
\label{table:link-pred-network-data}
\vspace{-3mm}
\small
\begin{tabularx}{0.7\linewidth}{l c c l @{}}
\toprule
\textbf{Graph}  &   $|\mathcal{T}_V|$ & $|\mathcal{T}_E|$ & \textbf{Heterogeneous Edge Types} 
\\
\midrule

\textsf{movielens} & 3 & 3 & \textbf{user-by-movie}, user-by-tag,  
\\ 

&&& tag-by-movie 
\\

\textsf{dbpedia} &  4 & 3 & \textbf{person-by-work} (produced work),  
\\

&&& person-has-occupation, 
\\

&&& work-by-genre (work-associated-genre) 
\\

\textsf{yahoo-msg} & 2 & 2 & \textbf{user-by-user} (communicated with), 
\\ 

&&& user-by-location (communication location) 
\\

\bottomrule
\end{tabularx}
\end{table}

\begin{table}[b!]
\centering
\renewcommand{\arraystretch}{1.1}
\setlength{\tabcolsep}{8pt}
\caption{
Link prediction results. 
These results demonstrate the effectiveness of typed graphlets for prediction.
}
\label{table:link-pred-results}
\vspace{-3mm}
\small
\begin{tabularx}{1.0\linewidth}{@{}ll ccccc c@{}H HHHH@{}}
\toprule

&& &      &     &   &  \textbf{Untyped} &  \textbf{Typed} &  \\ 

&& \textbf{DeepWalk}   &   \textbf{LINE}   &   \textbf{GraRep}   &  \textbf{Spectral} &  
\textbf{Graphlets} &  \textbf{Graphlets} & \\ 
\midrule

\multirow{4}{*}{\rotatebox{0}{\textbf{\sf movielens}}}  
&   $\mathbf{F_1}$          &     0.8544   &   0.8638       &   0.8550        &   0.8774   &   0.8728   &   \textbf{0.9409}   &   \\
&   \textbf{Prec.}          &     0.9136   &   0.8785       &   0.9235        &   0.9409   &   0.9454   &   \textbf{0.9747}   &   \\
&   \textbf{Recall}         &     0.7844   &   0.8444       &   0.7760        &   0.8066   &   0.7930   &   \textbf{0.9055}   &   \\
&   \textbf{AUC}            &     0.9406   &   0.9313       &   0.9310        &   0.9515   &   0.9564   &   \textbf{0.9900}   &   \\
\midrule

\multirow{4}{*}{\rotatebox{00}{\textbf{\sf dbpedia}}} 
&   $\mathbf{F_1}$          &     0.8414   &   0.7242   &   0.7136       &   0.8366   &   0.8768   &   \textbf{0.9640}   &   \\
&   \textbf{Prec.}          &     0.8215   &   0.7754   &   0.7060       &   0.7703   &   0.8209   &   \textbf{0.9555}   &   \\
&   \textbf{Recall}         &     0.8726   &   0.6375   &   0.7323       &   0.9669   &   0.9665   &   \textbf{0.9733}   &   \\
&   \textbf{AUC}            &     0.8852   &   0.8122   &   0.7375       &   0.9222   &   0.9414   &   \textbf{0.9894}   &   \\
\midrule

\multirow{4}{*}{\rotatebox{0}{\textbf{\bf \sf yahoo}}}
&   $\mathbf{F_1}$          &     0.6927   &   0.6269   &   0.6949   &   0.9140    &   0.8410  &   \textbf{0.9303}   &    \\
&   \textbf{Prec.}          &     0.7391   &   0.6360  &   0.7263   &   0.9346   &   0.8226   &   \textbf{0.9432}    &    \\
&   \textbf{Recall}         &     0.5956   &   0.5933   &   0.6300   &   0.8904   &   0.8699   &   \textbf{0.9158}   &    \\
&   \textbf{AUC}            &     0.7715   &   0.6745   &   0.7551   &   0.9709   &   0.9272   &   \textbf{0.9827}   &    \\

\bottomrule
\end{tabularx}
\vspace{2mm}
\end{table}

\subsubsection{Results}
We generate a labeled dataset of positive and negative edges.
Positive edge examples are obtained by removing $50\%$ of edges uniformly at random, whereas negative examples are generated by randomly sampling an equal number of node pairs $(i,j) \not\in E$. 
For each method, we learn node embeddings using the remaining graph.
Given embedding vectors $\vz_i$ and $\vz_j$ for node $i$ and $j$, 
we derive a $D$-dimensional edge embedding vector
$\vz_{ij} = \sigma(\vz_i, \vz_j)$
where $\sigma$ is defined as one of the following \emph{edge embedding functions}:
\begin{align}\label{eq:edge-embedding-ops}\nonumber
\sigma \,\in\, 
\Bigg\lbrace
\frac{\vz_i + \vz_j}{2},\;
\vz_i \odot \vz_j,\;
\abs{\vz_i - \vz_j},\;
(\vz_i - \vz_j)^{\circ 2},\;
\max(\vz_i, \vz_j),\;
\vz_i + \vz_j
\Bigg\rbrace
\end{align}\noindent
Note $\vz_i \odot \vz_j$ is the element-wise product, 
$\vz^{\circ 2}$ is the Hadamard power, 
and $\max(\vz_i, \vz_j)$ is the element-wise max.
Using the edge embeddings, we then learn a logistic regression model to predict if an edge in the test set exists in $E$ or not.
Experiments are repeated for 10 random seed initializations and the average performance is reported.
All methods are evaluated against four different evaluation metrics including $F_1$, Precision, Recall, and AUC.

Table~\ref{table:link-pred-network-data} summarizes the heterogeneous network data and the type/label of the edge predicted by the models.
The results are provided in Table~\ref{table:link-pred-results}.
We report the best result among the different edge embedding functions and untyped/typed graphlets.
In Table~\ref{table:link-pred-results}, the typed graphlet approach is shown to outperform all other methods across \emph{all} four evaluation metrics.
In all cases, the approach that leverages typed graphlets 
outperforms the other methods (Table~\ref{table:link-pred-results}) with an overall mean gain (improvement) in $F_1$ of 18.7\% (and up to 48.4\% improvement) across all graph data.
In terms of AUC, the typed graphlet approach achieves a mean gain of 14.4\% (and up to 45.7\% improvement) over all methods.
Furthermore, we posit that an approach similar to the one proposed in~\cite{Rossi2018a} could be used 
to achieve even better predictive performance by leveraging multiple typed graphlets simultaneously.

\section{Conclusion} \label{sec:conc}
\noindent
In this work, we introduced the notion of typed graphlet that generalizes the notion of graphlet to heterogeneous networks.
We proposed a fast, parallel, and space-efficient framework for counting typed graphlets.
The proposed typed graphlet algorithms count only a few typed graphlets and derives the others in $o(1)$ constant time using new non-trivial combinatorial relationships that involve counts of lower-order typed graphlets.
Thus, the proposed approach avoids explicit enumeration of any nodes involved in those typed graphlets.
For every edge, we count a few typed graphlets and obtain the exact counts of the remaining ones in $o(1)$ constant time.
Theoretically, the worst-case time complexity of the proposed approach is shown to match the best untyped graphlet algorithm.
Since this is the first investigation into typed graphlets, there are no existing methods for comparison.
However, we compared our approach to colored graphlet counting methods that solve a strictly simpler problem.
Empirically, our approach is shown to outperform the state-of-the-art in terms of runtime, space-efficiency, and scalability as it is able to handle large networks.
While these methods take hours on small graphs with thousands of edges, our typed graphlet counting approach takes only seconds on networks with millions of edges.
Finally, the proposed approach is able to handle \emph{large} general heterogeneous networks while lending itself to an efficient and highly scalable parallel implementation.
The proposed approach gives rise to new opportunities and applications for typed graphlets.
Future work should use the ideas introduced in this paper to extend and derive equations for typed graphlets of 5 nodes and larger.
This is similar to how recent work~\cite{dave2017clog,PinarWWW17} extended the ideas introduced by~\citet{pgd,pgd-kais} to 5-node untyped graphlets.

\balance
\bibliographystyle{ACM-Reference-Format}
\bibliography{paper}

\end{document}